\documentclass{iopjournal}
\usepackage{amsmath,amssymb,amsfonts,amsthm}
\usepackage{bbm}
\usepackage{cases}
\usepackage{mathrsfs}
\usepackage{tikz}
\usepackage{tikz-cd}
\usepackage{empheq}
\usepackage{orcidlink}
\usepackage{accents}
\usepackage{float}
\usepackage[T1]{fontenc}

\newcommand{\Int }{\displaystyle \int}

\newcommand{\SP}{{\cal S}p{\cal S}t}
\newcommand{\be}{\begin{equation}}
\newcommand{\de}{\end{equation}}

\newcommand{\R}{R}
\newcommand{\rgh}{H}
\newcommand{\alg}{\mathscr{D}_{pow}}
\newcommand{\blg}{\mathscr{H}_{pow}}
\newcommand{\V}{V_0}

\newcommand{\LSPO}{OLSP}

\newcommand{\thickbar}[1]{%
  \accentset{\rule{.95em}{.6pt}}{#1}%
}
\newcommand{\para}[1]{\medskip\noindent\textbf{#1.}\;}

\newtheorem{lemma}{Lemma}
\newtheorem{theorem}{Theorem}
\newtheorem{definition}{Definition}
\newtheorem*{definition*}{Definition}
\newtheorem{proposition}{Proposition}
\newtheorem{corollary}{Corollary}
\newtheorem{property}{Property}
\newtheorem*{nnproperty*}{Property}
\newcounter{remark}[section]
\renewcommand{\theremark}{\thesection.\arabic{remark}}

\newenvironment{remark}{%
  \refstepcounter{remark}%
  \par\noindent\textbf{Remark \theremark.}\ \rmfamily
}{\par}

\begin{document}

\articletype{Paper} 

\title{A Measure-Theoretic Approach to Spontaneous Stochasticity}

\author{Wandrille Ruffenach$^1$\orcidlink{0009-0008-8261-786X}, Eric Simonnet$^{2,4,*}$\orcidlink{0000-0003-4636-4910} and Nicolas Valade$^{3,4}$\orcidlink{0000-0002-3786-8933}}

\affil{$^1$LPENSL, ENS de Lyon, CNRS, UMR5672, 69342 Lyon cedex 07, France}

\affil{$^2$Institut de Physique de Nice, Universit\'e C\^ote d'Azur et CNRS, France}

\affil{$^3${DAMTP, Centre for Mathematical Sciences,
University of Cambridge, CB3 0WA United Kingdom}}

\affil{$^4$ Inria, Centre at Universit\'e C\^ote d’Azur
Team Calisto, Sophia Antipolis, France}

\affil{$^*$Author to whom any correspondence should be addressed.}

\email{eric.simonnet@univ-cotedazur.fr}


\begin{abstract}
Spontaneous stochasticity $(\SP$) has a long history, from Richardson's Lagrangian picture of turbulent dispersion to Lorenz's 1969 Eulerian view of finite-time loss of predictability, and was later formulated under this name by Gaw\k{e}dzki and collaborators. Over the last two decades, the work of Mailybaev and collaborators has further revealed the role of spontaneous stochasticity in shell models through a renormalization-group formalism. Whether spontaneous stochasticity occurs in fully developed turbulence remains a major open question.

Except in a few specific classes of systems, however, spontaneous stochasticity has not yet been defined in general mathematical terms. This is the main purpose of the present work. We introduce a measure-theoretic formalism in which $\SP$ is understood as a measure-selection principle. Given an inviscid problem, a well-posed regularization, and an ambient measure, one studies the pushforward of this ambient measure by the regularized flow. Strong $\SP$ occurs when this family converges to a unique non-Dirac probability measure. In this regime, the classical deterministic selection principle breaks down and is replaced by the selection of a probability law. 

Within this framework, we establish several structural results for finite-dimensional systems.
The central one is an attainability theorem: whenever the inviscid problem is nonunique, any probability measure supported on the set of inviscid states can be selected as the limiting law of a suitable regularization. We also identify singular sets in the inviscid dynamics, detected through Dini-type directional growth, as necessary obstructions underlying nonuniqueness. The relation between $\SP$ and sensitivity to initial data is also analyzed carefully, clarifying the role and limitations of turbulence-inspired criteria based on finite-time separation. Finally, we introduce a renormalization-(semi)group viewpoint in which limiting statistics appear as statistical attractors. Explicit finite-dimensional examples show how ambient measures, inviscid singularities, regularization scales, and initial-data sensitivity interact in the emergence of $\SP$.

\end{abstract}
\tableofcontents

\section{Introduction}
Spontaneous stochasticity ($\SP$) is the emergence of genuinely probabilistic behavior in a deterministic system in the singular limit in which the regularization is removed. In the simplest terms, one studies a well-posed regularized problem, lets the regularization parameter vanish, and asks what is selected in the inviscid limit. If a unique limiting state is selected, there is a classical selection principle. If not, then one enters the realm of spontaneous stochasticity.

The term \emph{spontaneous stochasticity} is often used in a much narrower context than the phenomenon itself. In fluid mechanics, it is frequently tied to the Euler--Navier--Stokes problem, or even more specifically to turbulence. But the underlying issue is more general: whenever an inviscid limit loses a classical selection principle, one must ask what, if anything, is selected instead. Historically, Lagrangian spontaneous stochasticity came first, because of its close connection with anomalous dissipation, one of the central mechanisms of turbulence. In contrast, Eulerian spontaneous stochasticity is conceptually closer to a loss of predictability for the full state itself.

From a historical perspective, four landmarks are especially relevant. Richardson's work \cite{Richardson1926} on turbulent pair dispersion already contained the seed of the Lagrangian picture. Gaw\c{e}dzki and collaborators later turned this into a precise mechanism in rough-flow models, showing that spontaneous stochasticity and anomalous dissipation are intimately linked in the Lagrangian setting; see, for example, \cite{Gawedzki98,Chaves2003,Drivas17}. On the Eulerian side, an important shift appears in Lorenz's 1969 paper \cite{Lorenz69}, which is often overshadowed by the much more famous 1963 work on low-dimensional chaos \cite{Lorenz63}. The point of \cite{Lorenz69} is very different: it is not the usual sensitive dependence of classical chaos, but a finite-time predictability barrier in multiscale flows. 
Interestingly, the term Eulerian spontaneous stochasticity appeared much later: it was first introduced by G. Eyink at the APS 2016 conference and was subsequently discussed in detail in \cite{Eyink_Bandak20}.

More recently, Mailybaev and collaborators have developed two complementary lines of work on spontaneous stochasticity. The first concerns shell models, where a renormalization-group (RG) and dynamical-systems formalism organizes the selected inviscid statistics through attractors and invariant measures of the RG dynamics; see, for instance, \cite{Maily2012,Maily2016,AM_Ra23,AM_Rb23,AM_NLN25}. The second concerns the important case of finite-dimensional systems \cite{Drivas21,Drivas24}, where a dynamical-systems viewpoint already captures key aspects of the phenomenon; see also \cite{GHR_SLDP2001},\cite{Eyink_Bandak20} for a complementary perspective based on stochastic differential equations and (singular) large-deviation theory. The present paper extends this finite-dimensional direction to a broader, fully measure-theoretic setting.

Our main point of departure is that $\SP$ should not be treated as a mysterious by-product of turbulence, nor as a vague synonym for complexity or chaos. At heart, it is a selection problem. When no classical selection principle exists in the inviscid limit, the next natural object that can be selected is not a single state, but a probability measure on the set of inviscid states. In other words, a Dirac mass corresponds to ordinary selection, while spontaneous stochasticity begins when the selected measure is non-Dirac. Moreover, there is an even more singular possibility: the inviscid limit may fail to select a unique measure, and only subsequential limiting measures may exist. In that regime, the issue is no longer the selection of one state, nor even of one law, but of a whole collection of possible laws. As we shall see, this weak form of $\SP$ may also indicate that the chosen ambient measure is not the most natural one for probing the singular limit.

This leads to the measure-theoretic viewpoint adopted throughout the paper. The phenomenon involves three basic ingredients:
\begin{itemize}
    \item[(i)] an inviscid system, typically ill-posed;
    \item[(ii)] a regularization that restores well-posedness;
    \item[(iii)] an ambient measure, which determines how statistics are extracted in the singular limit.
\end{itemize}
Within this framework, spontaneous stochasticity becomes a \emph{measure-selection principle}. Strong $\SP$ means that the inviscid limit selects a unique non-Dirac probability measure. Weak $\SP$ means that only subsequential limiting measures exist. This language has two advantages. First, it is flexible enough to encompass the definitions used in physics. Second, it makes transparent the relation with the classical lack of a selection principle used in mathematics.

This brings us to a second main message of the paper. In mathematics, the underlying problem is much older than the term $\SP$ itself: one studies ill-posedness, nonuniqueness, and the failure of classical selection. In that literature, the stochastic aspect is usually not emphasized as such. Our claim is that these two viewpoints should be reconciled. If the inviscid system is well posed, then a classical selection principle is available and there is no room for $\SP$. Thus nonuniqueness of the inviscid problem is a necessary ingredient. But once nonuniqueness is present, the right question is no longer which single solution should be chosen; it is which probability measure, if any, is selected by the regularization procedure. From this perspective, spontaneous stochasticity is not an additional mystery layered on top of ill-posedness. It is the natural language for describing what remains selectable once ordinary deterministic selection has failed.

The first part \ref{part_1} of the paper develops this point of view in a finite-dimensional setting. We formulate a definition of $\SP$ for the Cauchy problem, compare it with the lack of a selection principle and its observable version, and show that, modulo the purely Dirac case, these notions essentially coincide under natural compactness assumptions. This yields a useful trichotomy between strong $\SP$, weak $\SP$, and $\delta$-$LSP$, namely lack of selection with a Dirac limiting measure. We then derive a necessary condition for nonuniqueness in terms of singular sets detected by Dini-type directional derivatives. In this sense, the emergence of $\SP$ is tied to a precise geometric obstruction in the inviscid dynamics.

A central result of Part \ref{part_1} is that, once nonuniqueness occurs, the set of attainable limiting measures $\mathscr{M}$ is maximal: it coincides with the full space of probability measures on the compact set of inviscid states, which we denote by ${\rm Sol}$ for the Cauchy problem. In our notation,
$$
{\mathscr M} =  {\cal P}({\rm Sol}).
$$
This equality is one of the main messages of the paper. It shows that the essential question is not to understand $\SP$ for a single distinguished regularization, but rather to understand how the space of regularizations organizes the possible selection principles. In particular, every probability measure on the inviscid solution set may arise from a suitable regularization. For this reason the notion of universality is central: one should not seek a ``good'' regularization in isolation, but classes of regularizations that produce the same selected statistics.

Sensitivity to initial conditions is often invoked in discussions of spontaneous stochasticity.  More precisely, ill‑posedness of the inviscid problem means that solutions fail to depend continuously on initial data.  A crucial subtlety is that this failure need not produce observable randomness: a given regularization can impose uniform behavior with respect to initial data and thereby suppress sensitivity in the regularized dynamics.  To address this issue we extend our notion of spontaneous stochasticity to encompass a broad class of admissible regularizations.  The extended definition coincides with the one used above, and our main attainability theorem
asserts that every probability measure on the inviscid solution set can be realized as the limit law produced by a suitable vanishing‑parameter regularization (indeed, by an appropriate one‑parameter family). 

Remarkably, most of the mechanisms discussed above already arise in the simplest finite-dimensional setting, namely non-Lipschitz systems with a single degree of freedom. We analyze one representative model in detail to make the abstract theory completely explicit.

The second part \ref{part_2} of the paper develops a complementary RG viewpoint. There, the regularization parameter is treated as an evolution variable, and regularization curves are organized as orbits of a semigroup. This does not replace the measure-theoretic framework of Part \ref{part_1}; rather, it singles out a dynamically distinguished subclass of regularizations. Once a semigroup structure is imposed, the selected statistics acquire additional structure: they are tied to invariant and ergodic measures, and may be interpreted as statistical attractors. This is precisely where the work connects with the dynamical-systems and RG ideas developed by Mailybaev and collaborators. The gain is not only conceptual. In examples with isolated singularities, the RG formalism turns the selection problem into an exit-time problem near the singularity, followed by a dynamical organization of the resulting exit data.

\bigskip 
{
The paper is organized in two parts. Part~\ref{part_1} develops the measure-theoretic formulation of $\SP$, together with its consequences for nonuniqueness, singularities, and sensitivity to initial conditions. Part~\ref{part_2}, entitled \emph{Spontaneous stochasticity with a semigroup structure}, studies the same framework from a dynamical-systems viewpoint.

\medskip
\noindent
In Part~\ref{part_1}, Section~\ref{Defs} introduces $\SP$ as a general measure-selection principle and distinguishes its strong and weak forms. Section~\ref{BasicSensitivity} then records the first consequences of this definition and relates it to nearby notions. Its first subsection, Section~\ref{Basic}, gives the basic properties of the definition and clarifies its relation with classical non-selection, observable non-selection, and the Dirac-limit obstruction. Its second subsection, Section~\ref{sensitivityCI}, examines the relation between $\SP$, sensitivity to initial conditions, and turbulence-inspired criteria. Section~\ref{allprob} proves the main attainability result: once nonuniqueness is present, every probability measure on the inviscid solution set can be obtained as the limiting law of a suitable regularization. This leads, in Section~\ref{univc}, to a set-theoretic notion of universality classes, defined by identical limiting statistics. Since nonuniqueness of the inviscid problem is a prerequisite for $\SP$, it remains to understand when such nonuniqueness can occur. Section~\ref{CNforSP} addresses this question by giving a necessary condition in terms of singular sets detected by Dini-type directional growth. Finally, Section~\ref{exAmbrosio} illustrates the preceding results on a simple one-dimensional ODE model.

\medskip
\noindent
Part~\ref{part_2} develops $\SP$ with a semigroup structure. We introduce a formalism in which limiting statistics become statistical attractors of a semigroup acting on the regularization parameter. In this setting, universality classes acquire a dynamical-systems interpretation as basins of attraction of these statistical attractors; see Theorem~\ref{RGattractors} and Section~\ref{par1_4_3_2}. We then analyze this framework in the important case where the inviscid system has an isolated singularity at the origin. In that setting, the semigroup formulation naturally leads to an exit-time description of the selection mechanism, which is illustrated through simple concrete examples in Section~\ref{RGreg}.}

\part{\Large Formalism and general results}\label{part_1}
\section{Spontaneous stochasticity: definitions}
\label{Defs}

The main goal of this section is to formulate a general mathematical definition of \emph{spontaneous stochasticity} ($\SP$), while remaining consistent with the interpretation commonly used in physics. This is a delicate task, since the few rigorous definitions currently available in the physics literature apply only to rather specific classes of models; see, for instance, \cite{Drivas24}. In its classical interpretation, $\SP$ is associated with finite-time ``Richardson-like'' separation of trajectories in the inviscid limit, that is, with the blow-up of finite-time Lyapunov exponents. It is most often formulated in a Lagrangian framework for fluid particles, but it also admits a natural Eulerian interpretation in state space. More recently, a broader perspective has emerged through renormalization-group (RG) ideas; see \cite{AM_Ra23,AM_Rb23,AM_NLN25,AM_PRF26}. This viewpoint will be developed in Part~\ref{part_2}.

The notion of $\SP$ introduced here is essentially a \emph{measure-selection principle}, characterized by the selection of a non-Dirac probability measure in the inviscid limit.
Its formulation involves three basic ingredients:
\begin{itemize}
    \item[(i)] an inviscid problem, typically ill posed in the singular limit;
    \item[(ii)] a regularization of the ill-posed problem that restores well-posedness;
    \item[(iii)] an ambient measure, used to define the corresponding statistics.
\end{itemize}

To fix ideas, we consider finite-dimensional, deterministic, and autonomous problems. Let $H$ be a finite-dimensional Hilbert space, equipped with inner product $\langle x, y \rangle$ and corresponding norm $\|x\| = \langle x, x \rangle^{1/2}$. 
The same framework extends to infinite-dimensional systems, but doing so requires substantially more care, primarily because compactness is no longer available. This case is treated in a companion paper for a passive scalar transport PDE \cite{RSV2025}.

\paragraph{Inviscid system and its regularization.}
Let $f_0 \in C_b(H; H)$. The inviscid system is defined by
\begin{equation} \label{P_0}\tag{${\cal P}_0$}
\dot{x} = f_0(x), \quad
x(0) = x_0 \in H, \quad t \in [0, T].
\end{equation}
We study the inviscid limit $\epsilon\to0$ for a family of regularized problems
\begin{equation} \label{P_eps} \tag{${\cal P}_\epsilon$}
\dot{x} = f(x, \epsilon), \quad
x(0) = x_0 \in H, \quad t \in [0, T],
\end{equation}
where $f$ belongs to the class
\begin{equation}\label{H0}
\V = \left\{ 
f \in C_b(H\times \mathbb{R}^+; H) ~:~
\lim_{\epsilon \to 0} \|f(\cdot, \epsilon) - f_0(\cdot)\|_\infty = 0,~ 
\text{and}~ f(\cdot,\epsilon)~\in \operatorname{Lip}(H;H)~ \forall \epsilon > 0 
\right\}.
\end{equation}
{The global Lipschitz condition may be replaced by Hadamard well-posedness for every $\epsilon > 0$, namely existence, uniqueness, and continuous dependence on initial data.
The global Lipschitz assumption is retained as a technical closure condition needed in the proof of Theorem~\ref{M0M}.}
Since
$f_0\in C_b(H;H)$, Peano's theorem gives local-in-time existence of solutions
to $({\cal P}_0)$. For each $\epsilon>0$, the Lipschitz continuity of
$f(\cdot,\epsilon)$ gives uniqueness and continuous dependence on initial data,
and hence Hadamard well-posedness of $({\cal P}_\epsilon)$ on the prescribed
interval $[0,T]$.

We use the following flow map notation:
$$
\phi_t^\epsilon(x_0);~\mbox{solution of \eqref{P_eps} at time $t$ with
initial condition $x_0$}
$$
For brevity, we sometimes write $x^\epsilon$ for the solution of \eqref{P_eps} over the whole interval $[0, T]$. 
Let ${\mathscr S}_0 \subset H$ denote the set of solutions to 
\eqref{P_0} at fixed time $t$ and initial condition $x_0$:
\be \label{S0}
{\mathscr S}_0 := \{ \phi_t^0(x_0) : \mbox{solution at time $t$ and initial condition $x_0$ of \eqref{P_0}} \}.
\de
It is a classical result that $\mathscr S_0$ is compact in $H$. Moreover, by Kneser’s theorem, $\mathscr S_0$ is also connected. When $H$ is infinite-dimensional, as in PDE settings, 
compactness of $\mathscr S_0$ is no longer guaranteed.

Let $\mathscr{M}_0$ be the space of Borel probability measures with support on $\mathscr{S}_0$, endowed with the weak topology:
\be \label{M0}
\mathscr{M}_0 := {\cal P}({\mathscr S}_0).
\de 
Convergence $\mu_\epsilon \rightharpoonup \mu$ means
$
\langle \mu_\epsilon, F \rangle \to \langle \mu, F \rangle \quad \text{for all } F \in C_b(H; \mathbb{R}),
$
where $\langle \mu, F \rangle := \int_H F(x)\, \mu(dx)$. In particular, for every Borel set $B \subset H$, we have $\mu(B) = \int_B \mu(dx)$ and $\mu({\mathscr S}_0) = 1$.

\paragraph{Ambient measure.}
To fix notation, we consider a family of probability measures
\be
\mathbb{P}_\epsilon \in \mathcal{P}(\mathbb{R}^+) \quad\text{such that}\quad
\mathbb{P}_\epsilon \rightharpoonup \delta_0 .
\de
We refer to $\mathbb{P}_\epsilon$ as the \emph{ambient} measure, and we impose no further assumptions on it.
However, for reasons which will be made clear in the following, the case of interest is when $\mathbb{P}_\epsilon$ is absolutely continuous (a.c.) w.r.t. Lebesgue measure. The typical example is the normalized Lebesgue measure:
\be \label{Leb}
\mathrm{Leb}_\epsilon(B) := \frac{1}{\epsilon} \mathrm{Leb}(B \cap [0, \epsilon]), \quad \text{for every Borel set } B \subset \mathbb{R}.
\de 
Therefore, $\langle {\rm Leb}_\epsilon,F \rangle$ just mean the Ces\`aro average $\frac{1}{\epsilon}\int_0^\epsilon F(s)~ds$.

\paragraph{Spontaneous stochasticity ($\SP)$.}
For convenience, we first introduce the notion of a \emph{regularization curve}.
\begin{definition}[Regularization curve]\label{gammadef}
Let $t > 0$ and $x_0 \in H$ fixed. 
A \emph{regularization curve} associated with the system $({\cal P}_\epsilon)_{\epsilon \geq 0}$ is the function  defined by
\be \label{gamma}
\gamma: (0,\infty) \to H,~\epsilon \mapsto \phi_t^\epsilon(x_0).
\de
Since $f \in \V$, the curve $\gamma$ is moreover bounded in $\epsilon$.
\end{definition}
 This definition is central, as it reduces the phenomenon of spontaneous stochasticity to the study of such curves in the limit $\epsilon \to 0$.
In Part~\ref{part_2}, we endow these curves with an additional semigroup property, allowing them to be interpreted as genuine trajectories of a dynamical system.
\bigskip 

We consider the pushforward of $\mathbb{P}_\epsilon$ by the function $\gamma$ in~\eqref{gamma}, which is denoted $\gamma_\#\mathbb{P}_\epsilon = \int \delta_{\gamma(s)} \mathbb{P}_\epsilon(ds)$. For all test functions $F \in C_b(H;\mathbb{R})$, this is
$$
\langle \gamma_\# \mathbb{P}_\epsilon, F \rangle =  \int F(\gamma(s)) ~\mathbb{P}_\epsilon(ds).
$$
Since $H$ is finite-dimensional and $\gamma$ is bounded, tightness
of the family $(\gamma_\# \mathbb{P}_\epsilon)_{\epsilon > 0}$ is guaranteed.
Hence, by Prokhorov's theorem, every sequence $\epsilon_n \to 0$ admits a subsequence along which $\gamma_\#\mathbb{P}_{\epsilon_n}$ converges weakly.
One can now define the concept of spontaneous stochasticity ($\SP$) 
{
\begin{definition}[$\SP$]\label{SPdef}
The family \eqref{P_eps} is said to exhibit spontaneous stochasticity relative
to $(t,x_0,\mathbb{P}_\epsilon)$ if, with
$\gamma=\gamma_{t,x_0}$, one has
$$
\SP :=
\left\{
\begin{aligned}
&\text{Strong} && \text{if } \gamma_\# \mathbb{P}_\epsilon
\rightharpoonup \mu_0 \text{ as } \epsilon \to 0,
\text{ with } \mu_0=\mu_{t,x_0} \text{ non-Dirac}, \\[6pt]
&\text{Weak} && \text{if } \gamma_\# \mathbb{P}_\epsilon
\text{ admits at least two distinct subsequential limits}.
\end{aligned}
\right.
$$
Thus both the regularization map and the selected limiting statistics depend on $(t,x_0)$.
\end{definition}\noindent 
We will omit the $(t,x_0)$-dependence in the following unless it is needed.}
\bigskip

\begin{remark}
Typical examples where $\SP$, in both its strong and weak forms, can already be seen are Peano-type systems: elementary ODEs with continuous but non-Lipschitz vector fields, for which existence holds while uniqueness may fail. They provide a minimal laboratory for the phenomenon, since the inviscid dynamics allows several possible selection mechanisms, and the selected statistics depend on the regularization and on the ambient measure. We discuss one such example, $\dot x=\sqrt{|x|}$, in detail in Section~\ref{exAmbrosio}; another classical example, $\dot x=x^{1/3}$, is treated in Appendix~\ref{Ex1}.
\end{remark}
\bigskip
\begin{remark}
One may refer to $\SP$ without specifying the triplet with respect to which it is defined by speaking of \emph{Strong} $\SP$: there exist $(t,x_0,\mathbb{P}_\epsilon)$ such that \emph{Strong} $\SP$ relative to $(t,x_0,\mathbb{P}_\epsilon)$ holds. \emph{Weak} $\SP$ then corresponds to the situation in which \emph{Strong} $\SP$ never occurs, while excluding the case of a Dirac selection measure.
\end{remark}
\bigskip
\begin{remark}
The notion of \emph{Weak} $\SP$ has recently emerged in the context of shell models using a renormalization-group formalism in \cite{AM_NLN25,AM_24end}. An example of \emph{Weak} $\SP$ is analyzed in Section \ref{exAmbrosio} (see also Appendix \ref{Ex1}). 
It should be viewed as a substantially more severe scenario than \emph{Strong} $\SP$. Because no unique limiting measure is selected, the inviscid limit does not carry well-defined statistics: probabilities and observables may depend on the subsequence along which the limit is taken. In that sense, the system may drift among different statistical states, even though its asymptotic behavior remains mathematically well defined through the set of accumulation measures. A plausible interpretation is that, at least in some cases, this reflects a poor choice of ambient measure.
Indeed, the following example shows that Lebesgue measure itself need not be the most natural choice.  Let $
\gamma(s)=e^{i\log s}$, 
and consider analytic test functions 
$
F(z)=\sum_{k\ge0}F_k\,z^k,\quad 
F_k=\frac{1}{2\pi}\int_{0}^{2\pi}F(e^{i\theta})\,e^{-ik\theta}\,d\theta.
$
A direct calculation yields
$
\bigl\langle\gamma_{\#}\mathrm{Leb}_\epsilon,\;F\bigr\rangle
=\frac{1}{\epsilon}\int_{0}^{\epsilon}F\bigl(e^{i\log s}\bigr)\,ds
=\sum_{k\ge0}\frac{F_k}{1+ik}\,e^{ik\log\epsilon},
$
so that the limit as $\epsilon\to0$ exists only along subsequences.  By contrast, if one takes 
$
\mathbb{P}_\epsilon
=(e^{-1/x})_{\#}\mathrm{Leb}_\epsilon,
$
then 
$
\gamma_{\#}\mathbb{P}_\epsilon\rightharpoonup\frac{d\theta}{2\pi},
$
the Haar measure on $\mathbb{S}^1$ (whose real part obeys the arcsin law).  Thus \emph{Strong} $\SP$ can be obtained simply by modifying the ambient measure.  Since 
$
(\gamma\circ h)_{\#}\mathrm{Leb}_\epsilon
=\gamma_{\#}\bigl(h_{\#}\mathrm{Leb}_\epsilon\bigr),
$
this adjustment is equivalent to reparameterizing the regularization curve.
\end{remark}
\bigskip
\begin{remark}
As mentioned above, imposing no restriction on the ambient measure has consequences. For instance, if one takes $\mathbb{P}_\epsilon = \delta_\epsilon$, then $\gamma_\# \mathbb{P}_\epsilon = \delta_{\phi_t^\epsilon(x_0)}$. Up to a subtlety discussed in the next section, this choice indeed recovers a classical notion of selection principle. In such a case, only two scenarios are possible: either a selection principle holds and $\phi_t^\epsilon$ converges to a solution of the inviscid system ($\mu_0$ is then a Dirac mass), or no selection principle holds and the solution oscillates rapidly with respect to $\epsilon$, for example like $\sin(1/\epsilon)$. Consequently, there exist many atomic-type measures for which \emph{strong} $\SP$ cannot occur.
What remains open in this work is the following question:
\begin{center}
\emph{Can one always find a family of ambient measures $\mathbb{P}_\epsilon$ for which either Strong $\SP$ or Dirac selection occurs?}
\end{center}

\end{remark}\noindent
From now on, we take
$$
\mathbb{P}_\epsilon = \mathrm{Leb}_\epsilon,
$$
unless stated otherwise.
\section{Selection principles and sensitivity to initial conditions}
\label{BasicSensitivity}
We now relate the measure-selection definition of $\SP$ to two more classical viewpoints: lack of selection for the inviscid problem and sensitivity to initial conditions in the inviscid limit. This section clarifies what is genuinely captured by the definition, and why Lyapunov-type criteria should be viewed as diagnostics rather than definitions of $\SP$.
\subsection{Basic properties}\label{Basic}

\begin{definition}[$LSP$]\label{LSPdef}
We say that the family $({\cal P}_\epsilon)$ exhibits a \emph{lack of selection principle} ($LSP$) if there exist two sequences $\epsilon_n \to 0$ and $\epsilon_n' \to 0$ such that the corresponding solutions $x^{\epsilon_n}$ and $x^{\epsilon_n'}$ converge in $C([0,T];H)$ to two distinct limits $x_1 \neq x_2$.
\end{definition}

\begin{definition}[$\LSPO$]\label{LSPOdef}
We say that the family $({\cal P}_\epsilon)$ exhibits $\LSPO$ if there exist $t\in(0,T)$, $x_0\in H$, and an observable ${\cal O}\in C_b(H;\mathbb R)$ such that
\begin{equation}
\liminf_{\epsilon\to0}{\cal O}(\gamma(\epsilon))
<
\limsup_{\epsilon\to0}{\cal O}(\gamma(\epsilon)).
\end{equation}
\end{definition}

\begin{remark}
Perhaps unexpectedly, $\SP$ is neither equivalent to $LSP$ nor $\LSPO$. The underlying reason lies in the mismatch between a measure-theoretic concept ($\SP$) and a classical one ($LSP$): indeed, $LSP$ may occur on a set whose relative Lebesgue  measure vanishes asymptotically. 
These are characterized by the property that there exists $x^\star \in {\mathscr S}_0$ such that $\gamma_\# \operatorname{Leb}_\epsilon \rightharpoonup \delta_{x^\star}$, while still allowing different subsequences to converge to distinct solutions!
An explicit example is provided in Appendix~\ref{DiracLSP}. 

\end{remark}

Let $\delta$-$LSP$ denote the occurrence of $LSP$ together with weak convergence to a Dirac mass and $LSP \setminus \delta:=LSP \setminus (\delta\text{-}LSP)$ and define $OLSP \setminus \delta$ analogously.
We place no restriction on $H$ beyond being a separable Banach space on which the solutions of \eqref{P_eps} take values. Separability is required here to guarantee that $H$ is also a Polish space. Definitions \ref{SPdef}, \ref{LSPdef}, and \ref{LSPOdef} remain unchanged. 

\begin{proposition}\label{SPDEFS}
Let $H$ be a separable Banach space. Assume $\gamma \in C_b((0,\infty);H)$ and that the family
$
\big\{\gamma_\# \operatorname{Leb}_\epsilon\big\}_{\epsilon>0}
\subset \mathcal P(H)
$
is tight. Then
\begin{equation}\label{spst_lsp_general}
LSP\setminus\delta
\Longrightarrow
\LSPO\setminus\delta \Longleftrightarrow \SP
\end{equation}

If, in addition, the family of regularized trajectories
$
\left\{
s \mapsto \phi_s^\epsilon(x_0)
\right\}_{\epsilon>0}
$
is relatively compact in $C([0,T];H)$, then
\begin{equation}\label{spst_lsp_equal}
LSP\setminus\delta
\Longleftrightarrow
\LSPO\setminus\delta
\Longleftrightarrow
\SP.
\end{equation}

In particular, when $H$ is finite-dimensional, tightness of
$\big\{\gamma_\#\operatorname{Leb}_\epsilon\big\}_{\epsilon>0}$ is automatic,
and the relative compactness of the regularized trajectories follows from
their uniform boundedness and equicontinuity. Hence the three notions in
\eqref{spst_lsp_equal} coincide unconditionally.
\end{proposition}

{\bf Proof}: see Appendix \ref{SP_LSP} and Fig. \ref{tricho}. 

\begin{figure}[htbp]
\centerline{\includegraphics[width=0.60\columnwidth]{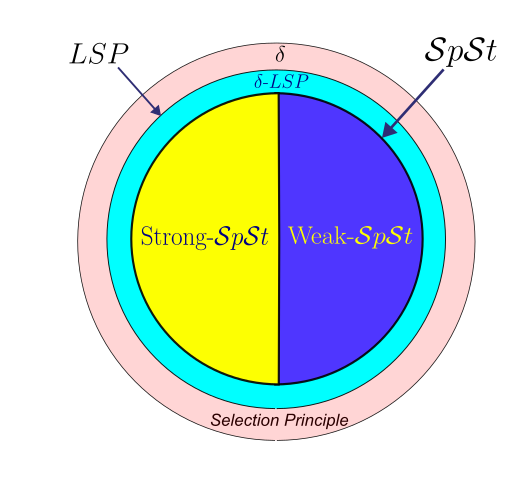}}
\caption{Mutually exclusive scenarios for the inviscid limit
in finite dimension. 
$\delta\text{-}LSP$ refers to cases having a weak limit being Dirac but at the same time a lack of selection principle. 
 It arises when $LSP$ occurs on a set whose relative
measure w.r.t. the ambient one goes to zero in the inviscid limit; see Appendix \ref{DiracLSP}. $LSP \setminus \delta$ corresponds to $LSP$ with $\delta\text{-}LSP$ excluded.
Classical
$LSP = \{ \delta\text{-}LSP\}~ \bigcup ~\{ \text{Strong-}\SP\} ~\bigcup ~\{ \text{Weak-}\SP \}$ and
Proposition \ref{SPDEFS} is $\SP\Longleftrightarrow LSP \setminus \delta \Longleftrightarrow 
\LSPO \setminus \delta$.}
\label{tricho}
\end{figure} 

An important question is to determine which classes of observables are eligible for $\LSPO$ when $\SP$ holds.
It is possible to establish that if $H$ is finite-dimensional and $\LSPO$ holds then necessarily, one must have ${\cal O}\left. \right|_{{\mathscr S}_0} \neq cst$; see Lemma \ref{OnotconstLemma}.
More generally, assume for instance that one has \emph{Strong} $\SP$ with selected non-Dirac measure $\mu_0$,
and that the system possesses certain symmetries, leading to conserved quantities.
Let $G$ be a group acting on $H$, and let $g\in G$.
If $g_{\#} \mu_0 = \mu_0$ and $\mu_0$ is ergodic
w.r.t. $g$
(see Appendix \ref{PFTP} for the definition),
then every observable satisfying $\mathcal{O}(g(x)) = \mathcal{O}(x)$
cannot satisfy $\LSPO \setminus \delta$,
since such a continuous observable must be constant $\mu_0$-a.e., hence constant on the support of $\mu_0$.
However, there always exists a nontrivial observable that satisfies $\LSPO$:
\be \label{obsuniv}
\forall x^\star \in \operatorname{supp} \mu_0,~~ {\cal O}^\star:H \to \mathbb{R}, {\cal O}^\star(x) = \| x - x^\star \|~\text{is eligible for}~
\LSPO.
\de 
This last result is indeed true in any Banach space where $\SP$ holds, no matter the dimension.
We have therefore as a direct consequence of \eqref{obsuniv} and Proposition \ref{SPDEFS}:
\begin{corollary}\label{zeroinsupp}
Assume that $\SP$ holds. Let $\mathscr{M}(\gamma)$ denote the set of all subsequential weak limits of $\gamma_{\#} \mathrm{Leb}_{\epsilon}$, namely ${\mathscr M}(\gamma) := \left\{ \mu \,:\, \exists \epsilon_n \to 0 \text{ such that } \gamma_\# \mathrm{Leb}_{\epsilon_n} \rightharpoonup \mu \right\}$ (this set is a
singleton in the case of Strong-$\SP$). If
$$
0 \in \bigcup_{\mu \in \mathscr{M}(\gamma)} \operatorname{supp}(\mu),
$$
then
\begin{equation}
\liminf_{\epsilon \to 0} \|\gamma(\epsilon)\|
\;<\;
\limsup_{\epsilon \to 0} \|\gamma(\epsilon)\|.
\end{equation}
\end{corollary}
This corresponds to the situation where there exists at least one inviscid solution of \eqref{P_0} whose value at time $t$ is $0$. 
In other words, the norm on $H$ provides an admissible observable for which OLSP holds.
In full generality, the converse of this corollary does not hold.
We can also indeed claim another corollary:
\begin{corollary}[Projection]\label{SpSt_O}
Let $F$ be some separable Banach space and ${\cal O}: H \mapsto F$ in $C_b(H;F)$ and assume that
Strong-$\SP$ holds, then provided ${\cal O}$ is not $\mu_0$-a.e. constant, one has
\be 
({\cal O} \circ \gamma)_\# \operatorname{Leb}_\epsilon \rightharpoonup
{\cal O}_\# \mu_0~\text{non-Dirac}.
\de 
\end{corollary}
\begin{proof}
It is a direct consequence of the continuous mapping theorem; see Appendix \ref{PFTP}. Since ${\cal O}$ is continuous, it gives $({\cal O} \circ \gamma)_\# \operatorname{Leb}_\epsilon \rightharpoonup
{\cal O}_\# \mu_0$. From the hypothesis that ${\cal O}$ is not $\mu_0$-a.e. constant, the limit is non-Dirac.
\end{proof}

\subsection{Spontaneous Stochasticity and Sensitivity to Initial Conditions}\label{sensitivityCI}

This section examines the relation between $\SP$ and the loss of continuity of the regularized flow map in the inviscid limit. A common strategy in the turbulence literature is to perturb the initial condition and monitor the $L^2$-distance between two solutions whose initial data coalesce as the regularization is removed. If this distance remains bounded away from zero at some finite time, then the regularized flow maps cannot admit a uniform modulus of continuity in the inviscid limit. This is stronger than classical chaotic sensitivity: for a regular chaotic flow, small errors may be strongly amplified, but at any fixed finite time the separation still vanishes with the size of the initial perturbation.

This criterion is useful, but it should not be taken as a definition of $\SP$. Interpreting $\SP$ as the blow-up of finite-time Lyapunov exponents is misleading, because it captures only one particular way in which the inviscid limit can lose deterministic continuity. In the present framework, $\SP$ is instead a measure-selection phenomenon: the regularized dynamics selects a non-Dirac limiting law in the inviscid limit. Sensitivity to initial conditions may therefore be a powerful diagnostic for $\SP$, but it is not equivalent to $\SP$ itself. The purpose of this section is to make this distinction precise.

In these approaches, randomness enters through the initial condition. It is therefore natural to ask how the limiting statistics obtained in this way compare with those obtained by pushing forward an ambient measure by the map $\gamma$ introduced in Definition~\ref{gammadef}. Equivalently, one asks how this initial-data randomization framework fits into our definition of $\SP$.

Before addressing these questions, we enlarge the class of regularized problems and inviscid limits to which $\SP$ applies. This is not a new definition of $\SP$: the criterion is unchanged, but the regularization framework is made more flexible.

\subsubsection{$\SP$ for general regularized problems, $\SP_{\rm turb}$ and TBM criterion}
The construction that follows differs from Definition~\ref{SPdef} only by a
few technical modifications. We introduce the \emph{regularization space}
$\Theta$, a Polish space endowed with its Borel $\sigma$-algebra
$\mathcal{B}(\Theta)$, and assume that it contains a distinguished element
$0 \in \Theta$ representing the inviscid (zero) regularization.
We consider regularized vector fields of the form
$$
f : H \times \Theta \to H, \qquad (x,\theta) \mapsto f(x,\theta),
$$
and replace the condition $f \in \V$ by the requirement that
$f(\cdot,\theta)$ is Lipschitz for every $\theta \in \Theta \setminus
\{0\}$. In addition, the initial condition is allowed to depend on the
regularization parameter, namely $x(0) = x_0 + \delta(\theta)$.

For each $\theta \in \Theta \setminus \{0\}$ we then consider the
regularized Cauchy problem
\begin{equation}\label{P_theta} \tag{${\cal P}_\theta$}
\dot{x} = f(x,\theta), \qquad
x(0) = x_0 + \delta(\theta) \in H, \qquad t \in [0,T].
\end{equation}

Convergence to the inviscid system is then expressed by
$$
\lim_{\theta \to 0} \| f(\cdot,\theta) - f_0 \|_\infty = 0,
\qquad
\lim_{\theta \to 0} \| \delta(\theta) \| = 0,
$$
where $f(\cdot,0) = f_0$

The regularization function $\gamma$ in Definition \ref{gammadef} is a measurable curve
$$
\gamma:\Theta \setminus \{0\} \to H,~\theta \mapsto \phi_t^\theta\left(x_0+\delta(\theta)\right)~
\mbox{solution at time $t$ of \eqref{P_theta}}.
$$
The definition of $\SP$ is the same as before, but with more general spaces involved. We rewrite it for convenience:

\begin{definition*}[$\SP$ for general regularizations]\label{genSPdef}
Let $(\mathbb{P}_\epsilon)_{\epsilon > 0}$ be a family of probability measures in ${\cal P}(\Theta \setminus \{0\})$ such that $\mathbb{P}_\epsilon \rightharpoonup \delta_0$. The family (\ref{P_theta})  is said to exhibit spontaneous stochasticity relative to $(t,x_0,\mathbb{P}_\epsilon)$ if
$$
\SP := 
\left\{
\begin{aligned}
&\text{Strong} && \text{if } \gamma_\# \mathbb{P}_\epsilon \rightharpoonup \mu_0 \text{ as } \epsilon \to 0, \text{ and } \mu_0 \text{ non-Dirac}, \\[6pt]
&\text{Weak} && \text{if } \gamma_\# \mathbb{P}_\epsilon \text{ admits at least two distinct subsequential limits}.
\end{aligned}
\right.
$$
\end{definition*}

We now compare the formulation $\SP$ of Definition~\ref{SPdef} with a turbulence-inspired version, denoted by $\SP_{\rm turb}$.  The comparison below isolates the structural ingredients of the two settings: the parameter space, the regularization parameter, the perturbed initial condition, and the associated sampling measure.
{
$$
\renewcommand{\arraystretch}{1.35}
\begin{array}{|l|l|l|l|}
\hline
&
\displaystyle \SP
&
\displaystyle \SP_{\rm turb}
& \mbox{TBM criterium}
\\
\hline

\displaystyle \Theta
&
\displaystyle \mathbb{R}_0^+
&
\displaystyle \{(0,0_H)\} \cup \mathbb{R}^+ \times H
&
\displaystyle \{(0,0_H)\} \cup \mathbb{R}^+ \times H
\\

\displaystyle \theta
&
\displaystyle \epsilon
&
\displaystyle (\epsilon,\xi)
&
\displaystyle (\epsilon,\xi)
\\

\displaystyle \delta(\theta)
&
\displaystyle 0
&
\displaystyle \xi
&
\xi 
\\

\displaystyle \mathbb{P}_\epsilon \rightharpoonup \delta_0
&
\displaystyle {\rm Leb}_\epsilon
&
\displaystyle
{\rm Leb}_\epsilon(ds)\rho_s(d\xi),
\;
\rho_s \rightharpoonup \delta_0
&
\delta_\epsilon \otimes \rho_\epsilon
\\
\displaystyle \gamma_\# \mathbb{P}_\epsilon & \mu_\epsilon = \int \delta_{\phi_t^s(x_0)} \mathbb{P}_\epsilon(ds) &
\displaystyle \nu_\epsilon = \int \delta_{\phi_t^s(x_0+\xi)} \mathbb{P}_\epsilon(ds,d\xi) 
&
\displaystyle \kappa_\epsilon = \gamma_\# (\delta_\epsilon \otimes \rho_\epsilon) 
\\
\hline
\end{array}
\vspace*{0.5cm}
$$
}
For every $F\in C_b(H)$, the corresponding laws are defined by
\be \label{spspsp}
\begin{array}{rcll}
\SP:
&
\langle \mu_\epsilon,F \rangle
&
=
&
{\displaystyle
\frac{1}{\epsilon}
\int_0^\epsilon F(\phi_t^s(x_0))\,ds}
\\[1em]

\SP_{\rm turb}:
&
\langle \nu_\epsilon,F \rangle
&
=
&
{\displaystyle
\frac{1}{\epsilon}
\int_0^\epsilon ds
\int F(\phi_t^s(x_0+\xi))\,\rho_s(d\xi)}
\end{array}
\de
One must notice that $\SP_{\rm turb}$ reduces to $\SP$ 
 when $\rho_s = \delta_0$ but in general the two notions involve \emph{two different perturbations of the inviscid problem}.
\\
We now introduce a TBM criterion, adapted from \cite{Simon_Jeremie_AM20}, for detecting the occurrence of spontaneous stochasticity within the turbulence-inspired framework $\SP_{\rm turb}$.
\begin{definition}[TBM]\label{TBMdef}
Denote $\kappa_\epsilon := \gamma_\# (\delta_\epsilon \otimes \rho_\epsilon)$
 with
$\rho_\epsilon \rightharpoonup \delta_0$.
We say that TBM occurs if for $X_i = x_0 + \xi_i$, $i=1,2$ and $\xi_i
\sim \rho_\epsilon$ i.i.d. random variables, one has
\be \label{tbm0}
\liminf_{\epsilon\to0}
\mathbb{E}\!\left[
\left\|
\phi_t^\epsilon(X_1)-\phi_t^\epsilon(X_2)
\right\|^2
\right]
=
2L>0.
\de 
Therefore, it is equivalent to
\be \label{tbm1}
\liminf_{\epsilon\to0} V_\epsilon=L>0,
\qquad
V_\epsilon
=
\operatorname{Var}(\phi_t^\epsilon(X))
= \langle \kappa_\epsilon,\|\cdot\|^2 \rangle - \left\| \langle \kappa_\epsilon,{\rm Id} \rangle \right\|^2.
\de 
\end{definition}
We note that if the inviscid system lacks a selection principle, one expects 
that $\kappa_\epsilon$ has no limit for the weak topology, possibly only along subsequences. It is similar to taking $\mathbb{P}_\epsilon = \delta_\epsilon$
in our original $\SP$ definition, giving $\gamma_\# \delta_\epsilon = 
\delta_{\phi_t^\epsilon(x_0)}$. In other words, when the system lacks a selection principle, for the atomic ambient measure $\mathbb P_\epsilon=\delta_\epsilon$, Strong $\SP$ cannot occur.
Note also that the TBM criterion is indeed similar to $\LSPO$, Definition \ref{LSPOdef}. Although useful, it cannot be used to define 
the probability measures involved in the inviscid limit.
This is why we consider $\SP_{\rm turb}$ in \eqref{spspsp} as the correct formulation to
characterize inviscid selected statistics, namely the weak limit of
the measures $\nu_\epsilon$.
\subsubsection{Comparison between $\SP$, $\SP_{\rm turb}$ and TBM criterion}\leavevmode\par
\para{TBM as a sufficient condition to detect spontaneous stochasticity}\\
It will be useful to compare the variance of $\nu_\epsilon$ with the variance
of $\kappa_\epsilon = \gamma_\# (\delta_\epsilon \otimes \rho_\epsilon)$ in TBM definition. One has
\be \label{Varnueps}
\operatorname{Var}(\nu_\epsilon) = 
\frac{1}{\epsilon} \int_0^\epsilon \operatorname{Var}(\kappa_s)~ds
+ \frac{1}{\epsilon} \int_0^\epsilon \| B_s -\thickbar{B}_\epsilon \|^2~ds,
\de 
where $B_s= \int \phi_t^s(x_0+\xi) \rho_s(d\xi)$ and $\thickbar{B}_\epsilon = 
\frac{1}{\epsilon} \int_0^\epsilon B_s  ds$. We can therefore claim without proof that

\begin{property}[TBM implies $\SP_{\rm turb}$]
Assume that TBM occurs:
$$
\liminf_{\epsilon\to0}\operatorname{Var}(\kappa_\epsilon)=L>0.
$$
Then, 
$$
\liminf_{\epsilon\to0}\operatorname{Var}(\nu_\epsilon)\geq L>0.
$$
Hence every weak limit point of $(\nu_\epsilon)_{\epsilon>0}$ is non-Dirac.
Consequently, if $\nu_\epsilon$ converges, Strong $\SP_{\rm turb}$ occurs;
if it has at least two distinct subsequential limits, Weak $\SP_{\rm turb}$
occurs, but with non-Dirac subsequential limits.
\end{property}
If $\rho_s=\delta_0$, then $\kappa_s=\delta_{\phi_t^s(x_0)}$, hence
$\operatorname{Var}(\kappa_s)=0$ for every $s>0$. Therefore TBM cannot detect
this situation.  The variance in \eqref{Varnueps} then takes the form
$$
\operatorname{Var}(\nu_\epsilon) = \frac{1}{\epsilon}
\int_0^\epsilon \left\|  \phi_t^s(x_0) - \frac{1}{\epsilon} \int_0^\epsilon 
\phi_t^\sigma(x_0) d\sigma
\right\|^2 ds = \langle \mu_\epsilon,\|\cdot \|^2 \rangle - \| \langle \mu_\epsilon,{\rm Id} \rangle \|^2  = \operatorname{Var}(\mu_\epsilon).
$$
In particular, Strong $\SP$ may still occur if
$\mu_\epsilon\rightharpoonup\mu_0$ with $\mu_0$ non-Dirac.
\bigskip

TBM provides a useful fixed-scale criterion for detecting non-collapse of the
laws $\kappa_s$. However, $\SP_{\rm turb}$ is formulated at the level of the
averaged laws $\nu_\epsilon$. Hence non-Dirac limiting statistics may also be
created by the scale-mixing of the fixed-scale means, a mechanism which is
already present in the original $\SP$ formulation and is not detected by TBM.
\begin{center}
\emph{TBM is a sufficient criterion for $\SP_{\rm turb}$, but not a necessary
one.}
\end{center}
{In addition, the TBM criterion detects only a particular subcase of Weak $\SP$: one in which all limiting laws are non-Dirac, namely have nonzero variance. The presence of a single Dirac accumulation point already prevents a full TBM criterion from holding, since the variance vanishes along the corresponding subsequence. Thus TBM is a sufficient criterion for a stronger form of Weak $\SP$, but not for Weak $\SP$ itself. By contrast, in the Strong $\SP$ regime, where the limiting law is unique, TBM remains a natural sufficient criterion for proving that this law is non-Dirac.}

\para{Comparison between $\SP$ and $\SP_{\rm turb}$}

The following elementary estimate gives a useful comparison between $\SP$ and
$\SP_{\rm turb}$. It shows that, if the measures $\rho_s$ concentrate
sufficiently fast relative to the modulus of continuity of the regularized
flow maps, then the additional perturbation of the initial condition does not
modify the selected inviscid statistics. Let $\Omega_s$ be a modulus of continuity for the flow map
$x\mapsto\phi_t^s(x)$. Then, for every $F\in{\rm Lip}_b(H)$, one has
$$
\left|
\langle \mu_\epsilon-\nu_\epsilon,F\rangle
\right|
\leq
C_F \,
\frac{1}{\epsilon}
\int_0^\epsilon
\int_H
\Omega_s(\|\xi\|)\,\rho_s(d\xi)\,ds.
$$
Since bounded Lipschitz functions determine weak convergence of probability
measures on $H$, this estimate yields the following comparison principle.

\begin{property}\label{modulus_spst}
Assume that
\be \label{omegas}
\frac{1}{\epsilon}
\int_0^\epsilon
\int_H
\Omega_s(\|\xi\|)\,\rho_s(d\xi)\,ds
\longrightarrow 0
\qquad
\text{as } \epsilon \to 0.
\de 
Then the families $(\mu_\epsilon)_{\epsilon>0}$ and
$(\nu_\epsilon)_{\epsilon>0}$ have the same weak limit points. More precisely,
for every sequence $\epsilon_j \to 0$, if
$
\mu_{\epsilon_j} \rightharpoonup \mu,~{\rm and}~
\nu_{\epsilon_j} \rightharpoonup \nu,
$
then
$$
\mu=\nu.
$$
\end{property}

The estimate above measures the competition between the concentration of
$\rho_s$ at the origin and the deterioration, as $s\to0$, of the modulus of
continuity of the map $x\mapsto\phi_t^s(x)$. If \eqref{omegas} holds, then
the initial-data perturbation is asymptotically invisible and does not change
the selected inviscid statistics. If \eqref{omegas} fails, this perturbation
may survive in the inviscid limit and produce a different selection mechanism.
In applications one may expect, for instance,
$\Omega_s(r)$ behaving as $ e^{L_s}r$ with $L_s\to+\infty$, so that the concentration
of $\rho_s$ has to compensate for the loss of continuity of the regularized
flow maps.
\bigskip 

By the preceding criterion, if canonical $\SP$ occurs relative to
$(t,x_0,{\rm Leb}_\epsilon)$, then there exists a family
$\rho_s\rightharpoonup\delta_0$ concentrating sufficiently fast such that
$\SP_{\rm turb}$ occurs relative to
$(t,x_0,{\rm Leb}_\epsilon(ds)\rho_s(d\xi))$, with the same limiting law.
This is only an existence result. A universal implication, valid for all
families $\rho_s\rightharpoonup\delta_0$, should not be claimed without
specifying an admissible class of perturbation laws, for instance by imposing
centering or compatibility with the moduli of continuity of the
regularized flow maps. Conversely, $\SP_{\rm turb}$ does not imply $\SP$ in
general, as illustrated by examples such as $\dot x=\tanh(x/\epsilon)$ at
$x_0=0$; see Appendix~\ref{tanh_eps}.
\bigskip

\section{Set-theoretic results: from attainability to universality classes}\label{setheo}
In what follows, it will be convenient to reinterpret the singular regime 
$\epsilon \to 0$ in \eqref{P_eps} as an asymptotic limit $\tau \to \infty$, 
where $\tau$ may be viewed as a Reynolds-type parameter rather than a 
viscosity-type one. The regularized system \eqref{P_eps} will therefore be 
denoted by $(\mathcal{P}_\tau)$, and we retain the same notation for the 
vector fields in $\V$, now considered in the limit $\tau \to \infty$, 
which corresponds to the inviscid regime. Another key difference is that we 
are interested in all regularizations of \eqref{P_0}, rather than a single 
fixed one. To this end, we introduce the space of \emph{regularization curves} 
as
\be \label{Gamma}
\Gamma := \{\, \gamma : \mathbb{R}^+ \to H \;:\; \exists f \in \V 
\ \text{such that}\ \gamma(\tau) = \phi_t^\tau(x_0) \,\} 
\subset C_b(\mathbb{R}^+;H).
\de
This set is equipped with the compact-open topology of uniform convergence on compact sets. Note that, for a given curve $\gamma$, there may be several vector fields that generate $\gamma$. We now restate Definition~\ref{SPdef}, giving it a more set-theoretic 
formulation.

\begin{definition}[$\SP$]\label{SPdef5}
Let $\gamma \in \Gamma$, spontaneous stochasticity relative to $(t,x_0,\mathbb{P}_\tau)$ occurs if
\be \label{SPdefbis}
\SP :=
\left\{
\begin{aligned}
&\text{Strong} \quad 
{\mathscr M}(\gamma)=\{\mu\} ~\text{ for some non-Dirac}~
\mu \in {\mathscr M}_0 = {\cal P}({\mathscr S}_0); \\
&\text{Weak} \quad 
{\mathscr M}(\gamma) \text{ is not a singleton}
\end{aligned}
\right.
\de
where ${\mathscr M}(\gamma)$ is the set of probability measures obtained as subsequential limits for the weak topology:
\begin{equation}\label{Mgamma_def}
{\mathscr M}(\gamma) := \left\{ \mu \,:\, \exists \tau_n \to \infty\text{ such that } \gamma_\# \mathbb{P}_{\tau_n} \rightharpoonup \mu \right\}.
\end{equation}
\end{definition}

\begin{remark}
Since one adopts a measure-theoretic approach, $\delta$-$LSP$ cannot be
discriminated from a genuine selection principle.
As we will see next, they live at the border of $\cup_\gamma {\mathscr M}(\gamma)$
as extreme points.
We also emphasize that the notion of \emph{Weak} $\SP$ does not require specifying
whether the accumulation measures are Dirac masses or not.
\end{remark}
\bigskip 
\begin{remark}
It turns out that under algebraic reparametrizations $g$ of the curves $\gamma$, one has $\mathscr{M}(\gamma \circ g) = \mathscr{M}(\gamma)$. 
In other words, for a given family $({\cal P}_\epsilon)_{\epsilon \in \mathbb{R}^+}$, 
there exist many equivalent regularized problems $({\cal P}_{g(\tau)})_{\tau \in \mathbb{R}^+}$ that 
yield the same inviscid statistics in the limit $\tau \to \infty$.
This fact is of independent interest for our later purposes, and we provide full details in 
Appendix~\ref{Alg} for the interested reader.

\end{remark}

\subsection{${\mathscr M}={\mathscr M}_0$: all non-Dirac probability measures are attainable through Strong $\SP$}
\label{allprob}
The main result of this subsection is that every probability measure in ${\mathscr M}_0$ can be attained as the unique subsequential limit associated with a suitable regularization curve. In particular, whenever the inviscid system is ill posed, every non-Dirac probability measure in ${\mathscr M}_0$ can be realized by a regularization exhibiting strong $\SP$. This remarkable flexibility is possible because one allows for a very large class of regularizations.

To simplify the discussion, we use again $\mathbb{P}_\tau = {\rm Leb}_\tau$
We recall the standard notation for the convex hull:
$$
\operatorname{co}(A) := \left\{ \sum_{i=1}^k \theta_i a_i \,:\, k \in \mathbb{N},~ \theta_i \geq 0,~ \sum_{i=1}^k \theta_i = 1,~ a_i \in A \right\}.
$$
We obtain the following result for the set of all subsequential limits arising from regularization curves in $\Gamma$:
\begin{theorem}\label{M0M}
Let ${\displaystyle {\mathscr M} := \bigcup_{\gamma \in \Gamma} {\mathscr M}(\gamma)}$, and ${\cal E} := \{\delta_x\}_{x \in {\mathscr S}_0}$ 
where  ${\mathscr S}_0$ and $\Gamma$ are defined in  \eqref{S0} and \eqref{Gamma}. Then ${\mathscr M}$ is compact, convex, and
\begin{equation}
{\mathscr M} = {\mathscr M}_0 = \overline{\operatorname{co}}({\cal E}),
\end{equation}
where the closure is taken with respect to the weak topology.

Moreover, $\forall \mu \in {\mathscr M}_0$, $\exists \gamma_\mu \in \Gamma$ such that
\be \label{singleton}
{\mathscr M}(\gamma_\mu) = \{ \mu \}.
\de 
\end{theorem}
\noindent
\begin{proof}
The inclusion ${\mathscr M} \subset {\mathscr M}_0$ is immediate; the main difficulty lies in proving ${\mathscr M}_0 \subset {\mathscr M}$. The argument proceeds in several steps and the proof is constructive.
First,  we build, for each $x \in {\mathscr S}_0$, an explicit vector field in $\V$ such that the associated regularization curve $\gamma_x \in \Gamma$ satisfies ${\mathscr M}(\gamma_x) = \{\delta_x\}$, i.e., ${\cal E} \subset {\mathscr M}$. In other words, for such regularizations, one has indeed a classical selection principle.
We then show that regularizations can be constructed to attain arbitrary finite convex combinations of these Dirac measures, so that ${\operatorname{co}}({\cal E}) \subset {\mathscr M}$. We thus have
$
\operatorname{co}(\mathcal{E}) \subseteq \mathscr{M} \subseteq \mathscr{M}_0 = \overline{\operatorname{co}}(\mathcal{E}) \Longrightarrow \thickbar{\mathscr M} = {\mathscr M}_0 =\overline{\operatorname{co}}(\mathcal{E}). 
$ The measures in ${\cal E}$ are often referred to as the \emph{extreme points} or \emph{pure states} of ${\mathscr M}_0$.
We then show that for all $\mu \in \thickbar{\mathscr M}$, one can construct some vector field and corresponding curve $\gamma_\mu \in \Gamma$ such that ${\mathscr M}(\gamma_\mu) = \{\mu\}$. The proof is more difficult and amounts to renormalize any sequence in ${\rm co}({\cal E})$ in a controlled way. The immediate consequence is that ${\mathscr M}$ is closed. Since ${\mathscr M}_0$ is compact, we conclude that ${\mathscr M}$ is compact.

\paragraph{{\bf Step 1.} ${\mathscr M}_0$ is compact and convex and
$\mathscr{M}_0 = \overline{\operatorname{co}}(\mathcal{E})$.}
\mbox{}\\
It is convex since taking $\mu_1,\mu_2 \in {\cal P}({\mathscr S}_0)$ then for $\theta \in [0,1]$, $\mu({\mathscr S}_0)=(1-\theta) \mu_1({\mathscr S}_0) + \theta \mu_2({\mathscr S}_0) = 1$. Moreover, for all Borel sets $B \subset {\mathscr S}_0$, $\mu(B)\geq 0$ since
$\mu_1(B),\mu_2(B) \geq 0$ and thus $\mu \in {\mathscr M}_0$.

Due to the compactness of ${\mathscr S}_0$, one has immediate tightness: for all $\mu \in {\cal P}({\mathscr S}_0)$, $\mu({\mathscr S}_0)=1$ so that the probability measures in ${\cal P}({\mathscr S}_0)$ are tight. Prokhorov theorem then states that tightness of probability measures in a Polish space is equivalent to be relatively compact in the weak topology. It thus remains to show that the set is closed. Let $\mu_k \rightharpoonup \mu$ then since $\limsup_{k \to \infty} \mu_k({\mathscr S}_0) \leq \mu({\mathscr S}_0)$, then $\mu({\mathscr S}_0) \geq 1$. By weak convergence $\mu(H) = 1$ and $\mu(B) \geq 0$ for all Borel sets, therefore $\mu({\mathscr S}_0) \leq 1$ i.e. $\mu({\mathscr S}_0)=1$ and $\mu \in {\mathscr M}_0$.

${\mathscr M}_0$ is a compact convex subset of the locally convex space of signed measures and by the Krein-Milman theorem it is the closed convex hull of its extreme points.
One must show that the extreme points of $\mathscr{M}_0$ are the Dirac measures. Let us take a Dirac measure $\delta_x$ for $x \in \mathscr{S}_0$, and assume $\delta_x = \theta \mu_1 + (1-\theta) \mu_2$ for $\theta \in (0,1)$, then $\delta_x(\{x \}) = 1$ and therefore $\mu_1(\{x\}) = \mu_2(\{x\}) = 1$ so that $\mu_1=\mu_2=\delta_x$. It is not possible to have other measures since it would decompose as a convex combination of Dirac measures. Therefore 
$\mathscr{M}_0 = \overline{\operatorname{co}}(\mathcal{E})$.

\paragraph{{\bf Step 2.} $\mathscr{M}(\gamma) \subset {\mathscr M}_0$.}
\mbox{}\\
Let $\mu \in \mathscr{M}(\gamma)$, by definition one can find a subsequence $\tau_n \to \infty$ such that $\gamma_\# \mathrm{Leb}_{\tau_n} \rightharpoonup \mu$. Since $f \in \V$ converges uniformly to $f_0$ and $\gamma(\tau) = \phi_t^\tau(x_0)$, the set of accumulation points $A= \left\{ a \in H~:~\exists \tau_p \to \infty~\text{such that}~\gamma(\tau_p) \to a \right\} \subset {\mathscr S}_0$ and since ${\mathscr S}_0$ is closed, $\thickbar{A}
\subset {\mathscr S}_0$.
Let $N$ be some open set with $N \cap {\mathscr S}_0 = \emptyset$ and thus $N \cap \thickbar{A} = \emptyset$, then
$\mu(N) = \lim_{n \to \infty} \frac{1}{\tau_n} 
\int_0^{\tau_n} \mathbbm{1}_N (\gamma(s)) ds = 0$.
Therefore, $\mu(N)=0$ for all $N$ like above, i.e. the support of $\mu \subset {\mathscr S}_0$. 

\paragraph{{\bf Step 3.} $\mathscr{M}(\gamma)$ is compact.} \label{step3}
\mbox{}\\
The space ${\mathscr M}(\gamma)$ is never empty due to tightness of
the measures $\gamma_\# {\rm Leb}_\tau$ and Prokhorov theorem. We first show that it is 
compact in the weak topology. It suffices to show that it is closed and since it is in the compact set $\mathscr{M}_0$, it is compact as well. The proof is a classical diagonal argument. Let $\mu_k
\in {\mathscr M}(\gamma) \to \mu$.
By hypothesis, there is a sequence $\tau_{n,k} \to \infty$ such that $
\lim_{n \to \infty} \frac{1}{\tau_{n,k}} \int_0^{\tau_{n,k}} F(\gamma(s)) ds = \langle \mu_k,F \rangle$. For each $n$, one can choose
$k=k(n)$ such that $\left| \frac{1}{\tau_{n,k}} \int_0^{\tau_{n,k}} F({\gamma(s)}) ds - \langle \mu_k,F\rangle \right| \leq \frac{1}{n}$.
The sequence $\tau_n' = \tau_{n,k(n)}$ is such that $\lim_{n \to \infty} \frac{1}{\tau_n'} \int_0^{\tau_n'} F({\gamma(s)}) ds = {\left\langle \mu, F \right \rangle }$, i.e. $\mu \in {\mathscr M}(\gamma)$.
\\

At this stage, it is interesting to note that $\mathscr{M}(\gamma)$ can be nonconvex. We give an example.
Let us assume that there exists a system such that $\gamma(s)=e^{i \log s}$.
We write the test functions like $F(z) = \sum_{k \geq 0} F_k z^k$. Let
$C(\epsilon) = \frac{1}{\epsilon} \int_0^\epsilon F(\gamma(s)) ds = \sum_{k \geq 0} G_k e^{i k \log \epsilon}$ with $G_k = \frac{F_k}{1+ik}$. Then we can exhibit subsequences $\epsilon_n^{(a)} \to 0$ of the form $\epsilon_n^{(a)} = \mathrm{exp}(a_n-2\pi n)$
with $a_n \to a$. It gives the family of measures $\langle \mu_a,F \rangle
= \sum_{k \geq 0} G_k e^{i k a}$. Consider two distinct measures $\mu_{a_1},\mu_{a_2}$ then one must find a sequence $\epsilon_n$ such that $C(\epsilon_n) \to \sum_{k \geq 0} G_k (\theta e^{i k a_1} + (1-\theta) e^{i k a_2})$. It gives the constraint: $\forall k \geq 0,
e^{i \log \epsilon_n} \to \left( \theta e^{i k a_1} + (1-\theta) e^{i k a_2} \right)^\frac1k$.
This is not possible unless $\epsilon_n$ depends on $k$.
Another consequence is that when ${\mathscr M}(\gamma)$ contains only Dirac measures as subsequential limits, then it cannot be convex by definition.
Otherwise, if it were convex, it would necessarily include convex combinations of Dirac measures that are not themselves Dirac.

\paragraph{{\bf Step 4.} $\mathcal{E} \subset \mathscr{M}$.}\label{EinM}
\mbox{}\\
Let $x \in \mathscr{S}_0$ we want to exhibit a curve denoted $\gamma_x$ such that $\gamma_x \in \Gamma$ and $\mathscr{M}(\gamma_x) = \{ \delta_x \}$. 
Consider the set of trajectories of the inviscid system 
$\left\{ \Phi_s^0(x_0), s \in [0,t] \right\}$. 
One can select a trajectory $g$ in this set such that $g(t)=x$.
(otherwise it would contradict $x \in \mathscr{S}_0$). This trajectory is $C^1$ since $\dot g = f_0(g)$ and $f_0$ is continuous, therefore:
$$
\exists g \in C^1([0,t];H)~\text{such that}~\dot g = f_0(g),~g(0)=x_0,g(t) = x.
$$

We first use Stone-Weierstrass to approximate $\dot g$ by $C^\infty$ functions as the regularization parameter $\tau$ goes to infinity. One obtains $g_\tau \in C^\infty([0,t];H)$ such that $g_\tau(0) = x_0$ and $\lim_{\tau \to \infty} ||g_\tau-g||_{C^1} = 0$. Moreover $g_\tau(t) \to x$ as $\tau \to \infty$. We then construct a vector field $F_\tau = 
\dot g_\tau$ in a small tubular neighborhood $\mathcal{T}_\tau$ of $g$ of radius size $1/\tau$. We use a smooth function cutoff $\theta_\tau  \in C^\infty(H;[0,1])$ 
such that $\theta_\tau = 0$ outside $\mathcal{T}_\tau$ and $\theta_\tau = 1$ inside say $\mathcal{T}_{2 \tau}$. We have built a smooth vector field defined as
$$
f(\cdot,\tau) := \theta_\tau(\cdot) F_\tau(\cdot) + (1-\theta_\tau(\cdot)) f_0(\cdot),
$$
namely $f(\cdot,\tau)$ is equal to $f_0$ outside the tube but connects smoothly to $F_\tau$ inside. 

We now show that this vector field is in $V_0$. Inside the tube, one has $||f(x,\tau)-f_0(x)|| = ||F_\tau(x)-f_0(x)|| \leq ||F_\tau(x) - f_0(g(x))|| + ||f_0(x)-f_0(g(x))||$. The first term goes to zero by the definition of $F_\tau$ since $\dot g_\tau \to \dot g = f_0(g)$ uniformly. The second term goes to zero uniformly by continuity of $f_0$ and the definition of the tubular neighborhood. Outside the tube $f(\cdot,\tau) = f_0(\cdot)$. In order to conclude, one has by construction $||g_\tau -g||_{C^0} \to 0$ and is the unique solution of $\dot x = f(x,\tau), x(0)=x_0$.

\paragraph{{\bf Step 5.} $\operatorname{co}(\mathcal{E}) \subset \mathscr{M}$.}
\mbox{}\\
We aim to show that any finite convex combination of Dirac measures belongs to $\mathscr{M}$. We establish the stronger result:
\be \label{co_in_M_singleton}
\forall~\mu \in \operatorname{co}(\mathcal{E}),~\exists \gamma \in \Gamma \text{ such that } \mathscr{M}(\gamma) = \{\mu\}.
\de 
It suffices to consider convex combinations of two distinct Dirac masses. Let $\theta \in (0,1)$ and $x, y \in \mathscr{S}_0$ with $x \neq y$. We claim that
$$
\exists \gamma_\theta \in \Gamma \text{ such that } \mathscr{M}(\gamma_\theta) = \left\{ \theta \delta_x + (1-\theta) \delta_y \right\}.
$$

Let $f_x, f_y \in \V$ be the regularizations constructed previously such that the inviscid limits of the associated trajectories converge to $x$ and $y$, respectively. Let us define $f_\theta$ by
\be \label{controlswitch}
\left\{
\begin{aligned}
f_\theta(\cdot, s) & := a_\theta(s) f_x(\cdot, s) + \left(1 - a_\theta(s)\right) f_y(\cdot, s), \\
a_\theta(s) &= \frac{1}{2} + \frac{1}{2} \tanh\left( s \left( \sin {2\pi}{s} + c_\theta \right) \right), \quad c_\theta = \sin\left( \pi \left( \theta - \frac{1}{2} \right) \right).
\end{aligned}
\right.
\de 

This construction ensures that
$$
\lim_{\tau  \to \infty} \frac{1}{\tau} \int_0^\tau a_\theta(s) \, ds = \theta,
$$
and that $f_\theta$ approximates $f_x$ (resp. $f_y$) when $a_\theta(s) \approx 1$ (resp. $0$). Indeed, define $\gamma_\theta(s) := 
\phi_t^{s,\theta}(x_0)$ the solution corresponding to the vector field $f_\theta$ above, we prove that
$$
\lim_{\tau \to \infty} \frac{1}{\tau} \int_0^\tau \delta_{\gamma_\theta(s)} \, ds
=
\lim_{\tau \to \infty} \frac{1}{\tau}   
   \int_0^\tau \delta_{\hat \gamma_\theta(s)} \, ds
$$
where
$$
\hat{\gamma}_\theta(s) := x \, \hat{a}_\theta(s) + y \left(1 - \hat{a}_\theta(s)\right), \quad
\hat{a}_\theta(s) := \frac{1}{2} + \frac{1}{2} \operatorname{sign} \left( \sin 2\pi s + c_\theta \right) \in \{0,1\}.
$$
It follows  from $\gamma_\theta(s) - \hat{\gamma}_\theta(s)\underset{s\to \infty}{=} o(1)$ and the dominated convergence theorem. Since $\hat{a}_\theta$ takes values in $\{0,1\}$, we obtain
$$
\lim_{\tau \to \infty} \frac{1}{\tau} \int_0^\tau \hat{a}_\theta(s) ds = \theta,
$$
and similarly for $1 - \hat{a}_\theta(s)$, yielding the result.
\\

It remains to verify that $f_\theta \in \V$. Indeed,
$$
\| f_\theta - f_0 \|_\infty \leq \| a_\theta (f_x - f_0) + (1 - a_\theta)(f_y - f_0) \|_\infty 
\leq \| f_x - f_0 \|_\infty + \| f_y - f_0 \|_\infty,
$$
which goes to zero by assumption. Since $f_x,f_y \in \operatorname{Lip}(H;H)$ one has $f_\theta$ Lipschitz as well.

This construction extends to any finite convex combination of Dirac measures. Therefore,
\begin{equation} \label{coEMM0}
\operatorname{co}(\mathcal{E}) \subseteq \mathscr{M} \subseteq \mathscr{M}_0 = \overline{\operatorname{co}}(\mathcal{E}).
\end{equation}

\paragraph{{\bf Step 6.} $\forall \mu \in {\mathscr M}_0$,
$\exists \gamma_\mu \in \Gamma$ such that ${\mathscr M}(\gamma_\mu) = \{\mu\}$ and therefore $\mathscr{M}$ is closed.}
\mbox{}\\
The strategy is to renormalize a sequence $\mu_n \to \mu$ and glue together the corresponding curves $\gamma_n$ in a controlled way. This is feasible because no bounded variation constraint is imposed—only the asymptotic behavior via Birkhoff averages is relevant.

Then, the conclusion will follow by taking the closure of 
Eq.~\eqref{coEMM0}), so that one has $\thickbar{\mathscr{M}}  = 
\overline{\operatorname{co}}(\mathcal{E}) $, namely
one can choose $\mu_n \in \operatorname{co}(\mathcal{E}) \subset {\mathscr M}$. Therefore from \eqref{co_in_M_singleton},  there exists a curve $\gamma_n$ such that
$$
\mathscr{M}(\gamma_n) = \{ \mu_n \}.
$$

\smallskip
\begin{figure}[htbp] \centering\includegraphics[scale=0.5]{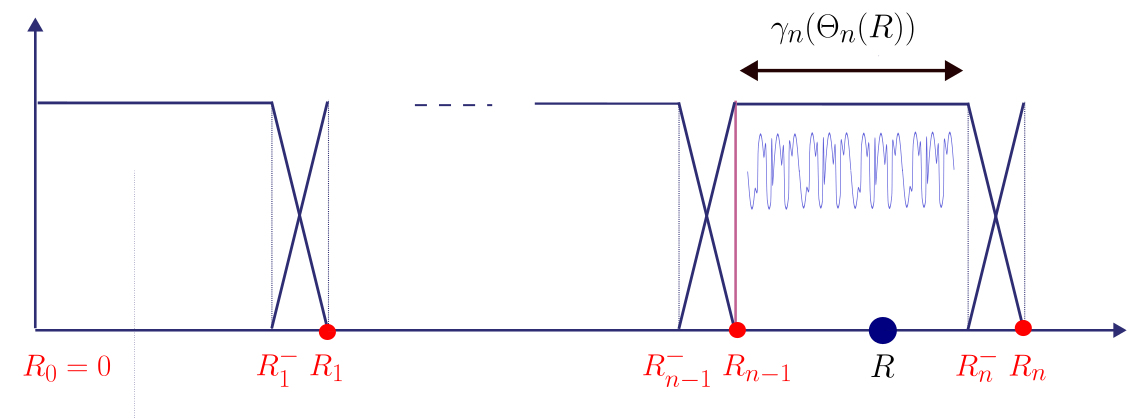}
\caption{A sketch view of the rescaling and concatenation of the sequence $\mu_n$.}
\label{gamu}
\end{figure}

Let $(\R_n^-)$, $(\R_n)$, and $(\lambda_n)$ be increasing sequences with $\R_n^-, \R_n, \lambda_n \to \infty$, $\R_0 = \R_0^- = 0$, and $\R_n^- < \R_n$, where $|\R_n - \R_n^-| \ll 1$. Let $S \geq 0$ and 
introduce a smooth partition of unity $(\chi_n)_{n \geq 1}$ with $0 \leq \chi_n \leq 1$, such that (see Fig. \ref{gamu}):
$$
\operatorname{supp} \chi_n = U_n := [\R_{n-1}^-, \R_n], \qquad
\chi_n(S) = 1 \text{ on } J_n := [\R_{n-1}, \R_n^-],
\qquad \sum_{n \geq 1} \chi_n(S) = 1.
$$

Define the vector field and corresponding regularization curve:
$$
f(\cdot, S) := \sum_{n \geq 1} \chi_n(S) f_n\left( \cdot, \Theta_n(S) \right),
\quad \Theta_n(S) := \lambda_n \frac{S - \R_{n-1}}{\R_n^- - \R_{n-1}},
\quad \gamma(S) := \phi_t^S(x_0).
$$
In particular, for $S \in J_n$, we recover:
$$
\gamma(S) = \gamma_n(\Theta_n(S)).
$$

Because $(\chi_n)$ is a partition of unity and each $f_n \in \V$, we have $f \in \V$ since $f_n \in \operatorname{Lip}$, and moreover
$$
\lim_{S \to \infty} \| f(\cdot, S) - f_0(\cdot) \|_\infty
\leq \lim_{S \to \infty} \sum_{n \geq 1} \chi_n(S)
\| f_n(\cdot, \Theta_n(S)) - f_0(\cdot) \|_\infty = 0.
$$
\\
Fix $\R \in (\R_{n-1}, \R_n^-]$. Define:
$$
A := \frac{1}{\R} \int_0^R F(\gamma(S)) \, dS = A_0 + A_1 + A_2,
$$
with
\begin{equation}
\begin{aligned}
A_0 &= \frac{1}{\R} \sum_{k = 1}^{n-1} \int_{J_k}
F(\gamma_k(\Theta_k(S))) \, dS 
= \sum_{k=1}^{n-1} \frac{\R_k^- - \R_{k-1}}{\R}  \frac{1}{\lambda_k}
\int_0^{\lambda_k} F(\gamma_k(S)) \, dS, \\
A_1 &= \frac{1}{\R} \sum_{k = 1}^{n-1} \int_{\R_k^-}^{\R_k}
F(\gamma(S)) \, dS, \\
A_2 &= \frac{1}{\R} \int_{\R_{n-1}}^{\R} 
F(\gamma_n(\Theta_n(S))) \, dS 
= \frac{\R_n^- - \R_{n-1}}{\lambda_n R}
\int_0^{\Theta_n(R)} F(\gamma_n(S)) \, dS \\
&= \frac{\R-\R_{n-1}}{\R} 
\frac{1}{\Theta_n(\R)}
\int_0^{\Theta_n(\R)} F(\gamma_n(S)) \, dS.
\end{aligned}
\end{equation}

Define the error term
$$
\epsilon_k(l) := \frac{1}{l} \int_0^l F(\gamma_k(S)) \, dS - \langle \mu, F \rangle,~\text{and}~\delta_k := \R_k-\R_k^- \ll 1.
$$
Since $\mu_k \rightharpoonup \mu$ and $\mathscr{M}(\gamma_k) = \{ \mu_k \}$, we have $\epsilon_k(l) \to 0$ as $l \to \infty$.
We obtain
\be  \label{A0A2}
A_0 + A_2 - \langle \mu, F \rangle = 
\frac{\R-\R_{n-1}}{\R} \epsilon_n(\Theta_n(\R)) 
-\frac{1}{\R} \langle \mu,F \rangle \sum_{k=1}^{n-1} \delta_k 
+ \frac{1}{\R}\sum_{k=1}^{n-1} (\R_k^- - \R_{k-1}) \epsilon_k(\lambda_k).
\de 
We end up with
\begin{equation}
\begin{aligned}
A - \langle \mu,F \rangle
= \frac{1}{\R} \Big\{ (\R_n^--\R_{n-1}) \theta_R \epsilon_n(\lambda_n \theta_R)
+ \sum_{k=1}^{n-1} (\R_k^- - \R_{k-1}) \epsilon_k(\lambda_k) \\
- \langle \mu,F \rangle \sum_{k=1}^{n-1}\delta_k
+ \sum_{k=1}^{n-1}
\int_0^{\delta_k} F(\gamma(S+\R_k^-)) \, dS \Big\}.
\end{aligned}
\end{equation}
where  $\theta_\R = \frac{\R-\R_{n-1}}{\R_n^--\R_{n-1}} \in (0,1]$.
Thus, we have the estimate:
\begin{equation} \label{upest}
\begin{aligned}
\left| \frac{1}{\R} \int_0^\R F(\gamma(S)) \, dS 
- \langle \mu, F \rangle \right|
\leq {} \\
\frac{1}{\R}
\left(C + |\langle \mu, F \rangle|\right)
\sum_{k = 1}^{n-1} \delta_k
+ \sum_{k = 1}^{n-1}
|\epsilon_k(\lambda_k)| (\R_k - \R_{k-1})
+ (\R_n - \R_{n-1})
\theta_R |\epsilon_n(\lambda_n \theta_R)|.
\end{aligned}
\end{equation}
The terms involving $\delta_k$ are easily controlled by choosing $\delta_k$ small enough such that $\sum_{k \geq 1} \delta_k \leq M$. Similarly, since we are free to choose $\lambda_k$ as large as we want, the terms
$\epsilon_k$ can be made small and we choose $\lambda_k$ such that $\sum_{k \geq 1} |\epsilon_k(\lambda_k)| (\R_k - \R_{k-1}) \leq M$. The third term is also bounded uniformly in $\R$. As a consequence
$$
\left| \frac{1}{\R} \int_0^R F(\gamma(S)) \, dS - \langle \mu, F \rangle \right| \leq \frac{C}{\R} \to 0.
$$
Finally, if $\R \in [\R_n^-, \R_n]$, then we set $A_2 = 0$ and the residual term is absorbed into $A_1$, the bracket in \eqref{A0A2} becomes
$$
-\frac{1}{\R} \sum_{k = 1}^{n-1} \delta_k + \frac{\R_n^-}{\R} - 1,
$$
which leads to a similar bound using $\R - \R_n^- \leq \delta_n$ with $\sum_{k=1}^n \delta_k$ instead. The result follows.
The direct consequence is that ${\mathscr M}$ is closed for the weak topology.
Theorem is proven by taking the closure of (\ref{coEMM0}) with $\thickbar{\mathscr M} = {\mathscr M}$:
\be 
{\mathscr M} = {\mathscr M}_0 = \overline{\rm co}({\cal E}).
\de
\end{proof}
{
\begin{remark}
All the objects above depend on $(t,x_0)$. More explicitly,
Theorem~\ref{M0M} should be read as
$\mathscr M(t,x_0)=\mathscr M_0(t,x_0)={\cal P}(\mathscr S_{t,x_0})$,
where $\mathscr S_{t,x_0}$ is the inviscid endpoint set at time $t$
starting from $x_0$. As long as $({\cal P}_0)$ admits a unique inviscid
endpoint at $(t,x_0)$, the set $\mathscr S_{t,x_0}$ is a singleton and
the statistical set reduces to a single Dirac mass. When nonuniqueness appears,
$\mathscr S_{t,x_0}$ opens into a compact connected continuum, and
${\cal P}(\mathscr S_{t,x_0})$ is the closed convex hull of the Dirac masses
supported on this continuum. Thus the time variable may act as a bifurcation
parameter: at times for which $\mathscr S_{t,x_0}$ is nontrivial, the
Dirac singleton regime is destroyed and genuine non-Dirac limiting laws become
attainable. In the Strong $\SP$ regime, the selected laws $\mu_{t,x_0}$
form the natural Young-measure-valued description of the inviscid statistics.
A structural condition for the loss of uniqueness behind this transition is
discussed in Section~\ref{CNforSP}.
\end{remark}}
\bigskip 
Two straightforward consequences of Theorem \ref{M0M} can be inferred.
\begin{corollary}\label{MND}
Assume (\ref{P_0}) admits nonunique solutions. Then for every non-Dirac probability measure $\mu \in {\mathscr M}_0$, 
there exists a regularization \eqref{P_eps} (equivalently, a family $({\cal P}_\tau)$) that exhibits strong $\SP$ (relative to $(t,x_0,{\rm Leb}_\epsilon)$) with selected measure $\mu$. Moreover, the set ${\mathscr M}_0\setminus{\cal E}$ is open and convex, and
\be 
\overline{{\mathscr M}_0\setminus{\cal E}}={\mathscr M}.
\de 
\end{corollary}
\begin{proof}
Note that (\ref{P_0}) has a unique solution if and only if ${\mathscr M}_0 \setminus {\cal E} = \emptyset$, we therefore assume that nonunique solutions exist.
We then take some arbitrary non-Dirac measure $\mu$ in ${\mathscr M}_0 \setminus {\cal E}$ so that a corresponding $\gamma_\mu$ exists such that ${\mathscr M}(\gamma_\mu) = \{ \mu \}$ by Theorem \ref{M0M}. Note that
$\mu$ by definition cannot belong to ${\cal E}$.
By Definition \ref{SPdef5}, since ${\mathscr M}(\gamma_\mu)=\{\mu\}$ and $\mu$ is non-Dirac, the corresponding regularization exhibits strong $\SP$ with selected measure $\mu$. 
\end{proof}

Of independent interest, the following result highlights the inherent flexibility of $\SP$, showing that a single regularization can indeed realize the full spectrum of statistical behaviors.
\begin{proposition}\label{CHOCBAR2}
There exists a regularization $\gamma_{\rm un} \in \Gamma$ such that 
\be 
{\mathscr M}(\gamma_{\rm un}) = {\mathscr M}_0.
\de 
\end{proposition}\noindent
\textbf{Proof}. See Appendix~\ref{Chocbar2proof}.

\subsection{Universality classes}\label{univc}
Our set-theoretic description gives a natural and unambiguous way to define \emph{universality classes}. They gather all regularizations with the same statistical behavior in the inviscid limit. In Part \ref{part_2}, we will show that these are in fact genuine dynamical-system basins of attraction. It is therefore only in Part \ref{part_2} that one can develop some intuition for how large such classes can be; see Sections \ref{par1_4_3_2} and \ref{RGreg}.

\begin{definition}[Universality classes]\label{univdef}
Define the equivalence relation $\sim$ on $\Gamma$ by
$$
\gamma_1 \sim \gamma_2
\quad \Longleftrightarrow \quad
\mathscr{M}(\gamma_1)=\mathscr{M}(\gamma_2).
$$
For each $\gamma\in\Gamma$, the equivalence class $[\gamma]=
\{ \widetilde \gamma \in \Gamma: \widetilde\gamma \sim \gamma \}$
is called the universality class of $\gamma$.
The family of universality classes is the quotient set
$$
\Gamma/\!\sim \;=\; \{[\gamma]:\gamma\in\Gamma\},
$$
which yields the partition
$$
\Gamma=\bigsqcup_{[\gamma]\in \Gamma/\!\sim} [\gamma].
$$
\end{definition}
\begin{remark}
It is natural to expect that universality classes can be described through suitable compact connected subsets of $\mathscr{M}$. Indeed, for every $\gamma\in\Gamma$, the set
$
\mathscr{M}(\gamma)\subset \mathscr{M}
$
is compact and connected, so each equivalence class $[\gamma]$ is uniquely associated with such a subset. Thus the quotient $\Gamma/\!\sim$ may be regarded as indexed by the compact connected subsets of $\mathscr{M}$ that are realized as $\mathscr{M}(\gamma)$ for some $\gamma\in\Gamma$.
Theorem \ref{M0M} is a particular instance of this for singleton sets $\{\mu \}$. 
\end{remark}
\bigskip 
\begin{remark}
A useful consequence of Theorem~\ref{M0M} concerns the relation with
$\SP_{\rm turb}$. Let $\nu_\epsilon$ be the probability measure defined in
\eqref{spspsp}, and assume that $\nu_\epsilon\rightharpoonup\nu$. Then
$\nu\in\mathscr M_0$. Since Theorem~\ref{M0M} proves the identity
$\mathscr M=\mathscr M_0$, this limiting law is not only attainable by the
general turbulence-inspired construction: it can also be realized by a
one-parameter regularization curve. More precisely, there exists
$\gamma_{\rm turb}\in\Gamma$ such that
$$
(\gamma_{\rm turb})_\#{\rm Leb}_\epsilon \rightharpoonup \nu .
$$
This is a nontrivial consequence of the attainability theorem. Indeed,
$\SP_{\rm turb}$ is formulated through random perturbations of the initial
condition, whereas $\Gamma$ contains only one-parameter regularization curves
sampled by the ambient measure ${\rm Leb}_\epsilon$. The theorem therefore
shows that, at the level of selected limiting laws, every convergent
$\SP_{\rm turb}$ statistics has a one-parameter representative, and hence a
representative universality class in the sense of Definition~\ref{univdef}. If
$\nu$ is non-Dirac, this representative realizes Strong $\SP$.
\end{remark}

\section{A necessary condition for observing $\SP$: singularities in the inviscid system}\label{CNforSP}

Corollary~\ref{MND} shows that, once the inviscid problem \eqref{P_0} has nonunique solutions, one can construct a regularization \eqref{P_eps} exhibiting strong $\SP$. The next question is to understand the origin of this nonuniqueness. In this section we address it from a geometric, dynamical-systems viewpoint. The guiding picture is that uniqueness can be lost when the inviscid flow reaches, in finite time, a region where the vector field is too singular to determine a unique trajectory. We make this picture precise by deriving a necessary condition in terms of singular sets detected by Dini-type directional growth.

{The definition of such singular sets requires some care. A typical mechanism is an isolated non-Lipschitz singularity $x^\star\in\mathbb R^n$ with $f(x^\star)=0$, although this condition is not sufficient by itself; an example of this type, inspired by \cite{Drivas21,Drivas24}, is discussed in Appendix~\ref{Ex3}.} In higher dimension, however, loss of uniqueness need not be tied to an equilibrium point. There are also situations where $f(x^\star)\neq0$, but trajectories passing through $x^\star$ still fail to be uniquely determined. This second mechanism is also illustrated in Appendix~\ref{Ex3}. In both cases, the examples are supplemented by regularizations exhibiting strong $\SP$.

In what follows, we assume that $f$ is at least continuous, so that existence of solutions is guaranteed by Peano's theorem.
Numerous theorems on uniqueness are available for general nonautonomous systems of the form $\dot x = f(x,t),~x(t_0) = x_0$ (see, in chronological order \cite{Walter1970,Agarwal_Lak,Hartman2002,Bahouri2011}). These results 
typically take the form:
{\it if condition $(A)$ holds then the solution to the Cauchy problem is unique}.
Consequently by contraposition, there are many possible necessary conditions of the form
$\neg{(A)}$ for nonuniqueness. In order to obtain some sharper condition, one must look at 
a "disjunction" of different uniqueness theorems.

The one-dimensional autonomous case is much better understood.
In that setting, nonuniqueness can only occur through points in the zero set $\{f=0\}$.
In the situation relevant here, the singular set reduces locally to an isolated critical point $x^\star$.
Moreover, near $x^\star$, the vector field must satisfy a sufficiently strong one-sided growth condition on at least one side; in particular, local Lipschitz continuity at $x^\star$ rules out nonuniqueness.
The sharp condition is the finiteness of at least one of the integrals
$$
\int_{x^\star}^{x^\star+\eta}\frac{ds}{f(s)}
\qquad\text{or}\qquad
\int_{x^\star-\eta}^{x^\star}\frac{ds}{|f(s)|},
$$
with the appropriate sign assumptions on $f$.
As for the initial condition, the one-dimensional setting is highly constrained:
either $x(0)=x^\star$, or the trajectory must reach $x^\star$ in finite time from one side.
Typical examples yielding nonunique solutions are
$\dot x=\operatorname{sign}(x)|x|^\alpha$ and $\dot x=\pm |x|^\alpha$ with $\alpha\in(0,1)$.
By contrast, systems such as $\dot x=-\operatorname{sign}(x)|x|^\alpha$ or $\dot x=|x|^\alpha+c, c \neq 0$ do not exhibit nonuniqueness.


The situation in higher dimension is much less favorable, in particular 
$f(x^\star) = 0$ is no longer necessary which is the main reason this problem becomes nontrivial. 
The aim here is thus to characterize regions where 1) $f$ has some expanding behavior
2) $f$ is not Lipschitz. 

We consider here a similar idea to that of Dini derivatives (see Appendix \ref{DiniApp}). 
They correspond to a generalized notion of derivatives in contexts where $f$ is not even continuous. Dini derivatives share many desirable properties with classical derivatives.
For instance, being Lipschitz continuous is equivalent to having
finite Dini derivatives. They 
therefore provide a natural and practical way to express non-Lipschitz behavior.
We slightly extend the notion to the Osgood property.
To simplify the presentation, we will consider
$$
H = \mathbb{R}^n.
$$
Let us introduce the following definitions:
\begin{definition}\label{DiniME}
Let $f: \mathbb{R}^n \to \mathbb{R}^n$, 
Let $v \in \mathbb{S}^{n-1}$, then
\be 
\Lambda_\Omega^+(x,v) := \limsup_{t \to 0^+} 
\left \langle \frac{f(x+t v)-f(x)}{\Omega(t)},v
\right\rangle .
\de 
where $\Omega$ is a modulus of continuity: $\Omega: \mathbb{R}^+ \to \mathbb{R}^+, \Omega(0) = 0$ and $\Int_{0^+} \frac{dz}{\Omega(z)} = \infty$.

We also introduce a notion of stable finite-time set:
Let the flow map $\Phi: (t,x_0) \mapsto \Phi_t(x_0)$, where
$\Phi_t(x_0)$ is the solution at time $t$ of 
the Cauchy problem $\dot x = f(x), x(0)=x_0$, and consider some arbitrary set $\Gamma \subset \mathbb{R}^n$
\begin{equation}\label{Wm_G}
{\cal W}^-(\Gamma) \equiv  \bigcup_{\Phi} \left\{ x \in \mathbb{R}^n
 \left| \right. \exists t^\star < \infty,~t^\star \geq  0, \lim_{t \to +t^\star} {\rm dist}(\Phi_t(x),\Gamma) = 0 \right\} 
\end{equation}
\end{definition}
Note that $t^\star = t^\star(x,\Phi)$, in particular for $x \in \Gamma$, $t^\star(x)=0$ and thus $\Gamma \subset {\cal W}^-(\Gamma)$.
From this definition, one can derive a useful necessary condition for breaking uniqueness
and thereby guaranteeing $\SP$:
\begin{theorem}\label{CN}
Assume that (\ref{P_0}) has nonunique solutions, then one 
can find some $\Omega$ and a nonempty set $\Gamma^\star \subset \mathbb{R}^n$ such that
$x_0 \in {\cal W}^-(\Gamma^\star)$ 
where
$$
{\displaystyle \Gamma^\star = \big\{ x \in \mathbb{R}^n ~|~ \exists v, ||v||=1, 
 ~\Lambda^+_\Omega(x,v)= +\infty \big\}.}
$$
\end{theorem}
\noindent {\bf Proof}: see Appendix \ref{DiniApp}.\\
The idea is to replace the classical notion of Jacobian of $f$
by some weaker notion involving Dini (directional) derivatives 
(say with $\Omega(z)=z$). The unit vectors $v$ indeed play the role of an eigenvector
and $\Gamma^\star$ simply detects regions in phase space where the eigenvalues blow up to $+\infty$. 
In addition, one can also characterize for free the set of initial conditions giving
spontaneous stochasticity by a straightforward generalization of a stable set.
Note that such a result is not sharp and cannot be sufficient. However, there is still room to improve this condition. 
\bigskip 

In dimension one, Theorem \ref{CN} should be read as a general necessary condition rather than as a sharp scalar criterion. For an autonomous scalar equation $\dot x=f(x)$, loss of uniqueness is tied to the presence of a point $x^\star$ such that $f(x^\star)=0$ and where the vector field fails an Osgood-type condition. Thus, in dimension one, the singular point responsible for nonuniqueness is also a critical point of the vector field. Theorem \ref{CN} detects the non-Osgood defect through the set $\Gamma^\star$, but it does not separately encode the scalar condition $f(x^\star)=0$. This is consistent with its purpose: the theorem is designed to capture more general mechanisms, in particular higher-dimensional situations where nonuniqueness may arise from directional non-Osgood behavior and need not reduce to the presence of an equilibrium point.
\bigskip

We now state an informal fixed-time consequence of Theorem \ref{CN}. Fix $t>0$ and let $\Gamma^\star$ be the singular set defined in Theorem \ref{CN}. In analogy with \eqref{Wm_G}, we define the time-$t$ stable set
$$
{\cal W}^-_t(\Gamma^\star)
\equiv
\bigcup_{\Phi}
\left\{
x\in\mathbb{R}^d
\left|
\right.
\exists t^\star\in[0,t),~
\lim_{s\to t^\star}
{\rm dist}(\Phi_s(x),\Gamma^\star)=0
\right\}.
$$

Let ${\cal D}_t$ denote the set of initial data for which the inviscid flow map $\Phi_t$ is well-defined as a single-valued map. Since Theorem \ref{CN} gives a necessary condition for loss of uniqueness, one obtains the informal inclusion
$$
\mathbb{R}^d\setminus{\cal D}_t
\subset
{\cal W}^-_t(\Gamma^\star).
$$
Equivalently,
\be \label{domainflowmapt}
\mathbb{R}^d\setminus{\cal W}^-_t(\Gamma^\star)
\subset
{\cal D}_t.
\de 
Thus an initial condition outside ${\cal W}^-_t(\Gamma^\star)$ cannot lose uniqueness before time $t$. On the other hand, belonging to ${\cal W}^-_t(\Gamma^\star)$ is only a necessary condition for loss of uniqueness: such an initial condition may still belong to ${\cal D}_t$.

\section{A simple illustrative example with explicit expressions}\label{exAmbrosio}
The aim of this section is to illustrate the abstract results of Part~I
through a single one-dimensional example. Further examples are collected in
Appendix~\ref{appendix_examples}. The fact that the main phenomena already
appear in systems with only one degree of freedom is itself noteworthy.
\subsection{Description of the system}
The example is inspired by \cite{ambrosio2004transport,Flandoli2009}. The inviscid system \eqref{P_0} is
$$
\dot x=\sqrt{|x|},
\qquad
x(0)=x_0<0.
$$
It has infinitely many solutions. The system is well-posed until time
$t^\star=2\sqrt{-x_0}$, we call it the \emph{pre-blowup branch}.
It then reaches the H\"older singularity at $x=0$. After that, the solution may remain at $0$ for an arbitrary waiting time $T\ge0$ before taking off again according to
$$
x(t)=\frac14\,(t-T-t^\star)^2,
\quad
t\ge t^\star+T, \quad t^\star=t^\star(x_0) = 2 \sqrt{-x_0}.
$$
We exploit the fact that $T$ is a free variable and consider the following regularizations \eqref{P_eps} adapted from \cite{ambrosio2004transport}.
\begin{equation}\label{regamb}
\tag{\ensuremath{{\rm Amb}_{T_\epsilon}}}
f(x,\epsilon)
=
\left\{
\begin{array}{ll}
\sqrt{|x|},
&
x\in(-\infty,-\epsilon^2],
\\[0.4em]
\epsilon,
&
x\in[-\epsilon^2,\epsilon T_\epsilon-\epsilon^2],
\\[0.4em]
\sqrt{x-\epsilon T_\epsilon+2\epsilon^2},
&
x\in(\epsilon T_\epsilon-\epsilon^2,+\infty).
\end{array}
\right.
\end{equation}
In order to have $f \in \V$; see \eqref{H0}, we must impose
$$
\epsilon T_\epsilon = o(1).
$$

This expression has the effect to regularize the singularity smoothly. 
The solution can be found explicitly, we give its full expression in Appendix \ref{flowamb}. Indeed, one does not need to perform numerical simulations for this system.

We first show how the regularization looks like visually. We also show how the
flow map behaves: in particular, one observes that the flow map of
the inviscid system is undefined on the interval $x \in (-\frac{t^2}{4},0)$. Regularizations induce well-defined flow maps in particular selecting compatible dynamics in the inviscid limit. We show both in Fig. \ref{Vec_Flow}.

We first briefly illustrate Theorem \ref{CN} and \eqref{domainflowmapt} on the inviscid system $\dot x=\sqrt{|x|}$. In this case the singular set is $\Gamma^\star=\{0\}$, and an initial point $x_0<0$ reaches $0$ at time $t^\star(x_0)=2\sqrt{-x_0}$. Hence, for fixed $t>0$, one has
$$
\Gamma^\star = \{0\},\qquad
{\cal W}^-_t = \left( -\frac{t^2}{4},0 \right],
\qquad
{\cal D}_t = \mathbb{R} \setminus {\cal W}^-_t.
$$
The endpoint $x_0=-t^2/4$ reaches the singular point exactly at time $t$, and therefore still belongs to the single-valued domain of the inviscid flow map ${\cal D}_t$.

\begin{figure}[H]
\hspace*{-0.5cm}
\centering
\includegraphics[scale=0.6]{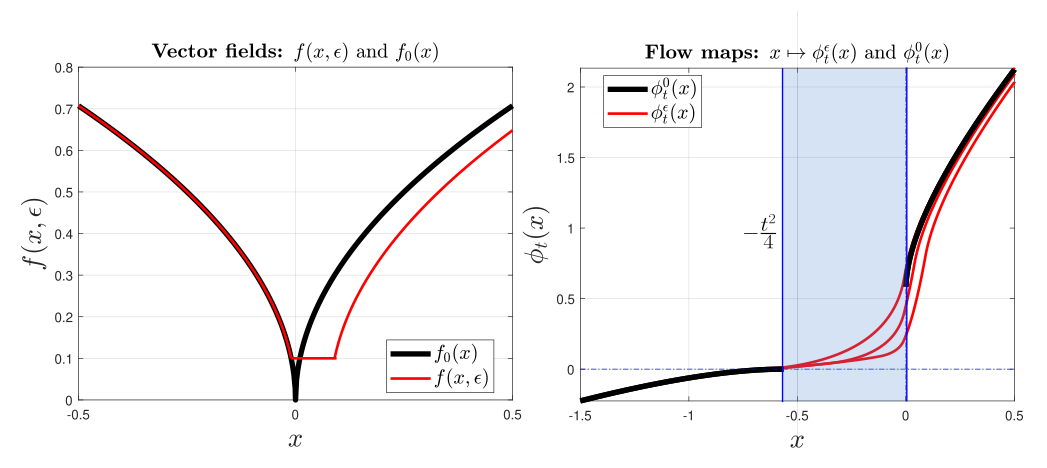}
\caption{Vector fields (left panel) and flow maps at $t=1.5$ (right panel) for $\epsilon=0.1$. the regularized vector field is shown for $T=1$ only, whereas the regularized flow map are shown with three different values $T=0.1,0.5,1$.
The inviscid flow map (black curve) is undefined on the interval $(-\frac{t^2}{4},0]$.}
\label{Vec_Flow}
\end{figure}

\subsection{Spontaneous stochasticity (weak and strong); Definition \ref{SPdef}}

Since the system is fully explicit, we can characterize when it is $\SP$. We focus on the Cauchy problem given a fixed initial condition $x_0 < 0$. We
assume that the system is observed after the inviscid hitting time
$t^\star=2\sqrt{-x_0}$, namely $t-t^\star>0$. The key choice is to use a
dependence of $T_\epsilon$ on $\epsilon$ such that $T_\epsilon$ has no limit as
$\epsilon\to0$, and we set
$$
T^-:=\liminf_{\epsilon\to0}T_\epsilon
<
T^+:=\limsup_{\epsilon\to0}T_\epsilon .
$$
Using the explicit expression of $\phi_t^\epsilon(x_0)$ given in
Appendix~\ref{flowamb}, in the region where the trajectory has already left
the plateau one has
\be \label{Gamma_T}
\gamma(\epsilon)
=
\phi_t^\epsilon(x_0)
=
\epsilon T_\epsilon-2\epsilon^2
+
\left(
\frac{t-t^\star-T_\epsilon}{2}+2\epsilon
\right)^2 .
\de
We define $\Gamma$ by $\Gamma(\epsilon,T_\epsilon)=\gamma(\epsilon)$, namely
$\Gamma(\epsilon,T)=\epsilon T-2\epsilon^2+
\left(\frac{t-t^\star-T}{2}+2\epsilon\right)^2$.

Then $\Gamma(\epsilon,T)=\Gamma_0(T)+o(1)$ as $\epsilon\to0$, with
$\Gamma_0(T)=\frac14(t-t^\star-T)^2$, as long as the limiting curve does not
hit zero. Moreover,
$$
\phi^-:=\liminf_{\epsilon\to0}\phi_t^\epsilon(x_0)
=
\frac14\left(t-t^\star-T^+\right)_+^2,
\qquad
\phi^+:=\limsup_{\epsilon\to0}\phi_t^\epsilon(x_0)
=
\frac14\left(t-t^\star-T^-\right)_+^2.
$$
Thus
$$
\phi^-<\phi^+
\quad\Longleftrightarrow\quad
T^-<t-t^\star .
$$
Therefore, under $T^-<T^+$, one has $\SP$ relative to
$(t,x_0,{\rm Leb}_\epsilon)$ whenever $T^-<t-t^\star$. This follows from
Proposition~\ref{SPDEFS} applied with ${\cal O}={\rm Id}$.

\paragraph{Strong $\SP$.}
We choose $T_\epsilon$ such that
$(T_\epsilon)_\#{\rm Leb}_\epsilon\rightharpoonup{\cal T}$; for instance,
$T_\epsilon=1+a\sin(1/\epsilon)$ with $0<a<1$. Then ${\cal T}$ is the shifted
and rescaled arcsine law on $[1-a,1+a]$, namely
${\cal T}(dT)=
{\bf 1}_{(1-a,1+a)}(T)\,dT/
\left(\pi\sqrt{a^2-(T-1)^2}\right)$.

If $t-t^\star>1+a$, the limiting curve does not hit zero on the support of
${\cal T}$, and we obtain strong $\SP$ with
$$
\gamma_\#{\rm Leb}_\epsilon
\rightharpoonup
\mu
=
(\Gamma_0)_\#{\cal T}.
$$
In this case one obtains the explicit probability measure
$$
\mu(dx)
=
\frac{dx}
{\pi\sqrt{x}\,
\sqrt{
a^2-\left(t-t^\star-1-2\sqrt{x}\right)^2
}},~x \in 
\left(
\frac14(t-t^\star-1-a)^2,
\frac14(t-t^\star-1+a)^2
\right).
$$
If the condition $t-t^\star>1+a$ is not satisfied, the limiting curve hits
zero and the limiting law has a different expression, involving the truncated
curve $\frac14(t-t^\star-T)_+^2$ and possibly an atomic contribution at zero.

\begin{figure}[H]
\centering
\includegraphics[scale=0.45]{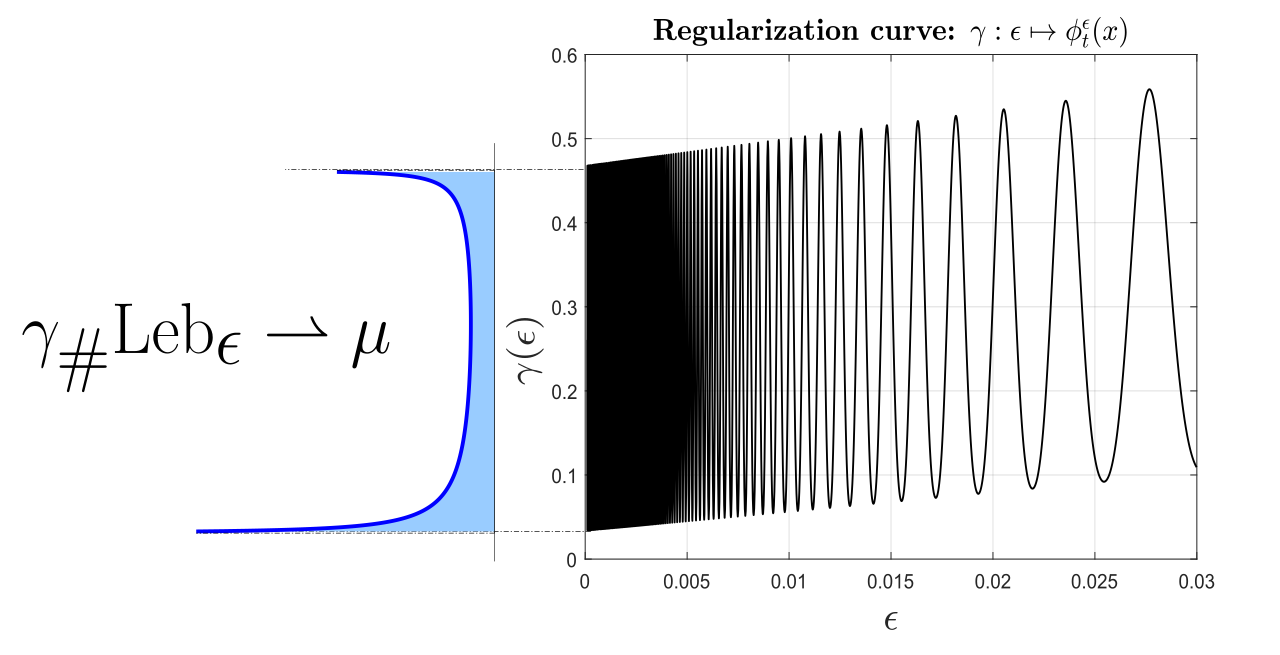}%
\caption{Regularization curve $\gamma$ for $t=2.5$, $x_0=-0.1$, $a=0.5$, using $T_\epsilon = 1 + a \sin 1/\epsilon$, together with the distribution of states $\mu$ in the limit: system \eqref{regamb} is Strong-$\SP$ relative to $(t,x_0,{\rm Leb}_\epsilon)$. }
\label{Amb_gamma}
\end{figure}

\paragraph{Weak $\SP$.}

It is also easy to obtain a situation where one has Weak-$\SP$. Consider
$$
T_\epsilon=1+a\sin(b\log\epsilon),
\qquad 0<a<1,\quad b\neq0.
$$
Let us investigate
$I_\epsilon:=\langle (s\mapsto T_s)_\#{\rm Leb}_\epsilon,F\rangle$. After the
change of variables $s=\epsilon r$, one obtains
$$
I_\epsilon
=
\int_0^1
F\left(1+a\sin(b\log\epsilon+b\log r)\right)\,dr .
$$
We now choose subsequences $\epsilon_n^{(\theta)}$ such that
$b\log\epsilon_n^{(\theta)}$ converges to $\theta$ modulo $2\pi$. Then
$$
(s\mapsto T_s)_\#{\rm Leb}_{\epsilon_n^{(\theta)}}
\rightharpoonup
{\cal T}_\theta,
$$
where
$$
{\cal T}_\theta
=
\left(
r\mapsto 1+a\sin(\theta+b\log r)
\right)_\#{\rm Leb}_1 .
$$
Thus different subsequences may select different limiting laws
${\cal T}_\theta$. Following the same argument as in the strong $\SP$ case,
one obtains the corresponding family of inviscid measures
$$
\mu_\theta
=
(\Gamma_0)_\#{\cal T}_\theta .
$$
Hence the regularization curve has subsequential limiting statistics, and
Weak-$\SP$ holds.
From this, we also deduce that the set $\mathscr{M}(\gamma)$ defined in \eqref{Mgamma_def} is
$$
\mathscr{M}(\gamma) = \{ \mu~:~\exists \epsilon_n \to 0 ~\mbox{such that}~
\gamma_\# {\rm Leb}_{\epsilon_n} \rightharpoonup \mu \} = 
\{ \mu_\theta,~\theta \in [0,2\pi) \}.
$$
This set is not convex (see remark in \ref{step3}).
We illustrate these results in Fig.~\ref{Amb_gamma_weak}. In particular, we
plot the variance of the regularization curve as a simple diagnostic for
Weak-$\SP$. Note, however, that this criterion is only sufficient and not
necessary: the second moment may converge while higher-order moments do not.

\begin{figure}[H]
\centering
\includegraphics[scale=0.5]{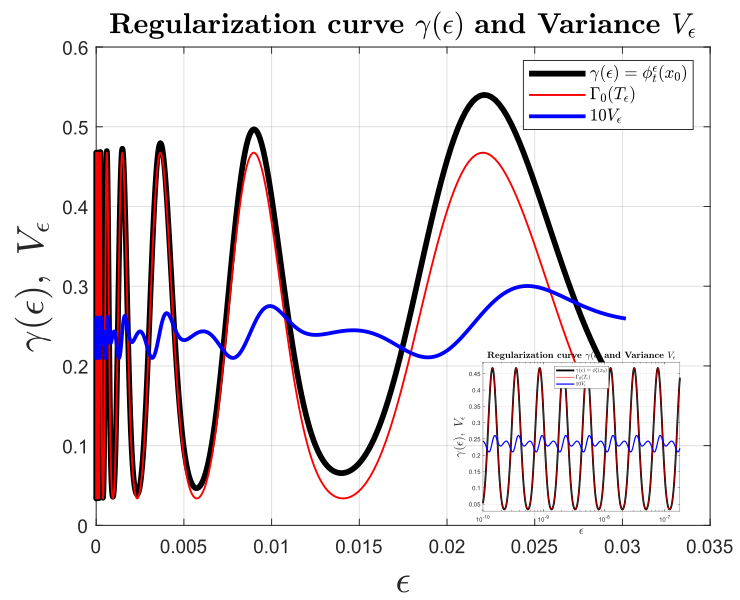}%
\caption{Regularization curve $\gamma$ for $t=2.5$, $x_0=-0.1$, $a=0.5$, $b = 7$, using $T_\epsilon = 1 + a \sin b \log \epsilon$, together with the variance $V_\epsilon$ of $\gamma$. 
The inset provides a logarithmic-scale view in $\epsilon$ of
the same curves suggesting log-periodic behavior.
System \eqref{regamb} is Weak-$\SP$ relative to $(t,x_0,{\rm Leb}_\epsilon)$. }
\label{Amb_gamma_weak}
\end{figure}

This example illustrates that the ambient measure must be adapted to the scale on which the realization map oscillates. With the flat measure ${\rm Leb}_\epsilon$, the map is sampled on the wrong scale: its phase does not sweep through enough oscillations inside the Lebesgue window to produce an averaged limiting law. The resulting statistics retain a dependence on the phase of $\log\epsilon$, in a way reminiscent of a stroboscopic effect.

\subsection{Randomization of $T$: illustration of Definition \ref{genSPdef}}

Instead of taking a deterministic regularization with $T=T_\epsilon$, one can
randomize the variable $T$. This fits Definition~\ref{genSPdef} as follows.
The vector field \eqref{regamb} is now written as
\be \label{randTamb}
\tag{\ensuremath{{\rm Amb}_{\theta}}}
f(x,\theta) = 
\left\{
\begin{array}{ll}
\sqrt{|x|},
&
x\in(-\infty,-\epsilon^2],
\\[0.4em]
\epsilon,
&
x\in[-\epsilon^2,\chi-\epsilon^2],
\\[0.4em]
\sqrt{x-\chi+2\epsilon^2},
&
x\in(\chi-\epsilon^2,+\infty),
\end{array}
\right.
\qquad
\theta=(\epsilon,\chi).
\de
One checks that $\|f(\cdot,\theta)-f_0\|_\infty\to0$ as
$\|\theta\|\to0$. We thus take
$$
\Theta=\mathbb R_0^+\times\mathbb R_0^+,
\qquad
\theta=(\epsilon,\chi),
\qquad
\mathbb P_\epsilon(ds,d\chi)
=
{\rm Leb}_\epsilon(ds)\,\rho_s(d\chi),
$$
where
$$
\rho_s=(T\mapsto sT)_\#\rho.
$$
Equivalently, if $\chi=sT$ and $T$ is distributed according to $\rho$, then
$\rho_s$ is the induced law of $\chi$. More explicitly, for every bounded
continuous function $F$,
$
\langle \rho_s,F\rangle
=
\int_{\mathbb R_0^+} F(\chi)\,\rho_s(d\chi)
=
\int_{\mathbb R_0^+} F(sT)\,\rho(dT).
$
\\

Figure~\ref{AmbrandT} shows the randomized case obtained from
\be \label{rho3}
\rho(dT)
=
\frac{2}{3}\,{\bf 1}_{\rm W}(T)\,dT,
\qquad
{\rm W}
=
\left[1,\frac32\right]\cup\left[2,\frac52\right]\cup\left[3,\frac72\right].
\de
Assuming $t-t^\star>\sup{\rm W}=7/2$, the limiting statistics are absolutely
continuous and given by
$$
\mu(dy)
=
\frac23\,
\frac{{\bf 1}_{\rm W}\!\left(t-t^\star-2\sqrt{y}\right)}
{\sqrt{y}}\,dy.
$$
The same figure also shows, in red, the deterministic extremal solution
corresponding to $\rho(dT)=\delta_0(dT)$; see
Definition~\ref{extremalsol}. This second case does not produce spontaneous
stochasticity, since it selects a unique trajectory.

\begin{figure}[H]
\centering
\includegraphics[scale=0.55]{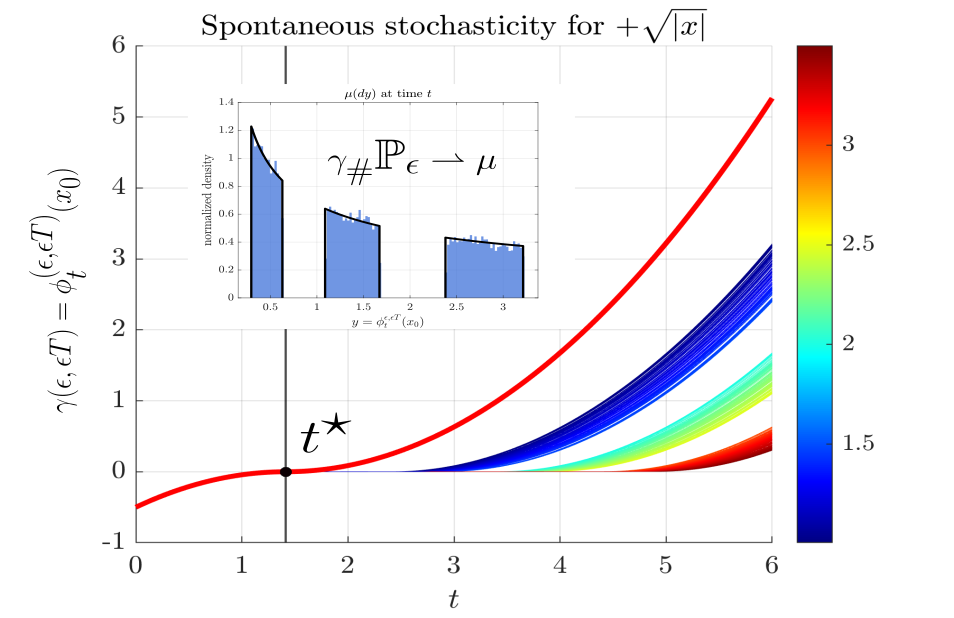}%
\caption{The system \eqref{randTamb} is $\SP$ relative to
$(t,x_0,\mathbb P_\epsilon)$ for the law \eqref{rho3}. The three
bands correspond to the three waiting-time intervals in ${\rm W}$, while the
red curve is the extremal solution obtained from $\rho=\delta_0$. Here
$\epsilon=10^{-3}$, $x_0=-0.5$, and $O(1000)$ i.i.d. realizations are used.}
\label{AmbrandT}
\end{figure}

Theorem~\ref{M0M} states that it is possible to design a regularization of
$\sqrt{|x|}$ that converges to any target probability distribution in
${\mathscr M}_0={\cal P}({\mathscr S}_0)$, where here
$$
{\mathscr S}_0
=
\left[0,\frac14(t-t^\star)^2\right],
\qquad
t\geq t^\star.
$$
In particular, our random-$T$ regularized system has a representative
one-parameter regularized system, not necessarily of the form \eqref{regamb}.

\subsection{Failure of TBM criterion.}
Note, somewhat surprisingly, that if the initial condition is kept uniformly
away from the singularity, then the regularized flow maps remain uniformly
Lipschitz with respect to the initial condition. More precisely, let $x_0<0$ and
choose an interval $I\subset(-\infty,0)$ with $x_0\in I$. Then, for every
fixed $t>t^\star$, the explicit formula in Appendix~\ref{flowamb} yields a
constant $C_{t,I}>0$, independent of $\epsilon$, such that, for all
$\epsilon>0$ small enough,
$$
\left|
\phi_t^\epsilon(x)-\phi_t^\epsilon(y)
\right|
\leq
C_{t,I}|x-y|,
\qquad
x,y\in I.
$$
In particular, if $x_0+\operatorname{supp}\rho_\epsilon\subset I$ for all
$\epsilon>0$ small enough and $\rho_\epsilon\rightharpoonup\delta_0$, then
TBM cannot hold.

\part{\Large Spontaneous stochasticity with a semigroup structure}\label{part_2}

Part~\ref{part_1} has left us with a structural fact: as soon as the inviscid
problem \eqref{P_0} admits nonunique solutions, essentially any statistics can
be selected in the inviscid limit. Recall that a \emph{regularization curve}
$\gamma$, see \eqref{gamma}, records the dependence of the flow map on the
regularization level,
$$
\gamma:\mathbb R^+\to H,
\qquad
\tau\mapsto\phi_t^\tau(x_0),
$$
the physical time $t$ and the initial condition $x_0$ being fixed.
Theorem~\ref{M0M} then asserts that the set of limiting laws attainable within
the class $\Gamma$ of admissible curves, see \eqref{Gamma}, is maximal:
$$
\mathscr M
=\bigcup_{\gamma\in\Gamma}\mathscr M(\gamma)
=\mathscr M_0
={\cal P}(\mathscr S_0).
$$

This flexibility, however, comes at a price: the class $\Gamma$ is so large
that it carries no information about how a particular law is actually
selected. In particular, the universality classes introduced in
Section~\ref{univc} remain, at this level of generality, purely
set-theoretic. The purpose of Part~\ref{part_2} is to restore structure by
restricting attention to regularizations that are themselves dynamical
objects.

The key step is to consider curves that arise as orbits of a semigroup
$({\cal R}_\tau)_{\tau\ge0}$ acting on $\rgh$, namely
$$
\gamma(\tau)=\gamma(\tau;y)={\cal R}_\tau y,
\qquad y\in\rgh .
$$
Such a curve is entirely specified by the pair $({\cal R},y)$, so that the
limit sets $\mathscr M(\gamma)$ of Definition~\ref{SPdef5} can be written in
the parametrized form $\mathscr M_{\cal R}(y)$. For simplicity, we restrict
throughout to the ambient measure ${\rm Leb}_\tau$; see \eqref{Leb}. Allowing
both the semigroup ${\cal R}$ and the initial state $y$ to vary yields the
dynamically generated subset
\be \label{MRG}
\mathscr{M}^{\rm (RG)}
:=\bigcup_{\cal R}\,\bigcup_{y\in\rgh}\mathscr M_{\cal R}(y)
\;\subset\;
\mathscr M=\mathscr M_0 .
\de
The inviscid statistics are now indexed by the semigroups ${\cal R}$ (or their
vector-field generators) and the initial states $y\in\rgh$, rather than by
arbitrary curves $\gamma\in\Gamma$. One expects the inclusion in \eqref{MRG}
to be strict in general: imposing a semigroup structure constrains the
attainable statistics. This loss of flexibility is precisely the point.
Theorem~\ref{RGattractors} and Corollary~\ref{coroDS} show that, for a fixed
semigroup ${\cal R}$, the family $\bigcup_{y\in\rgh}\mathscr M_{\cal R}(y)$ is
governed by the ergodic theory of ${\cal R}$: the limiting laws arise as
invariant measures of the RG dynamics, organized into statistical attractors,
and the universality classes of Section~\ref{univc} thereby acquire a genuine
dynamical meaning as basins of attraction; see also~\cite{MMRG26}.

Part~\ref{part_2} is organized as follows. Section~\ref{par1_4_3_1} sets up
the semigroup formalism: RG flows, statistical attractors, an elementary
illustrative example, and a comparison with the Feigenbaum--Mailybaev RG
approach. Section~\ref{par1_4_3_2} develops the associated Perron--Frobenius
formulation and proves the structural results mentioned above.
Section~\ref{RGreg} applies the framework to inviscid systems with an isolated
singularity, where the selection mechanism reduces to an exit problem near the
singularity, coupled to the RG dynamics. Finally, Section~\ref{par1_4_3_3}
treats regularization curves that are not semigroup orbits: lifting the
problem to the space of curves via the Bebutov flow restores a semigroup
structure in full generality and yields Theorem~\ref{thm:M0Mupup}, an abstract
statistical-attractor counterpart of Theorem~\ref{M0M}.

\section{Definitions}\label{par1_4_3_1}
The formalism is intentionally broad. Its purpose is not to prescribe a unique renormalization-group construction, but to isolate the minimal ingredients needed later on: a semigroup structure in the regularization parameter, a compatible family of regularized vector fields, and the associated notion of statistical attractor.

\subsection{Definition of the semigroup}
We recall that we start from an inviscid system $f_0$ at $(t,x_0)$ fixed and defined
by
$$
\dot x = f_0(x),~x(0) = x_0 \in H, ~t \in [0,T]
$$
We again study the inviscid limit $\tau \to \infty$ for a family of regularized problems of the form 
$$
\dot x = f(x,\theta),~x(0) = x_0, \qquad \theta = (\tau,y) \in \Theta,
$$
where the regularization space $\Theta$ has the product form
$$
\Theta=\mathbb R_0^+\times\rgh.
$$
The regularized vector field $f$ must be compatible
with the inviscid system \eqref{P_0}:
$$
\lim_{\tau\to\infty}
\|f(\cdot,\tau,y)-f_0(\cdot)\|_\infty
=0,
\qquad
\forall y\in\rgh.
$$
Our aim is to restrict the regularized vector fields $f$ to the case where the
dependence on the scale variable $\tau$ is itself generated by an autonomous
dynamical system. It yields the coupled system:
\begin{subequations}\label{RGODE}
\begin{empheq}[left=\empheqlbrace]{alignat=5}
&\dot x\hphantom{(\tau;y)}
&=&~~ f(x,\tau,y),
&\qquad& x(0)
&&= x_0,
\label{RGODE:x}
\\
&\frac{d\gamma}{d\tau}\hphantom{(\tau;y)}
&=&~~ \mathrm G(\gamma),
&\qquad& \gamma(0)
&&= y,
\label{RGODE:gamma}
\\
&\gamma(\tau;y)
&=&~~ \phi_t^{\tau,y}(x_0).
&\qquad&
&&
\label{RGODE:coupling}
\end{empheq}
\end{subequations}
The first line, \eqref{RGODE:x}, is the regularized Cauchy problem; the dot refers to the physical time. The second line,
\eqref{RGODE:gamma}, is an autonomous equation in the scale variable
$\tau$ (the RG time introduced below). The vector field
$\mathrm G$ is assumed globally Lipschitz, and therefore generates a
semigroup. We write $\gamma(\tau;y)$ for the solution of \eqref{RGODE:gamma} at time $\tau$
starting from $y$, so that the semigroup property reads
$$
\gamma(\tau_1+\tau_2;y)
=
\gamma(\tau_1;\gamma(\tau_2;y)).
$$
The third line, \eqref{RGODE:coupling}, is the compatibility condition between
the regularized Cauchy problem and this semigroup orbit. The quantity
$\phi_t^{\tau,y}(x_0)$ denotes the solution of \eqref{RGODE:x} at the fixed
time $t$, with initial condition $x_0$, and with parameter
$(\tau,y)$. Strictly speaking, the notation should keep track of the fixed
pair $(t,x_0)$; we suppress this dependence for ease of reading. 

From now on, we use the dynamical-system notation
\be \label{SRG}
{\cal R}_\tau y := \gamma(\tau;y),
\qquad
{\cal R}_{\tau_1+\tau_2}
=
{\cal R}_{\tau_1}\circ{\cal R}_{\tau_2}.
\de
Thus, in the present setting, a regularization curve is precisely an orbit
$\tau\mapsto{\cal R}_\tau y$ starting from some $y\in\rgh$. The main
object is the asymptotic behavior of this dynamical system as
$\tau\to\infty$, which corresponds here to removing the regularization and
approaching the inviscid system.

\bigskip

\begin{remark}
Although \eqref{RGODE} may look overdetermined, the coupling is imposed only
at the fixed observation $(t,x_0)$. It does not require the whole
regularized flow $(s,x)\mapsto\phi_s^{\tau,y}(x)$ to be generated by the
semigroup; it only requires the endpoint realization
$(\tau,y)\mapsto\phi_t^{\tau,y}(x_0)$ to form a semigroup orbit. Thus the
condition is a slice constraint on the observed inviscid-limit statistics,
rather than a rigidity assumption on the full dynamics.
Explicit
examples are discussed in Section~\ref{RGreg}.
\end{remark}
\bigskip
\begin{remark}
The coupling in \eqref{RGODE} calls for some additional explanations. Two equivalent viewpoints are possible: either the family $f$ is the primitive object and \eqref{RGODE:coupling} is the requirement that its endpoint realization be generated by some Lipschitz $\mathrm G$; or one prescribes the generator $\mathrm G$ and looks for regularizations $f$ realizing it. The second point of view is the one adopted in the explicit constructions of Section~\ref{RGreg}. A similar philosophy underlies the construction of Feigenbaum‑like renormalization groups \cite{AM_NLN25}; see Section~\ref{FMRG}.
\end{remark}

\bigskip

\begin{remark}\label{remark_fmrg}
 One can generalize this formulation by allowing the semigroup to act on an
auxiliary RG space $\mathbb H$, rather than directly on the physical state
space $H$. In that case $y\in\mathbb H$, and the physical endpoint must be
recovered through a map $\Pi:\mathbb H\to H$. The coupling
condition \eqref{RGODE:coupling} is then replaced by
$
\Pi\bigl(\gamma(\tau;y)\bigr)
=
\phi_t^{\tau,y}(x_0).
$
This formulation is more general, since the RG state may contain information
beyond the physical endpoint itself; in the present work we keep the simpler
case $\mathbb H=H$ and $\Pi=\mathrm{Id}$.
\end{remark}
\bigskip

\subsection{RG flow and statistical attractors}

\begin{definition}[RG flow]\label{RGflow}
A family of strongly continuous maps $({\cal R}_\tau)_{\tau \ge 0}$, ${\cal R}_\tau:\rgh \to \rgh$, is called an \emph{RG flow} associated with \eqref{P_0} if it satisfies
\begin{enumerate}
  \item[(i)] {\bf Semigroup}: ${\cal R}_0={\rm Id}$ and $\mathcal{R}_{\tau+\tau'}=\mathcal{R}_\tau\circ\mathcal{R}_{\tau'}$.
  \item[(ii)] {\bf Regularization}: The solution of \eqref{RGODE:x} at time $t$ and initial condition $x_0$ is ${\cal R}_\tau y$. It therefore becomes increasingly singular (less regularized) as $\tau\to\infty$. 
  We will say that the regularization is {\it autonomous} and call \emph{orbit}, the curve $\gamma:\tau \mapsto {\cal R}_\tau y$.
\end{enumerate}
\end{definition}\noindent 
In concrete terms, the point of Definition \ref{RGflow} is the following. Instead of viewing each regularization as an isolated curve $\tau\mapsto\gamma(\tau)$, we organize all such curves as orbits of a single dynamical system $({\cal R}_\tau)$. This gives access to the usual tools of dynamical systems and ergodic theory, and later allows us to relate $\SP$ to invariant sets, attractors, and basin structure in the RG dynamics.

Although $(\mathcal{R}_\tau)$ acts deterministically on individual states $y \in \rgh$, it induces a corresponding evolution on probability measures via its pushforward; see Definitions in Appendix \ref{PFTP}.
It allows one to study the evolution not just of individual regularizations, but of entire ensembles, thus providing a natural statistical framework.
Within this setting, we define:
\begin{definition}[Statistical Attractor]\label{DefstatAtt}
Let $X$ be a Polish space.  Denote by ${\cal P}(X)$
the space of all Borel probability measures on $X$, endowed with the topology of weak convergence.  Suppose
$(S_\tau)_{\tau\ge0}$ is a continuous semigroup on $\mathcal P(X)$.
A set $\mathcal A\subset\mathcal P(X)$ is called a \emph{statistical attractor} for $(S_\tau)$ if:
\begin{enumerate}
\item[(i)] {\bf Compactness}: $\mathcal A$ is nonempty and compact in the weak topology.
  \item[]
  \item[(ii)] {\bf Invariance}: It is \emph{forward-invariant}: 
    $$
      S_\tau(\mathcal A)\;\subset\;\mathcal A
      \quad\forall\,\tau\ge0.
    $$
  \item[(iii)] {\bf attraction}: Its \emph{Birkhoff basin of attraction} satisfies ${\cal B}({\cal A}) \neq \emptyset$, with
    $$
      \mathcal B(\mathcal A)
      \;:=\;
      \bigl\{\nu\in\mathcal P(X)\colon\omega(\nu)\subset\mathcal A\bigr\},~\text{with}~
      \omega(\nu)
      \;:=\;
      \bigcap_{R>0}
      \overline{\Bigl\{
        \tfrac1{R'}\!\int_0^{R'} S_\tau\nu \,d\tau 
        \;\colon\; R'\ge R
      \Bigr\}}
      \;\subset\;\mathcal P(X),
    $$
\end{enumerate}
\end{definition}
Definition \ref{DefstatAtt} differs from the classical definition in~\cite{KARABACAK_ASHWIN_2011} by omitting the minimality assumption, making it strictly weaker than the notion of \emph{minimal attractors} introduced by Ilyashenko~\cite{arnold1999bifurcation}. This broader formulation allows us to interpret emergent statistical laws as attractors of the dynamics, encompassing not only the asymptotic behavior of deterministic regularizations as defined in~\eqref{RGODE}, but also that of random regularizations.

In this work, we will be concerned in particular with the choice\
of the Perron-Frobenius operator: $S_\tau := ({\cal R}_\tau)_\#$ (see 
Appendix \ref{PFTP}).
We give next a simple example of a Perron-Frobenius operator together with its statistical attractor. 
\subsection{A simple example}
Although the Perron--Frobenius operator $({\cal R}_\tau)_\#$ is linear on probability measures, this does not make its dynamics trivial. The linearity is with respect to the superposition of measures; it does not erase the motion induced by the underlying dynamical system. Thus the Perron--Frobenius dynamics should not be understood as producing nonlinear bending in the affine space of measures. Rather, it transports probability mass along trajectories. The following elementary example illustrates that the raw Perron--Frobenius orbit may keep moving, while its Birkhoff averages select invariant statistical laws.

Let
$$
X=\mathbb S^1_a\cup\mathbb S^1_b,
\qquad
\mathbb S^1_a=\mathbb R/\mathbb Z,
\qquad
\mathbb S^1_b=\mathbb R/\mathbb Z.
$$
We write $d\theta_a$ and $d\theta_b$ for the normalized Lebesgue measures on the two circles. Define the dynamics by
$$
{\cal R}_\tau\theta_a=\theta_a+\omega_a\tau
\quad\text{on }\mathbb S^1_a,
\qquad
{\cal R}_\tau\theta_b=\theta_b+\omega_b\tau
\quad\text{on }\mathbb S^1_b.
$$
Let $\mu\in{\cal P}(X)$ be of the form, for $p \in [0,1]$
$$
\mu
= p \mu_a + (1-p) \mu_b,~{\rm with}~\mu_a(d\theta_a) = \rho_a(\theta_a) d\theta_a,~\mu_b(d\theta_b) = \rho_b(\theta_a) d\theta_b,
$$
where $\rho_a$ and $\rho_b$ are probability densities on their respective circles. Then
$$
({\cal R}_\tau)_\#\mu 
= p ~({\cal R}_\tau)_\# \mu_a + (1-p) ~({\cal R}_\tau)_\# \mu_b 
=
p\,\rho_a(\theta_a-\omega_a\tau)\,d\theta_a
+
(1-p)\,\rho_b(\theta_b-\omega_b\tau)\,d\theta_b.
$$
Hence, unless the densities are already constant, the raw Perron--Frobenius orbit $\tau\mapsto({\cal R}_\tau)_\#\mu$ does not converge: the two conditional densities keep rotating on their respective circles.
By contrast, the Birkhoff averages converge weakly. Since
$1/T \int_0^T  \rho_a(\theta_a-\omega_a \tau) d\tau \to \int_{\mathbb{S}_a^1} \rho_a(ds)  = 1$, one has
$$
\frac1T\int_0^T({\cal R}_\tau)_\#\mu\,d\tau
\;\rightharpoonup\;
p\,d\theta_a+(1-p)\,d\theta_b.
$$
Thus the corresponding statistical attractor is the convex segment
$$
{\cal A}
=
\left\{
p\,d\theta_a+(1-p)\,d\theta_b:\ p\in[0,1]
\right\}.
$$
This example separates the two effects: the raw Perron--Frobenius dynamics transports measures nontrivially, despite its linearity, while the Birkhoff average filters out this motion and selects invariant statistical laws.

\subsection{Comparison with Feigenbaum-Mailybaev RG approach}
\label{FMRG}

In recent years, a renormalization-group formalism for shell models has been developed by Mailybaev and collaborators \cite{Maily2012,Maily2016,AM_Ra23,AM_Rb23,AM_NLN25}; we refer to this approach as FMRG. As in the present framework, a level-$\tau$ regularization is mapped to a level-$\tau'\geq\tau$ regularization by a semigroup operator ${\cal R}$. There are, however, several distinctions that we now clarify.
\\
\paragraph{}     
    The FMRG construction intrinsically relies on an infinite-dimensional phase space; in its present form, it therefore cannot be defined for finite-dimensional systems. It distinguishes a specific renormalization operator ${\cal R}$. It is built from the scale symmetries of the inviscid system and corresponds to a regularization procedure reminiscent of Galerkin truncation. More precisely, its definition involves a general version of the Feigenbaum operator \cite{Feigenbaum1976}. This Feigenbaum-type transformation zooms into the flow map toward finer scales and then rescales it back in a way that is self-consistent with the inviscid vector field. It should not be confused with the Wilson RG scheme, where fine-scale fluctuations are integrated out in order to produce an effective description at coarser scales. In this sense, the two approaches run in opposite directions in scale space.

    The semigroup operator ${\cal R}$ in FMRG has the property of selecting, in the limit $\tau\to\infty$, particular solutions of the inviscid system. These solutions involve scale symmetries in a nontrivial way, including scale-invariant solutions, which makes the framework particularly relevant for physics and, in particular, for turbulence.

    In the present work, by contrast, we consider all admissible semigroup operators ${\cal R}$, specified through their generators $\mathrm{G}$ in \eqref{RGODE:gamma}.

\paragraph{}
The domain on which ${\cal R}$ acts is also different. In FMRG, the operator acts on the space of full flow maps $(t,x)\mapsto \phi_t(x)$. In the present construction, by contrast, both the physical time $t$ and the initial condition $x_0$ are fixed, so that the relevant object is only the state reached at that time; see \eqref{SRG}. This reduction is precisely what makes the construction finite-dimensional. In FMRG, the self-consistency of the renormalization operator requires \emph{nonlocal} expressions involving both time and initial data, and therefore the full flow map cannot be replaced by its value at a single time and position.

{This distinction can be formulated more intrinsically by introducing an abstract RG space $\mathbb H$, together with an observation map $\Pi:\mathbb H\to H$; see Remark \ref{remark_fmrg} above. In such a formulation, the semigroup acts on $\mathbb H$, while the physical endpoint is obtained only after projection:
$$
\Pi({\cal R}_\tau y)=\phi_t^{\tau,y}(x_0).
$$
The present construction corresponds to the closed reduced case $\mathbb H=H$ and $\Pi=\mathrm{Id}$, where the observed endpoint itself is a sufficient coordinate for the RG dynamics. By contrast, FMRG should be viewed as the case where \(\mathbb H$ is a space of full flow maps and $\Pi$ is the evaluation at the fixed pair $(t,x_0)$. The semigroup property in $\mathbb H$ need not close after this projection: two distinct RG states may have the same observed endpoint while evolving differently under the semigroup. Thus the finite-dimensional formulation used here is a closed projected version of the full flow-map RG picture.}

\section{Autonomous regularizations: Perron--Frobenius formulation}
\label{par1_4_3_2}

We now make precise the idea that the sets ${\mathscr M}(\gamma)$ appearing in Theorem~\ref{M0M} become, in the RG setting, genuine \emph{statistical attractors}. This requires passing from the pointwise RG flow ${\cal R}_\tau$ on $ \rgh $ to the induced Perron--Frobenius semigroup acting on probability measures. This passage is not only natural but essentially forced by the definition of ${\mathscr M}(\gamma)$, which is already expressed in terms of empirical measures along an orbit.
Recall that ${\mathscr M}(\gamma)$ in Definition~\ref{SPdef5} is
$$
\mathscr{M}(\gamma)
= \left\{ \mu~:~\exists\,\tau_n \to \infty \ \text{such that} \
\frac{1}{\tau_n} \int_0^{\tau_n} \delta_{\gamma(s)}\,ds \to \mu \right\}.
$$
In the present autonomous setting, $\gamma$ is an orbit of ${\cal R}$, namely
$$
\gamma(s;y)={\cal R}_s y
$$
for some $y\in \rgh$. Hence, for every bounded continuous $F$,
$$
\langle ({\cal R}_\tau)_\#\delta_y , F\rangle
= \langle \delta_y , F\circ{\cal R}_\tau\rangle
= F({\cal R}_\tau y)
= \langle \delta_{{\cal R}_\tau y} , F\rangle .
$$
Thus
\begin{equation}\label{Diracid}
\delta_{{\cal R}_\tau y}=({\cal R}_\tau)_\#\delta_y .
\end{equation}
This identity is the key point: the empirical measures defining
${\mathscr M}(\gamma)$ are precisely Birkhoff averages of the
Perron--Frobenius orbit of the initial Dirac mass $\delta_y$. Therefore the
correct statistical object is not only the pointwise orbit
$\tau\mapsto{\cal R}_\tau y$, but the lifted orbit
$\tau\mapsto({\cal R}_\tau)_\#\delta_y$ in ${\cal P}(\rgh)$.

This motivates the definition
\be \label{MRy}
\mathscr{M}_{\cal R}(y)
:= \left\{ \mu~:~\exists\,\tau_n \to \infty \ \text{such that} \
\frac{1}{\tau_n} \int_0^{\tau_n} ({\cal R}_s)_\#\delta_y\,ds \to \mu \right\}.
\de 
By \eqref{Diracid}, this is exactly the RG version of
${\mathscr M}(\gamma)$ for the orbit $\gamma(s;y)={\cal R}_s y$.

Here $({\cal R}_\tau)_\#=(y\mapsto{\cal R}_\tau y)_\#$ is the
Perron--Frobenius, or transfer, operator associated with the map
${\cal R}_\tau$; see Appendix~\ref{PFTP}. It should be distinguished from
the pushforward $(\tau\mapsto{\cal R}_\tau y)_\#$, which instead pushes
forward measures on the $\tau$ variable. The former evolves probability
measures on $\rgh$, while the latter records the distribution of a single
orbit under a prescribed sampling of RG time.

The Perron--Frobenius lift also allows one to act on arbitrary probability
measures $\mu\in{\cal P}(H)$, not only on Dirac masses. The operators
$({\cal R}_\tau)_\#$ form a Markov semigroup: they preserve positivity and
total mass, and satisfy the Chapman--Kolmogorov relation
\begin{equation}\label{Chapman}
({\cal R}_{\tau+\tau'})_\#
=
({\cal R}_\tau)_\#\circ({\cal R}_{\tau'})_\# .
\end{equation}
Thus, while ${\cal R}_\tau$ describes the evolution of individual
regularization states, the Perron--Frobenius semigroup describes the evolution
of statistical ensembles. It is this lifted dynamics that makes limiting
statistics into statistical attractors.
Putting these remarks together, we obtain the following result.

\begin{theorem}[Statistical attractors for $({\cal R}_\tau)_\#$]\label{RGattractors}
Let $({\cal R}_\tau)_{\tau\ge0}$ be an RG flow acting continuously on a
(possibly infinite-dimensional) space $\rgh$.
Assume there exists a compact absorbing set $K\subset \rgh$ such that
$$
{\cal R}_\tau K\subset K,
\qquad \forall \tau\ge0.
$$
Since $K$ is absorbing, no generality is lost by restricting attention to initial data in $K$.
For $\mu\in{\cal P}(K)$, define the Birkhoff $\omega$-limit set
$$
\omega(\mu):=
\bigcap_{R>0}
\overline{\left\{
\mu_{R'}(\mu):R'\ge R
\right\}},
\qquad
\mu_R(\mu):=\frac1R\int_0^R({\cal R}_{\tau'})_\#\mu\,d\tau'.
$$
Then:
\begin{enumerate}
\item[(i)]
For every $\mu\in{\cal P}(K)$, the set $\omega(\mu)$ is nonempty, compact, and invariant under $({\cal R}_\tau)_\#$. In particular, $\omega(\mu)$ is a statistical attractor in the sense of Definition~\ref{DefstatAtt}.

\item[(ii)]
The following inclusions hold:
\begin{equation}\label{inclusionsRG}
{\cal E}_{\cal R}\subset{\mathscr M}_{\cal R}\subset
{\mathscr M}_{\cal R}^\#
=
\operatorname{Inv}({\cal R})
=
\overline{\operatorname{co}}\bigl({\cal E}_{\cal R}\bigr),
\end{equation}
where
$$
{\mathscr M}_{\cal R}:=\bigcup_{y\in K}\omega(\delta_y)
=\bigcup_{y\in K}{\mathscr M}_{\cal R}(y),
\qquad
{\mathscr M}_{\cal R}^\#:=\bigcup_{\mu\in{\cal P}(K)}\omega(\mu),
$$
$$
\operatorname{Inv}({\cal R})
:=
\bigl\{
\eta\in{\cal P}(K):({\cal R}_\tau)_\#\eta=\eta,\ \forall \tau\ge0
\bigr\},
$$
and
$$
{\cal E}_{\cal R}
:=
\bigl\{
\mu\in\operatorname{Inv}({\cal R}):\omega(\delta_y)=\{\mu\}
\ \text{for } \mu\text{-a.e. } y
\bigr\}.
$$
Moreover, ${\mathscr M}_{\cal R}^\#$ is a compact and convex subset of ${\cal P}(K)$, whereas ${\mathscr M}_{\cal R}$ is, in general, neither closed nor convex.
\end{enumerate}
\end{theorem}

\begin{proof}
We first show that $\omega(\mu)$ is nonempty and compact for every $\mu\in{\cal P}(K)$.
Since $K$ is compact and forward-invariant, one has
$$
\operatorname{supp}\bigl(({\cal R}_\tau)_\#\mu\bigr)\subset K,
\qquad \forall \tau\ge0,
$$
and therefore also
$$
\operatorname{supp}\bigl(\mu_R(\mu)\bigr)\subset K,
\qquad \forall R>0.
$$
Hence the family $\{\mu_R(\mu):R>0\}$ is tight. By Prokhorov's theorem, it is relatively compact in ${\cal P}(K)$, so every sequence $\mu_{R_n}(\mu)$ admits a weakly convergent subsequence. It follows that
$$
\omega(\mu)
=
\bigcap_{R>0}
\overline{\left\{
\mu_{R'}(\mu):R'\ge R
\right\}}
$$
is nonempty. Since it is an intersection of closed subsets of the compact space ${\cal P}(K)$, it is compact as well.

We next show invariance under $({\cal R}_\tau)_\#$.
An equivalent description of the $\omega$-limit set is
$$
\omega(\mu)
=
\left\{
\eta\in{\cal P}(K):
\exists\,R_n\to\infty,\ \mu_{R_n}(\mu)\rightharpoonup\eta
\right\}.
$$
Fix $\tau>0$ and $\eta\in\omega(\mu)$. Then there exists $R_n\to\infty$ such that
$
\mu_{R_n}(\mu)\rightharpoonup\eta.
$
By continuity of the pushforward map,
$
({\cal R}_\tau)_\#\eta
=
\lim_{n\to\infty}({\cal R}_\tau)_\#\mu_{R_n}(\mu).
$
Using the semigroup property,
$$
({\cal R}_\tau)_\#\mu_{R_n}(\mu)
=
\frac1{R_n}\int_0^{R_n}({\cal R}_{\tau+s})_\#\mu\,ds
=
\mu_{R_n}(\mu)+E_n,
$$
where
$$
E_n
:=
\frac1{R_n}\int_{R_n}^{R_n+\tau}({\cal R}_s)_\#\mu\,ds
-
\frac1{R_n}\int_0^\tau({\cal R}_s)_\#\mu\,ds.
$$
For every $F\in C_b(H;\mathbb R)$,
$$
\bigl|\langle E_n,F\rangle\bigr|
\le
\frac{2\tau}{R_n}\|F\|_\infty
\longrightarrow 0.
$$
Hence
$$
({\cal R}_\tau)_\#\eta
=
\lim_{n\to\infty}\mu_{R_n}(\mu)
=
\eta.
$$
Thus every element of $\omega(\mu)$ is invariant under $({\cal R}_\tau)_\#$, and in particular $\omega(\mu)$ is forward-invariant.
Moreover, by definition, the measure $\mu$ belongs to the Birkhoff basin of $\omega(\mu)$. Therefore $\omega(\mu)$ is a statistical attractor in the sense of Definition~\ref{DefstatAtt}. This proves (i).

We now turn to (ii). The inclusions
$$
{\cal E}_{\cal R}\subset{\mathscr M}_{\cal R}\subset{\mathscr M}_{\cal R}^\#
$$
follow directly from the definitions.
By part (i), every element of ${\mathscr M}_{\cal R}^\#$ is invariant, hence
$$
{\mathscr M}_{\cal R}^\#\subset \operatorname{Inv}({\cal R}).
$$
Conversely, if $\eta\in\operatorname{Inv}({\cal R})$, then
$$
({\cal R}_\tau)_\#\eta=\eta,
\qquad \forall \tau\ge0,
$$
so that
$
\mu_R(\eta)=\eta,~\forall R>0.
$
Therefore
$
\omega(\eta)=\{\eta\},
$
and thus
$$
\operatorname{Inv}({\cal R})\subset{\mathscr M}_{\cal R}^\#.
$$
This proves
$$
{\mathscr M}_{\cal R}^\#=\operatorname{Inv}({\cal R}).
$$

The set $\operatorname{Inv}({\cal R})$ is compact because it is closed in the compact space ${\cal P}(K)$.
It is also convex: if $\eta_1,\eta_2\in\operatorname{Inv}({\cal R})$ and $\theta\in[0,1]$, then
$$
({\cal R}_\tau)_\#\bigl(\theta\eta_1+(1-\theta)\eta_2\bigr)
=
\theta({\cal R}_\tau)_\#\eta_1+(1-\theta)({\cal R}_\tau)_\#\eta_2
=
\theta\eta_1+(1-\theta)\eta_2.
$$
Hence ${\mathscr M}_{\cal R}^\#=\operatorname{Inv}({\cal R})$ is compact and convex.

It remains to justify the identity
$$
\operatorname{Inv}({\cal R})=\overline{\operatorname{co}}\bigl({\cal E}_{\cal R}\bigr).
$$
We first show that ergodic invariant measures are extreme points of
$\operatorname{Inv}({\cal R})$.
Let $\mu\in\operatorname{Inv}({\cal R})$ be ergodic, and suppose
$$
\mu=\theta\mu_1+(1-\theta)\mu_2,
\qquad
\mu_1,\mu_2\in\operatorname{Inv}({\cal R}),
\qquad
\theta\in(0,1).
$$
Then $\mu_i\ll\mu$, so by the Radon--Nikodym theorem there exist
$f_i=\frac{d\mu_i}{d\mu}\ge0$ such that
$$
\int f_i\,d\mu=1,
\qquad
1=\theta f_1+(1-\theta)f_2
\quad \mu\text{-a.e.}
$$
For every Borel set $A\subset H$ and every $\tau\ge0$,
$$
\int_A f_i\,d\mu
=
\mu_i(A)
=
\mu_i({\cal R}_\tau^{-1}A)
=
\int_{{\cal R}_\tau^{-1}A}f_i\,d\mu
=
\int_A f_i\circ{\cal R}_\tau\,d\mu,
$$
where we used invariance of both $\mu_i$ and $\mu$.
Hence
$$
f_i\circ{\cal R}_\tau=f_i
\qquad \mu\text{-a.e.}
$$
By ergodicity of $\mu$, every invariant function in $L^1(\mu)$ is $\mu$-a.e. constant. Thus $f_i\equiv c_i$ $\mu$-a.e. Since $\int f_i\,d\mu=1$, one gets $c_i=1$, so $\mu_i=\mu$. Therefore $\mu$ is an extreme point of $\operatorname{Inv}({\cal R})$.
Conversely, we show that every extreme point of $\operatorname{Inv}({\cal R})$ is ergodic.
We argue by contraposition.
Assume that $\mu\in\operatorname{Inv}({\cal R})$ is not ergodic. Then there exists an invariant Borel set $B_\star$ such that
$$
0<\mu(B_\star)<1.
$$
Define
$$
\mu_1(B):=\frac{\mu(B_\star\cap B)}{\mu(B_\star)},
\qquad
\mu_2(B):=\frac{\mu(B_\star^c\cap B)}{\mu(B_\star^c)}.
$$
Since $B_\star$ is invariant, both $\mu_1$ and $\mu_2$ are invariant. Moreover,
$$
\mu
=
\mu(B_\star)\mu_1+\mu(B_\star^c)\mu_2
=
\mu(B_\star)\mu_1+\bigl(1-\mu(B_\star)\bigr)\mu_2,
$$
which is a nontrivial convex decomposition. Therefore $\mu$ is not an extreme point.

We have thus shown that the extreme points of $\operatorname{Inv}({\cal R})$ are precisely the ergodic invariant measures, namely the elements of ${\cal E}_{\cal R}$. Since $\operatorname{Inv}({\cal R})$ is compact and convex, the Krein--Milman theorem gives
$$
\operatorname{Inv}({\cal R})
=
\overline{\operatorname{co}}\bigl({\cal E}_{\cal R}\bigr).
$$
This concludes the proof.
\end{proof}

The interest of Theorem~\ref{RGattractors} is not merely that it recovers classical facts from ergodic theory in the present setting. Its main contribution is to show that, once a semigroup structure is imposed on the regularizations, the sets ${\mathscr M}_{\cal R}(y)$ are no longer arbitrary accumulation sets: they become organized by invariant and ergodic measures of the RG dynamics, and hence acquire a genuine dynamical meaning. In contrast with the set-theoretic flexibility of Part~I, where essentially all probability measures in ${\mathscr M}_0$ may be realized, the semigroup framework produces a more rigid and therefore more informative subclass, governed by attractors, invariant measures, and basin structure. This is precisely what makes the RG viewpoint interesting: it does not merely rephrase $\SP$, but connects it to the deeper architecture of dynamical systems.

A simple illustration of these results is given in Section~\ref{RGreg}.
The inclusions in \eqref{inclusionsRG} take different forms depending on the underlying dynamical system.
The main point is that, once a semigroup structure is imposed, the measures selected in the inviscid limit are no longer arbitrary: they are constrained by the recurrent and ergodic structure of the RG flow.
As a first consequence, one obtains the following nonexhaustive list of scenarios
(see the dynamical-systems definitions recalled in Appendix~\ref{PFTP}).

\begin{corollary}\label{coroDS}
Let $({\cal R}_\tau)_{\tau\ge0}$ be an RG flow as in Theorem~\ref{RGattractors}.
If $\mu\in{\cal E}_{\cal R}$ is non-Dirac, then for $\mu$-a.e.\ $y$ one has
$$
{\mathscr M}_{\cal R}(y)=\{\mu\},
$$
and therefore the corresponding regularization exhibits Strong-$\SP$ with selected measure $\mu$.

Typical scenarios include:
\begin{enumerate}
    \item {\bf Uniquely ergodic:}
    $$
    {\cal E}_{\cal R}
    =
    {\mathscr M}_{\cal R}
    =
    {\mathscr M}_{\cal R}^{\#}
    =
    \operatorname{Inv}({\cal R})
    =
    \{\mu\}.
    $$
    In particular, if $\mu$ is non-Dirac, then Strong-$\SP$ holds for every $y\in K$.

    \item {\bf Axiom A:}
    Assume that the nonwandering set $\Omega$ decomposes into finitely many basic sets $\Lambda_i$ (compact, invariant, and topologically transitive). Then
    $$
    {\mathscr M}_{\cal R}
    \subset
    \bigcup_i \operatorname{Inv}\!\left(\left.{\cal R}\right|_{\Lambda_i}\right)
    \subset
    \operatorname{Inv}({\cal R})
    =
    \overline{\operatorname{co}}
    \left(
    \bigcup_i {\cal E}\!\left(\left.{\cal R}\right|_{\Lambda_i}\right)
    \right).
    $$
    Under the standard generic-point results available for Axiom~A basic sets, the first inclusion is in fact an equality.

    \item {\bf Morse--Smale systems:}
    The set ${\mathscr M}_{\cal R}$ consists of ergodic measures supported on attracting equilibria and attracting periodic orbits. In particular,
    $$
    {\mathscr M}_{\cal R}\subset {\cal E}_{\cal R}\subset \operatorname{Inv}({\cal R}).
    $$
    If, in addition, the system is gradient (that is, it has no periodic orbit), then ${\mathscr M}_{\cal R}$ contains only Dirac masses at equilibria, and hence one cannot obtain $\SP$.
\end{enumerate}
\end{corollary}

It is natural to ask whether the same degree of flexibility as in Theorem~\ref{M0M} can still be achieved within the RG framework.
Theorem~\ref{RGattractors} strongly suggests that this is a more rigid problem: once a semigroup structure is imposed, the admissible limiting measures must arise as invariant measures of a single dynamical system, and are therefore subject to genuine geometric and ergodic constraints.
In particular, producing non-Dirac limits requires recurrent dynamics inside ${\mathscr S}_0$, which is far from automatic on a compact set with boundary.
At the very least, in the conservative case one expects the generator $g$ to satisfy
$
\operatorname{div} g = 0~\text{in }{\mathscr S}_0,~
g\cdot \mathbf n = 0~
\text{on }\partial{\mathscr S}_0,
$
so that the flow preserves volume and remains tangent to the boundary.
Whether such constraints still allow one to recover the full flexibility of Part~I remains open.
More precisely, do we have
\begin{equation}
\bigcup_{\cal R} {\mathscr M}_{\cal R} = {\mathscr M}_0,
\end{equation}
where ${\cal R}$ ranges over all RG flows in the sense of Definition~\ref{RGflow}?

\section{Application to systems with an isolated singularity}\label{RGreg}

In light of Section~\ref{CNforSP}, a necessary condition for observing $\SP$ is that the inviscid system possess a singular set and that the initial condition $x_0$ belongs to its stable set. We therefore consider the simplest nontrivial situation, namely an inviscid vector field with an isolated non-Lipschitz singularity at the origin. This setting illustrates concretely how the abstract framework of Section~\ref{par1_4_3_2}, together with Theorem~\ref{RGattractors}, can be applied in practice, while also suggesting new directions opened up by the RG viewpoint. More generally, the same strategy extends to singular sets; see Theorem~\ref{CN}.

The basic idea is to regularize the vector field only inside a controlled neighborhood of the singularity, with this neighborhood shrinking to zero in the inviscid limit. The essential information is then not the detailed regularized dynamics itself, but the way trajectories exit the shrinking region and are subsequently propagated by the inviscid flow.

We assume that the inviscid vector field $f_0\in C_b(\mathbb R^d;\mathbb R^d)$ is Lipschitz on $\mathbb R^d\setminus\{0\}$, satisfies $f_0(0)=0$, and has an isolated non-Lipschitz singularity at the origin, assumed to be of expanding type. Typical examples are systems of the form
$$
f_0(x)=|x|^\alpha F_0\!\left(\frac{x}{|x|}\right),
\quad
$$
for $\alpha \in (0,1)$;
see \cite{Drivas21,Drivas24}. 
To avoid an unnecessary discussion of the dependence on $x_{0}$ and on the observation time $t$ (see Theorem~\ref{CN}), we simply fix $x_{0}=0$ and $t=1$.

For $\tau\ge0$, $\rho_\tau\ge0$, and $x\in\mathbb R^d$, we define the regularized vector field $f$ in~\eqref{RGODE:x}, by
\begin{equation}\label{reguiso}
f(x,\tau,y)=
\begin{cases}
h(x,\tau,y), & |x|\le \rho_\tau,\\
f_0(x),      & |x|\ge \rho_\tau,
\end{cases}
\qquad
\lim_{\tau\to\infty}\rho_\tau=0,
\end{equation}
where, for every $(\tau,y)$, the map $x\mapsto h(x,\tau,y)$ is Lipschitz and is chosen so that $f(\cdot,\tau,y)$ is continuous across the boundary of the ball $B_{\rho_\tau}$. 
{We may also write the second autonomous equation \eqref{RGODE:gamma} explicitly, but we do not need it since it is fully described
by the associated nonlinear semigroup  $(\mathcal R_\tau)_{\tau\ge0}$; see \eqref{RGODE:coupling} and \eqref{RGconstr} below.}

The key quantity is the \emph{last exit time} from $B_{\rho_\tau}$, rather than the detailed flow inside.
Let $X_\tau(s)$ denotes the solution at time $s\ge0$ of
\begin{equation}\label{inside}
\dot X=h(X,\tau,y),
\qquad
X(0)=0,
\end{equation}
with the dependence on $y$ left implicit. Define the last exit time from $B_{\rho_\tau}$ by
$$
t^\star_\tau
:=
\sup\bigl\{
s\in[0,1]:\|X_\tau(s)\|=\rho_\tau
\bigr\},
$$
with the convention $t^\star_\tau=+\infty$ if $\|X_\tau(s)\|<\rho_\tau$ for all $s\in[0,1]$. In other words, $t^\star_\tau$ is the last time before $t=1$ at which the inner trajectory meets the boundary $\partial B_{\rho_\tau}$.

Let $s \mapsto Z_\tau(s)$ denotes the solution of the inviscid system with the \emph{modified initial condition} induced by the last exit:
\begin{equation}\label{reguP0sol}
\dot Z=f_0(Z),
\qquad
s\ge t^\star_\tau,
\qquad
Z_\tau(t^\star_\tau)=X_\tau(t^\star_\tau).
\end{equation}
Then, from \eqref{RGODE}, one has
\begin{equation}\label{RGconstr}
{\cal R}_\tau y=
\begin{cases}
X_\tau(1), & t^\star_\tau=+\infty,\\
Z_\tau(1), & t^\star_\tau\le1.
\end{cases}
\end{equation}
A sketch of this construction is shown in Figure~\ref{RG_exit}. The problem is therefore reduced to understanding the asymptotic behavior of \eqref{RGconstr} as $\tau\to\infty$.

\begin{figure}[htbp]
\centerline{\includegraphics[scale=0.6]{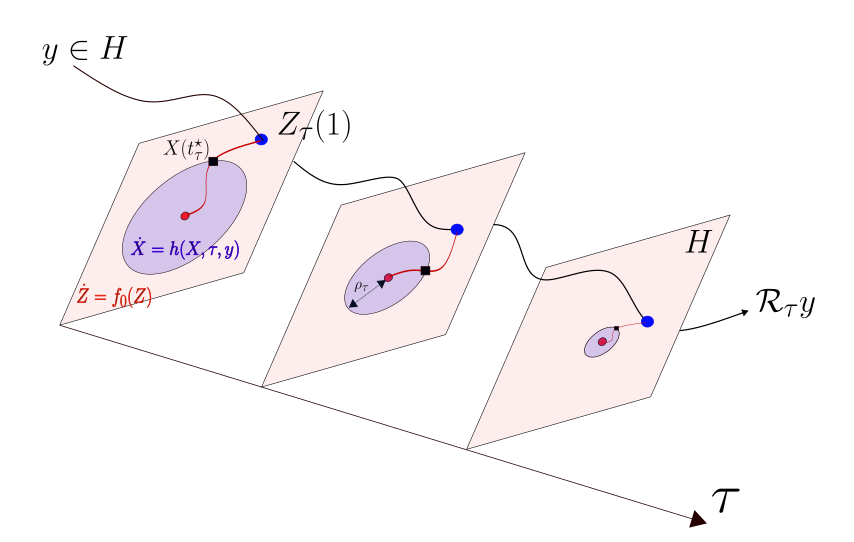}}
\caption{Regularization near an isolated singularity and matching with the outer inviscid flow at the last exit time $t^\star_\tau$.}
\label{RG_exit}
\end{figure}

The key point is that, in this construction, the detailed dynamics inside $B_{\rho_\tau}$ matters only through the exit configuration it produces. Once the trajectory has left the shrinking regularization region, the subsequent evolution is governed entirely by the inviscid flow. In this sense, the selection mechanism is encoded by the asymptotic behavior of the exit time and exit location.

In full generality, one should not expect $t^\star_\tau$ to have a limit as $\tau\to\infty$. We first consider the trapped regime, namely
\begin{equation}\label{trap}
\forall \tau\ge0,
\qquad
t^\star_\tau=+\infty.
\end{equation}
In that case, the trajectory never reaches $\partial B_{\rho_\tau}$ before time $t=1$, and \eqref{RGconstr} gives
$$
{\cal R}_\tau y = X_\tau(1).
$$
Since the entire motion remains inside $B_{\rho_\tau}$ and $\rho_\tau\to0$, the trapped regime forces
$$
{\cal R}_\tau y \longrightarrow 0
\qquad\text{as }\tau\to\infty.
$$
In particular, the RG dynamics must admit $0$ as an invariant state and must attract the corresponding orbit toward it. Otherwise, the regularization encoded by ${\cal R}_\tau$ cannot be realized.

The more interesting case is when the trapped regime does not occur. Assume that
\begin{equation}\label{notrap}
\forall \tau\ge0,
\qquad
t^\star_\tau\le1.
\end{equation}
Let ${\cal L}_y$ denote the set of accumulation points, as $\tau\to\infty$, of the outer inviscid solutions \eqref{reguP0sol} observed at time $t=1$:
\begin{equation}\label{LyDef}
{\cal L}_y
:=
\bigcap_{T\ge0}
\overline{
\left\{
Z_\tau(1):\tau\ge T
\right\}
}
\subset {\mathscr S}_0,
\end{equation}
where ${\mathscr S}_0$ is defined in \eqref{S0}. Since \eqref{RGconstr} gives ${\cal R}_\tau y=Z_\tau(1)$ in this regime, one immediately obtains
\begin{equation}\label{omegay_equals_Ly}
\omega(y)
:=
\bigcap_{T\ge0}
\overline{
\left\{
{\cal R}_\tau y:\tau\ge T
\right\}
}
=
{\cal L}_y
\subset {\mathscr S}_0.
\end{equation}
This identity is useful because it shows that one does not need the precise form of the regularization \eqref{reguiso}; what matters is the limit set of the semigroup generated by $g$, subject only to the constraint that it lie in ${\mathscr S}_0$.

Note that the state-space $\omega$-limit set $\omega(y)$ is not the same object as the measure-space limit set $\omega(\delta_y)$ from Theorem~\ref{RGattractors}. By construction, the latter consists of invariant probability measures supported on $\omega(y)$. Consequently, if
$$
\omega(\delta_y)=\{\mu\}
\qquad\text{with }\mu \text{ non-Dirac},
$$
then \eqref{reguiso} yields Strong-$\SP$ for the RG regularization associated with $y$.

There is also an intermediate possibility in which neither \eqref{trap} nor \eqref{notrap} holds uniformly in $\tau$. A simple toy example is obtained by taking
$$
h(X,\tau,y)=\sin^2(\pi\tau),
$$
so that the inner dynamics oscillates rapidly with $\tau$. In that case, whenever $\sin(\pi\tau)\neq0$, one has an exit time of order
$$
t^\star_\tau \sim \frac{\rho_\tau}{\sin^2(\pi\tau)},
$$
whereas for $\tau\in\mathbb N$ the trajectory remains trapped at the origin and $t^\star_\tau=+\infty$. Thus both behaviors may coexist along different subsequences.

Independently, the notion of \emph{extremal} solution, often cited as the one selected in inviscid limits, has a simple formulation in the present context.

\begin{definition}[Extremal solutions]\label{extremalsol}
An extremal solution of \eqref{P_0} is a solution whose value at time $t$ belongs to $\partial{\mathscr S}_0$.
\end{definition}

The immediate consequence is that all regularizations for which the exit time tends to zero as $\tau\to\infty$ select only extremal states:
$$
\bigcup_{y\in{\cal Y}_0}\omega(y)\subset\partial{\mathscr S}_0,
\qquad
{\cal Y}_0:=\left\{
y:\lim_{\tau\to\infty}t^\star_\tau=0
\right\}.
$$

\paragraph{A simple example}

We consider the classical one-dimensional problem
$$
\dot x=f_0(x),
\qquad
x(0)=0,
\qquad
f_0(x)=x^{1/3},
$$
which has explicit solutions; see also Section~\ref{Ex1}. Observing the system at time $t=1$, one obtains
$$
{\mathscr S}_0
=
\left[
-\left(\frac23\right)^{3/2},
\left(\frac23\right)^{3/2}
\right].
$$

Our aim is to build an explicit system of the form \eqref{RGODE}. Choose a monotone decreasing generator $g$ with a unique equilibrium at $\chi$, namely
$$
G(Y)=\chi-Y,
\qquad
{\cal R}_\tau y=\chi+(y-\chi)e^{-\tau}.
$$
Let $t^\star_\tau$ be the last exit time from the interval $(-\rho_\tau,\rho_\tau)$ for the inner system \eqref{inside}. In this one-dimensional setting, once the trajectory leaves $(-\rho_\tau,\rho_\tau)$, it cannot re-enter for the class of regularizations considered here.

Assume that \eqref{notrap} holds, namely that the trajectory exits before the observation time $t=1$. Then one must solve \eqref{reguP0sol}. If $X_\tau(t^\star_\tau)=x_\tau$, the corresponding inviscid solution is
$$
Z_\tau(t)
=
\operatorname{sign}(x_\tau)
\left(
|x_\tau|^{2/3}
+\frac23(t-t^\star_\tau)
\right)^{3/2}.
$$
Since \eqref{RGconstr} imposes
$$
{\cal R}_\tau y = Z_\tau(1),
$$
the limit $\tau\to\infty$ forces a relation between the asymptotic exit time and the selected state $\chi$. Because ${\cal R}_\tau y\to\chi$, one obtains
\begin{equation}\label{T_star}
1-\frac32 |\chi|^{2/3}
=
\lim_{\tau\to\infty} t^\star_\tau
\in[0,1].
\end{equation}
Accordingly, extremal solutions correspond to the regime $t^\star_\tau\to0$, i.e. $\chi = \pm (\frac23)^{3/2}$ whereas trapped solutions at the origin correspond to the limiting regime $t^\star_\tau\to1$, i.e. $\chi=0$.

Thus, for each selected state
$$
\chi\in
{\mathscr S}_0=
\left[
-\left(\frac23\right)^{3/2},
\left(\frac23\right)^{3/2}
\right],
$$
one obtains a simple RG dynamical system ${\cal R}_\tau$ selecting $\chi$.

That only Dirac masses are selected here is not a defect of this particular
construction: it is a one-dimensional rigidity. A continuous autonomous flow
on an interval is necessarily gradient-like: it has no periodic orbit, every
$\omega$-limit set reduces to equilibria, and hence every ergodic invariant
measure is a Dirac mass at an equilibrium. By Corollary~\ref{coroDS}
(Morse--Smale, gradient case), the sets ${\mathscr M}_{\cal R}(y)$ then
contain only such Dirac masses, so that autonomous RG regularizations acting
on an interval can reach only the extreme points
${\cal E}=\{\delta_x\}_{x\in{\mathscr S}_0}$ of $\mathscr M_0$, and never the
convex combinations required for $\SP$. The obstruction is dimensional, not
dynamical: Example~\ref{Ex1} shows that nonautonomous regularizations of the
same inviscid system do produce $\SP$, while in higher-dimensional analogues
the availability of recurrent sets allows Strong-$\SP$ regularizations with
measures supported, for instance, on limit cycles.

Within this rigid one-dimensional setting, the RG viewpoint nevertheless
retains genuine content: it determines \emph{which} Dirac mass is selected,
and organizes the answer through the basin structure of the generator.

For the same inviscid system, one may produce a regularization with $n$ distinct classes, not necessarily of comparable size. A different choice of RG flow may lead to a completely different scenario.

Assume that the RG semigroup generator $g$ vanishes at several points within the interval
$$
\left[
-\left(\frac23\right)^{3/2},
\left(\frac23\right)^{3/2}
\right].
$$
Suppose that $g$ has $n$ roots $\rho_k$ separated by $n-1$ saddle points $\sigma_k$:
$$
-\left(\frac23\right)^{3/2}
\le
\rho_1<\sigma_1<\rho_2<\sigma_2<\cdots<\sigma_{n-1}<\rho_n
\le
\left(\frac23\right)^{3/2}.
$$
The corresponding RG flow then has $n$ basins of attraction, denoted ${\cal B}_k$, such that all regularizations with parameter $y\in{\cal B}_k$ select the same state $\rho_k$ in the inviscid limit. These basins are
$$
{\cal B}_k=(\sigma_{k-1},\sigma_k),
\qquad
k=1,\dots,n,
$$
with the convention $\sigma_0=-\infty$ and $\sigma_n=+\infty$. In this setting one typically obtains a Morse--Smale system satisfying Corollary~\ref{coroDS}. In higher-dimensional analogues, this mechanism may produce Strong-$\SP$ regularizations with measures supported on limit cycles, although this cannot occur in one dimension, where continuous autonomous dynamics is necessarily gradient-like.

An extreme version is obtained by taking $g|_{{\mathscr S}_0}=0$, with
$$
g(\phi)>0
\quad\text{for }\phi<-\left(\frac23\right)^{3/2},
\qquad
g(\phi)<0
\quad\text{for }\phi>\left(\frac23\right)^{3/2},
$$
so that ${\mathscr S}_0$ attracts all trajectories with $y\notin{\mathscr S}_0$. Then ${\cal R}_\tau$ has ${\mathscr S}_0$ as a global attractor, and every $\mu\in{\cal P}({\mathscr S}_0)$ is an invariant measure for ${\cal R}_\tau$. By Theorem~\ref{RGattractors}, one has in this case
$$
{\cal E}_{\cal R}
=
\{\delta_x\}_{x\in{\mathscr S}_0}
=
{\mathscr M}_{\cal R}
\subset
{\mathscr M}_{\cal R}^{\#}
=
\operatorname{Inv}({\cal R})
=
{\cal P}({\mathscr S}_0).
$$
This dynamical system is highly degenerate and is far from hyperbolic. Nevertheless, the example can be generalized to arbitrary systems by imposing
$$
g|_{{\mathscr S}_0}=0,
$$
which gives an alternative simple proof that
$$
{\cal E}_{\cal R}={\cal E}\subset{\mathscr M}_0;
$$
see \ref{EinM}.

To summarize, this example shows how a local singularity can generate a whole hierarchy of selection mechanisms once it is coupled to an RG dynamics. The role of the shrinking regularization region is to convert the singular behavior of the inviscid system into an exit problem, while the RG flow determines how the corresponding exit data are organized asymptotically. In this way, the problem of selecting inviscid states or statistics is recast as a dynamical question on attractors, invariant measures, and basins of attraction. The one-dimensional example only illustrates the most rigid case; in higher dimensions, the richness of the available dynamics is expected to make rigidity largely irrelevant and to restore a high degree of flexibility.

\section{Non-autonomous regularizations}\label{par1_4_3_3}
The previous framework is effective when the regularization curve is generated by a single semigroup orbit. However, not every regularization curve needs to arise in this way. To recover a semigroup structure in complete generality, one can lift the problem from the state space $H$ to the space of curves itself. This is the role of the Bebutov flow, which acts by translation in the regularization parameter.

\begin{definition}[Bebutov]\label{Bebutov}
The \emph{Bebutov flow} $(\mathscr{S}_\tau)_{\tau \in \mathbb{R}^+}$ is the semigroup acting on the topological space $\Gamma$, endowed with the compact-open topology. It is defined by left translations
\be 
\mathscr{S}_\tau: \Gamma \to \Gamma, \quad (\mathscr{S}_\tau \gamma)(\tau') := \gamma(\tau + \tau').
\de 
In addition, define the evaluation map
\be \label{evalmap}
\Psi_\tau: \Gamma \to H, \quad \Psi_\tau(\gamma) := \gamma(\tau).
\de 
One has $\gamma(\tau+\tau_0) = \Psi_{\tau_0}(\mathscr{S}_\tau \gamma)$, and for brevity we write $\Psi := \Psi_{\tau_0}$ for some arbitrary $\tau_0 \geq 0$. The curve $\gamma$ is thus regarded as the \emph{trace} of the dynamical system $(\mathscr{S}_\tau)$. Moreover, $(\mathscr{S}_\tau)$ is a continuous semigroup on $\Gamma$ for the compact-open topology.
\end{definition}

This construction corresponds to Theorem 7.2.1 in \cite{Lasota1994} and is a classical result in ergodic theory.
This shift acts on the space $\Gamma$ of regularization curves defined in \eqref{Gamma} and induces a corresponding pushforward semigroup $(\mathscr{S}_\tau)_\#$ on the space of probability measures $\mathcal{P}(\Gamma)$.

Finally, since one does not expect $\Gamma$ to be a Polish space care must be taken when considering ${\cal P}(\Gamma)$. However, since only the inviscid limit $\tau \to \infty$ ($\epsilon \to 0$ for \eqref{P_eps}) matters, one introduces the \emph{$\Psi$-topology}:
\begin{definition}[$\Psi$-topology]
Let $X$ and $Y$ be topological spaces with $Y$ Polish, let $\Psi: X \to Y$. We define the $\Psi$-topology, the coarsest topology that makes the pushforward map $\Psi_\#:{\cal P}(X) \to 
{\cal P}(Y)$ continuous:
 $$
 \nu_n \overset{\scriptscriptstyle \Psi}{\rightharpoonup} \nu \Longleftrightarrow \Psi_\# \nu_n \rightharpoonup \Psi_\# \nu.  
 $$ 
 It induces a quotient topology on ${\cal P}(X)\big/\!\sim$ via the equivalence relation:
 $$
 \nu_1 \sim \nu_2 \Longleftrightarrow \Psi_\# \nu_1 = \Psi_\# \nu_2.
 $$
\end{definition}

In our case $X=\Gamma$, $Y=H$, and $\Psi$ is the evaluation map
in (\ref{evalmap}).
In other words, this topology corresponds to convergence in expectation against the class of observables
$$
\mathcal{O}_\Psi := \left\{ F \circ \Psi \;:\; F \in C_b({H};\mathbb{R}) \right\},
$$
i.e., observables that depend only on the evaluation $\gamma(0)$ (or any $\gamma(\tau_0)$ for fixed $\tau_0$).

\bigskip

The Bebutov flow defined in \eqref{Bebutov} acts as a shortcut against RG flows, by bypassing/relaxing totally the semigroup property on ${H}$. It does so by considering another semigroup property in a larger space of regularization orbits (the set $\Gamma$) no matter it is inherently autonomous or not.
As explained below, it is {\it universal} in the sense that it yield a complete upgrading of Theorem \ref{M0M}.

Although the Bebutov flow $(\mathscr{S}_\tau)$ also defines a semigroup, it is different from $\mathcal{R}_\tau$, since it acts on a space of curves $\Gamma \subset C_b(\mathbb{R}^+; {H})$. 
Let $i: {H} \to \Gamma$ denotes a (non-unique) embedding from the state space to the space of regularization curves, namely
for $\phi \in H$, $i(\phi) = \gamma_\phi \in \Gamma$.
Moreover, in order for ${\cal R}_\tau$ to be a semigroup on $H$, one imposes the lift to form autonomous curves:
$$
\Psi \circ i = {\rm Id}~\text{and}~i(\phi)(\tau_1+\tau_2)=i(i(\phi)(\tau_2))(\tau_1).
$$
This is indeed the more abstract rephrasing of \eqref{RGODE}.
The relationship between these two semigroups is expressed by the following diagram:
$$
\begin{tikzcd}[column sep=large, row sep=large]
{H} \arrow[r, hook, "i"] \arrow[d, "\mathcal{R}_\tau"']
&
\Gamma \arrow[d, "\mathscr{S}_\tau"] \arrow[r, "\Psi"']
&
{H} \\
{H} \arrow[r, hook, "i"]
&
\Gamma \arrow[r, "\Psi"']
&
{H}
\end{tikzcd}
$$
They both decrease the level of regularization of (\ref{P_0}) with
\be \label{RPSi}
{\cal R}_\tau = \Psi \circ \mathscr{S}_\tau \circ i.
\de 
 Given a solution $\phi$ of $(\mathcal{P}_\tau)$ for some fixed $\tau$, multiple regularizations can arise, corresponding to different vector fields $f \in \V$ generating flows that agree at that point.
One can interpret this relation as parameterizing the set of RG semigroups through the embeddings $i$, namely given
${\cal R}_\tau$, it is defined as $i(\phi)(\tau) := {\cal R}_\tau \phi$ and conversely given some embedding $i$, ${\cal R}_\tau$
is just defined through the relation (\ref{RPSi}).

In a similar way than previously, rather than $\mathscr{S}_\tau$, the natural object to consider is its pushforward $(\mathscr{S}_\tau)_\#$. The identity (\ref{Diracid}) is replaced by:
\be \label{Bebuid}
\delta_{\Psi \circ \mathscr{S}_\tau \gamma} = \Psi_\# (\mathscr{S}_\tau)_\# \delta_\gamma = \delta_{\gamma(\tau)},~\gamma \in \Gamma.
\de 
\begin{theorem}[Statistical attractors for $(\mathscr S_\tau)_\#$]\label{thm:M0Mupup}
Let $\Gamma\subset C_b(\mathbb R^{+};H)$ be equipped with the compact–open
topology and $(\mathscr S_\tau)_{\tau\ge0}$ the Bebutov left–shift on $\Gamma$.
Write $\Psi(\gamma):=\gamma(0)$ and endow $\mathcal P(\Gamma)$ with the
$\Psi$–topology.  We moreover assume that $\Psi(\Gamma)$ is closed.
For $\nu\in\mathcal P(\Gamma)$ set 
$$
\mathscr M(\nu):=\omega(\nu):=
\bigcap_{\R>0}
\overline{\left\{\, \nu_{\R'}(\nu)
  \;:\;R'\ge R
\right\}}^{\scriptscriptstyle \Psi},~\text{with}~\nu_{\R}(\nu) :=\frac{1}{\R}\int_0^\R (\mathscr{S}_{\tau'})_\# \nu\, d\tau'.
$$
Then:
\begin{enumerate}
\item[(i)] $\mathscr M(\nu)$ is non–empty, compact for the $\Psi$–topology,
forward–invariant under $(\mathscr S_\tau)_\#$, and therefore a
statistical attractor with $\nu\in\mathcal B(\mathscr M(\nu))$.
\item[] 
\item[(ii)] $\Psi_\#\mathscr M(\nu)\subset\mathscr M_0:=\mathcal P(\mathscr S_0)$.
If $\nu=\delta_\gamma$ with $\gamma\in\Gamma$, then
$\Psi_\#\mathscr M(\delta_\gamma)=\mathscr M(\gamma)$.
\item[] 
\item[(iii)] Define
$$
\mathscr M:=\bigcup_{\nu\in\mathcal P(\Gamma)}\mathscr M(\nu).
$$
The set $\mathscr M$ is a global (non-minimal) statistical attractor for the $\Psi$–topology, 
and satisfies
$$
\Psi_\#\mathscr M=\mathscr M_0
=\overline{\operatorname{co}}({\cal E}),
\qquad
   \mathscr{M}
   \;=\;
   \operatorname{Inv}(\mathscr{S})
   \;:=\;
   \bigl \{\eta\in\mathcal P(\Gamma)\;:\;(\mathscr{S}_\tau)_\# \eta=\eta,\ \forall\tau\ge0\bigr\}.
$$
\end{enumerate}
\end{theorem}
Since the objects to investigate are regularization curves $\gamma \in \Gamma$ and since one is only interested in their marginal statistical behavior in the inviscid limit, one needs to restrict the set of observables through the lens of the evaluation map 
$\Psi$, namely to consider the class of observables:
$$
\mathcal{O}_\Psi := \left\{ F \circ \Psi \;:\; F \in C_b(H;\mathbb{R}) \right\},
$$
i.e., observables that depend only on the evaluation $\gamma(0)$ (or any $\gamma(\tau_0)$).
It implies to consider a coarser topology which makes $\Psi_\#$ continuous. The following general lemma will be useful:
\begin{lemma}\label{topolemma}
 Let $X$ and $Y$ be topological spaces and $Y$ a Polish space, $\Psi:X \to Y$ and its pushforward
 $\Psi_\#: {\cal P}(X) \to {\cal P}(Y)$, with the standard weak topology. We define the $\Psi$-topology, the coarsest topology that makes $\Psi_\#$ continuous:
 $$
 \nu_n \overset{\scriptscriptstyle \Psi}{\rightharpoonup} \nu \Longleftrightarrow \Psi_\# \nu_n \rightharpoonup \Psi_\# \nu.  
 $$
For $G \subset {\cal P}(X)$, the following hold:
\begin{enumerate}
\item $G$ is $\Psi$-closed $\Longleftrightarrow$ there exist $ C \subset {\cal P}(Y)$ closed such that $G = \Psi_\#^{-1} C$.
\item $G$ is $\Psi$-compact $\Longleftrightarrow$ $\Psi_\# G$ is compact.
\end{enumerate}

\end{lemma}
\begin{proof} 1. $\Rightarrow$: 
Let $C = \overline{\Psi_\# G}$, one has $\Psi_\# G \subset \overline{\Psi_\# G}$, giving 
$G \subset \Psi_\#^{-1}(C)$. Let $\nu \in \Psi_\#^{-1}(C)$, i.e. $\Psi_\# \nu \in C=\overline{\Psi_\# G}$, one has therefore a sequence 
$\Psi_\# \nu_n \rightharpoonup \Psi_\# \nu$, i.e. 
$\nu_n \overset{\scriptscriptstyle \Psi}{\rightharpoonup} 
\nu$ and since $G$ is $\Psi$-closed $\nu \in G$ and
$\Psi_\#^{-1}(C) \subset G$.
1. $\Leftarrow$: since $G = \Psi_\#^{-1}(C)$ with $C$ closed and $\Psi_\#$ is continuous then its preimage is closed in the $\Psi$-topology. 2. $\Rightarrow$: since $\Psi_\#$ is continuous for the $\Psi$-topology, then the image by $\Psi_\#$ of a compact set is compact.
2. $\Leftarrow$: Let $g_n$ be a sequence in $G$, we need to exhibit a subsequence such that it converges to some $g \in G$. Since, $\Psi_\# G$ is compact, the sequence $\Psi_\# g_n$ has a subsequence $\Psi_\# g_{n_k} \rightharpoonup y \in \Psi_\# G$. Let $x \in {\cal P}(X)$ such that $\Psi_\# x = y$, i.e. $x \in \Psi_\#^{-1}\{y\}$ and since 
$\Psi_\#^{-1}\{y \} \cap G \neq \emptyset$ one can choose $x \in G$.
\end{proof}

\begin{proof}[Proof of Theorem \ref{thm:M0Mupup}]
We use Lemma \ref{topolemma} with $X=\Gamma,Y=H$ where $Y$ is Polish. Moreover, $\Psi(\Gamma)$ closed means that ${\rm Im}(\Psi_\#)$ is closed. In term of curves $\gamma \in \Gamma$, it translated to $\cup_{\gamma \in \Gamma} \{ \gamma(0) \}$ is closed, where $\gamma(0)$ is describing some highly regularized version of the inviscid problem $({\cal P}_0)$. This is therefore a very weak constraint on the set $\Gamma$.

\smallskip
\emph{(i)} Let $N_\R:= \{\nu_{\R'}(\nu)~:~\R' \geq \R \} \subset{\cal P}(\Gamma)$. Then, one has for $\rho$ large enough
$$
\Psi_\# N_R \subset {\cal P}(B_\rho),
$$
where $B_\rho$ is the closed ball of $H$ of radius $\rho$ which is compact. Thus by Prokhorov's Theorem, $\overline{\Psi_\# N_R}$ is compact. Moreover, $\Psi_\# \omega(\nu) = \bigcap_{\R > 0} 
\Psi_\# \overline{N_R}^\Psi = \bigcap_{\R > 0} 
\overline{\Psi_\# N_R} \cap {\rm Im}(\Psi_\#) = 
\bigcap_{\R > 0} 
\overline{ \Psi_\# N_R}$ since ${\rm Im}(\Psi_\#)$ is closed.
The intersection of compact sets is also compact, and thus $\Psi_\# \omega(\nu) $ is compact. By applying Lemma \ref{topolemma} 2. one concludes that $\omega(\nu)= \mathscr{M}(\nu)$ is $\Psi$-compact. Note also
that since $\Psi_\# \omega(\delta_\gamma) = {\mathscr M}(\gamma) \neq \emptyset$ then $\mathscr{M}(\nu)$ is non-empty.
For the forward invariance, we use the alternative expression 
$$
\mathscr{M}(\nu) = \omega(\nu) = \left\{ \eta~:~\exists \R_n \to \infty,  \nu_{\R_n} (\nu)\rightharpoonup \eta ~\mbox{in $\Psi$-topology}\right\}.
$$
Let $\eta \in \mathscr{M}(\nu)$,  
one has some $\R_n \to \infty$ such that $\eta$ is the limit of $\nu_{\R_n}(\nu)$ and then
$$
(\mathscr{S}_\tau)_\# \eta = (\mathscr{S}_\tau)_\# \lim_{n \to \infty} \nu_{\R_n}(\nu) = \lim_{n \to \infty}
\nu_{\R_n}(\nu) + \R_n^{-1} \left(
\int_{\R_n}^{\R_n+\tau} (\mathscr{S}_\tau')_\# \nu d\tau' - 
\int_0^\tau (\mathscr{S}_\tau')_\# \nu d\tau'
\right).
$$
The remainder converges (weakly) to zero in the $\Psi$-topology implying that $(\mathscr{S}_\tau)_\# \eta = \eta$.
As a consequence, one has $(\mathscr{S}_\tau)_\# \mathscr{M}(\nu) \subset \mathscr{M}(\nu)$.
One therefore concludes (i), i.e. $\mathscr{M}(\nu)$ is a statistical attractor.

\smallskip
\emph{(ii)} We use the fact that $\Psi_\#$ is continuous, 
and the alternative expression: 
$\mathscr{M}(\nu) = \omega(\nu) = \left\{ \eta~:~\exists \R_n \to \infty,  \nu_{\R_n} (\nu)\rightharpoonup \eta \right\}$. In such a case, 
$\Psi_\# \omega(\nu) \subset \omega(\Psi_\# \nu) \subset {\mathscr M}_0$. For $\nu = \delta_\gamma$, using (\ref{Bebuid}), 
one
has $\Psi_\# \omega(\delta_\gamma) = \omega(\Psi_\# \delta_\gamma) = {\mathscr M}(\gamma)$.

\smallskip
\emph{(iii)}  
From (ii) one has $\Psi_\#\mathscr M\subset\mathscr M_0$.
Conversely, Theorem~\ref{M0M} provides for every
$\mu\in\mathscr M_0$ a curve $\gamma_\mu$ with
$\mu\in\Psi_\#\mathscr M(\delta_{\gamma_\mu})\subset\Psi_\#\mathscr M$,
so equality $\Psi_\# \mathscr{M} = {\mathscr M}_0$ holds.
Step (i) also shows that $\mathscr{M} \subset {\rm Inv}(\mathscr{S})$. Conversely, if $\eta \in {\rm Inv}(\mathscr{S})$, then $\nu_\R(\eta) = \eta$, i.e. $\omega(\eta) = \{ \eta \}$. Since
${\mathscr M}_0$ is compact, from Lemma \ref{topolemma}, $\mathscr{M}$
is $\Psi$-compact (non-empty) and it is also forward invariant, therefore it is a statistical attractor.
\end{proof}

\paragraph{Discussion}
In our deterministic setting \eqref{P_eps}, Theorem 
 \ref{thm:M0Mupup} is restricted to $\nu = \delta_\gamma, \gamma \in \Gamma$. Therefore, it provides a straightforward new interpretation of Theorem \ref{M0M} and \eqref{SPdefbis}.
\begin{corollary}[Theorem \ref{M0M} and Theorem \ref{thm:M0Mupup}]
Let ${\cal A}_\gamma:= \mathscr{M}(\delta_\gamma)\big/\!\sim$.
Then ${\cal A}_\gamma$ is a statistical attractor for the Bebutov flow under the quotient topology and
$$
  {\cal A}_\gamma ~\subset \operatorname{Inv}(\mathscr{S})\big/\!\sim ~~\cong \mathscr{M}_0 = \overline{\operatorname{co}}(\mathcal{E}) .
$$
Moreover, the following characterizations hold:
$$
\gamma \in \begin{cases}
\text{\emph{Strong $\SP$}} & \text{if } {\cal A}_\gamma \cong \{\mu \}  \text{ and } \mu \text{ is non-Dirac}, \\
\text{\emph{Weak $\SP$}} & \text{if } {\cal A}_\gamma \text{ is not a singleton}.
\end{cases}
$$
\end{corollary}

Working at the level of $(\mathscr{S}_\tau)$ offers also the possibility to upgrade deterministic regularizations \eqref{P_eps} to random ones. Let us define the random differential equation (RDE):
\be 
\dot x = f_\Theta(x,\tau), x(0)=x_0 \in H,~\Theta:(\Omega,{\cal F},\mathbb{P}) \to \mathscr{E}.
\de 
The solution at time $t$ and initial condition $x_0$ is now a random variable $X_\Theta(t;x_0;\tau)$ and we define a Markov kernel as
$$
K(\tau,A) := \Pr(X_\Theta(t;x_0;\tau) \in A)~A~\mbox{Borel of $H$ i.e.}~K(\tau,dx) \in {\cal P}(H).
$$
The natural question is therefore to ask when $\{K(\tau,\dot)\}_{\tau \geq 0}$ define a measure $\nu \in {\cal P}(\Gamma)$?
The answer is simply: 
\begin{itemize}
    \item for $\mathbb{P}$-almost $\omega \in \Omega$, 
$f_{\Theta(\omega)} \in \V$; see \eqref{H0}.
\item $\tau \mapsto X_{\Theta}(t;x_0;\tau)$ is continuous and bounded and $\omega \mapsto X_{\Theta(\omega)}$ is measurable.
\end{itemize}

Under these conditions, the assumptions of Theorem \ref{thm:M0Mupup} hold true without modification, so that the conclusions of this theorem remain unchanged. In particular, the introduction of randomness into the regularizations has no effect on the asymptotic results, since only the behavior in the limit $\tau\to\infty$ matters. This fact highlights the relevance of the $\Psi$-topology, which captures precisely the asymptotic distribution. As Theorem \ref{thm:M0Mupup} is essentially a quotient-topological reformulation of Theorem \ref{M0M}, the conclusions of Theorem \ref{M0M} also remain valid, up to interpreting the statistical attractors via the quotient topology induced by~$\Psi$. In this sense, we have provided a true upgrading to randomized regularizations.

In principle, upgrading these results to SDEs is largely feasible but more technical. One must work pathwise and relax the boundedness hypothesis. We do not intend to provide this extension and let it for future works.

\section{Perspectives}
This work has presented a clean formalism for defining spontaneous stochasticity. It shows that many structural features of the phenomenon are already present in simple rough, non-Lipschitz systems with a single degree of freedom. In particular, it emphasizes the flexibility of spontaneous stochasticity: within the class of admissible regularizations, arbitrary probability measures can be selected, and the proof is largely constructive. The approach developed here combines measure theory, analysis, and dynamical systems, and appears to be a natural framework for formulating measure-selection principles in regularized ill-posed problems.

To keep the presentation self-contained, we have deliberately restricted the analysis to finite-dimensional systems. In this setting, compactness is available, which allows for a general treatment without introducing the additional technical layers specific to PDEs. This restriction, however, should not obscure the importance of infinite-dimensional systems, and in particular of PDE models. In that setting, the main difficulty is no longer only conceptual: one must also obtain sharp quantitative bounds in order to make the relevant pushforward limits meaningful.

A substantial test of the formalism is provided in the companion work \cite{RSV2025}, where we prove Eulerian spontaneous stochasticity for the passive scalar transport equation introduced in \cite{AAV25}. This result is not a mere illustration: the same definition, Definition \ref{SPdef}, is implemented in a genuinely PDE setting, where the pushforward limit can only be justified through sharp quantitative control. The proof relies crucially on the estimates developed in \cite{AAV25}, and shows that the measure-selection viewpoint survives beyond the finite-dimensional compactness framework used in the present paper.

The preceding analysis, together with the examples discussed above, shows that a central issue is the choice of an ambient measure leading to Strong-$\SP$. This choice may in fact be one of the most delicate steps in proving $\SP$ for a given system. In \cite{RSV2025}, the relevant ambient measure is absolutely continuous with respect to Lebesgue measure, but it is not the canonical normalized measure ${\rm Leb}_\epsilon$. Rather, it is a structured measure $\mathbb{P}_\epsilon$, locally flat on each shell, while the shells themselves are selected according to the log--log scale $\log\log(1/\epsilon)$. This structure is not incidental: it is tightly linked to the multiscale fractal construction of the velocity field in \cite{AAV25}. Thus, the adapted ambient measure is not merely a technical device; it encodes nontrivial and relevant information about the inviscid limit.

Further evidence in the PDE setting has recently been obtained in a numerical study by one of us \cite{W_MMT}. This work again exhibits Eulerian spontaneous stochasticity, now in a one-dimensional fractional nonlinear Schr\"odinger equation, and thereby illustrates that the $\SP$ formalism developed here around the spontaneous-stochasticity triptych extends beyond fluid-related equations.

In light of these results, we view spontaneous stochasticity as a generic mechanism in inviscid systems whose lack of well-posedness is organized around singular sets, rather than as a phenomenon confined to isolated pathological examples. The presence of such sets reflects the rough, typically non-Lipschitz, character of the inviscid dynamics. Our analysis indicates that the emergence of spontaneous stochasticity is governed by two essential ingredients: the structure of these singular sets and the manner in which they are treated by regularization. The central question is therefore to determine their geometry and dynamical role, and to understand which limiting information is selected by different regularization procedures.

%
%

\ack{
We would like to thank S.Thalabard and J.Bec for stimulating discussions on the subject and the IDEX summer school "100 years of cascades", from which this project has started. We also gratefully acknowledge the Calisto team at INRIA for their warm hospitality and continuous support throughout the project. {NV is
funded in part by the EPSRC through grant EP/Z534766/1.}
}




\bibliographystyle{iopart-num}
\bibliography{REFESS}

@article{KARABACAK_ASHWIN_2011, 
title={On statistical attractors and the convergence of time averages}, 
volume={150}, 
DOI={10.1017/S0305004110000642}, 
number={2}, 
journal={Mathematical Proceedings of the Cambridge Philosophical Society}, 
author={Karabacak, O. and Ashwin, P.}, 
year={2011}, 
pages={353--365}}

@book{arnold1999bifurcation,
  author    = {V. I. Arnold and V. S. Afrajmovich and Yu. S. Il'yashenko and L. P. Shil'nikov},
  title     = {Bifurcation Theory and Catastrophe Theory},
  series    = {Encyclopaedia of Mathematical Sciences},
  volume    = {5},
  publisher = {Springer-Verlag},
  year      = {1999},
  address   = {Berlin, Heidelberg},
  isbn      = {978-3-540-65379-0}
}

@article{AAV25,
  title={Anomalous diffusion via iterative quantitative homogenization: an overview of the main ideas},
  author={Armstrong, S. and Vicol, V.},
  journal={arXiv preprint arXiv:2503.11744},
  year={2025},
  archivePrefix={arXiv},
  eprint={2503.11744},
  primaryClass={math.AP}
}

@article{ambrosio2004transport,
  title={{Transport equation and Cauchy problem for BV vector fields}},
  author={Ambrosio, L.},
  journal={Inventiones mathematicae},
  volume={158},
  number={2},
  pages={227--260},
  year={2004},
  publisher={Springer},
  doi={10.1007/s00222-004-0367-2}
}

@misc{RSV2025,
  author        = {Ruffenach, Wandrille and Simonnet, Eric and Valade, Nicolas},
  title         = {{Spontaneous stochasticity in the Armstrong--Vicol passive scalar}},
  year          = {2025},
  eprint        = {2509.15683},
  archivePrefix = {arXiv},
  primaryClass  = {math-ph},
  doi           = {10.48550/arXiv.2509.15683}
}

@book{bruckner1978,
  title={Differentiation of Real Functions},
  author={Bruckner, A.M.},
  series={Lecture Notes in Mathematics},
  volume={659},
  year={1978},
  publisher={Springer Berlin, Heidelberg},
  doi={10.1007/BFb0069821},
  isbn={978-3-540-08910-0}
}

@article{giorgi1992dini,
  title={{D}ini derivatives in optimization — Part I},
  author={Giorgi, G. and Koml{\'o}si, S.},
  journal={Rivista di matematica per le scienze economiche e sociali},
  volume={15},
  pages={3--30},
  year={1992},
  publisher={Springer},
  doi={10.1007/BF02086523}
}

@article{Richardson1926,
  author    = {L.F. Richardson},
  title     = {Atmospheric Diffusion Shown on a Distance-Neighbour Graph},
  journal   = {Proceedings of the Royal Society of London. Series A, Containing Papers of a Mathematical and Physical Character},
  year      = {1926},
  volume    = {110},
  pages     = {709-737},
  doi       = {10.1098/rspa.1926.0043}
}

@article{Feigenbaum1976,
  author    = {M.J. Feigenbaum},
  title     = {Universality in Complex Discrete Dynamics},
  journal   = {Los Alamos Theoretical Division Annual Report},
  year      = {1976},
  pages     = {98-102},
}

@article{lorenz63,
  title={Deterministic nonperiodic flow},
  author={Lorenz, E.N.},
  journal={J. Atmos. Sci.},
  volume={20},
  number={2},
  pages={130--141},
  year={1963},
  publisher={American Meteorological Society}
}

@article{AM_Ra23,
    author = {Mailybaev, A.A. and Raibekas, A.},
    title = {Spontaneous Stochasticity and Renormalization Group in Discrete Multi-scale Dynamics},
    journal = {Commun. Math. Phys.},
    volume = {401},
    pages  = {2643--2671},
    year = 2023
}

@article{Flandoli2009,
  author    = {F. Flandoli},
  title     = {Remarks on uniqueness and strong solutions to deterministic and stochastic differential equations},
  journal   = {Metrika},
  year      = {2009},
  volume    = {69},
  pages     = {101-123},
  doi       = {10.1007/s00184-008-0210-7}
}

@book{Hartman2002,
  author    = {P. Hartman},
  title     = {Ordinary Differential Equations},
  edition   = {Second},
  year      = {2002},
  publisher = {SIAM},
  series    = {Classics in Applied Mathematics},
  volume    = {38},
  isbn      = {978-0-89871-510-1},
  doi       = {10.1137/1.9780898719222}
}

@book{Walter1970,
  author    = {W. Walter},
  title     = {Differential and Integral Inequalities},
  year      = {1970},
  publisher = {Springer},
  series    = {Ergebnisse der Mathematik und ihrer Grenzgebiete},
  volume    = {55},
  doi       = {10.1007/978-3-642-86405-6}
}

@book{Bahouri2011,
  author    = {H. Bahouri and J.-Y. Chemin and R. Danchin},
  title     = {Fourier Analysis and Nonlinear Partial Differential Equations},
  year      = {2011},
  publisher = {Springer},
  series    = {Grundlehren der mathematischen Wissenschaften},
  volume    = {343},
  doi       = {10.1007/978-3-642-16830-7}
}

@book{Lasota1994,
  author    = {A. Lasota and M.C. Mackey},
  title     = {Chaos, Fractals, and Noise: Stochastic Aspects of Dynamics},
  edition   = {Second},
  year      = {1994},
  publisher = {Springer},
  series    = {Applied Mathematical Sciences},
  volume    = {97},
  doi       = {10.1007/978-1-4612-4286-4}
}

@article{AM_Rb23,
    author = {Mailybaev, A.A. and Raibekas, A.},
    title = {Spontaneous Stochasticity {A}rnold's {C}at},
    journal = {Arnold Math J.},
    volume = {9},
    pages  = {339--357},
    year = 2023
}

@article{AM_PRF26,
  author    = {A. A. Mailybaev},
  title     = {RG theory of spontaneous stochasticity for Sabra model of turbulence},
  journal   = {Phys. Rev. Fluids},
  volume    = {11},
  issue     = {3},
  pages     = {034605},
  year      = {2026},
  publisher = {American Physical Society}
}

@article{AM_NLN25,
  author    = {A. A. Mailybaev},
  title     = {RG approach to the inviscid limit for shell models of turbulence},
  journal   = {Nonlinearity},
  volume    = {38},
  pages     = {085010},
  year      = {2025},
  publisher = {IOP Publishing}
}

@book{Agarwal_Lak, 
	title={Uniqueness and Nonuniqueness Criteria for Ordinary Differential Equations},
	author={Agarwal, R.P. and Lakshmikantham, V.}, 
	year={1993}, 
	publisher={World Scientific}, address={Singapore}, 
	series={Series in Real Analysis}, 
	volume={6}, isbn={978-981-02-1357-2}
}

@article{AM_24end,
  author    = {A. A. Mailybaev},
  title     = {RG analysis of spontaneous stochasticity on a fractal lattice: stability and bifurcations},
  journal   = {J. Stat. Phys.},
  volume    = {192},
  pages     = {37},
  year      = {2025},
  publisher = {Springer}
}

@article{Lorenz69,
author = {Lorenz, E.N.},
title = {The predictability of a flow which possesses many scales of motion},
journal = {Tellus},
volume = {21},
number = {3},
pages = {289-307},
year  = 1969
}

@article{Dombre1998,
  author    = {Dombre, T. and Gilson, J.-L.},
  title     = {{Intermittency, chaos and singular fluctuations in the mixed Obukhov–Novikov shell model of turbulence}},
  journal   = {Physica D: Nonlinear Phenomena},
  volume    = {111},
  pages     = {265--287},
  year      = {1998},
  doi       = {10.1016/S0167-2789(97)00202-5}
}

@article{Eyink_Bandak20,
  title = {Renormalization group approach to spontaneous stochasticity},
  author = {Eyink, G.L. and Bandak, D.},
  journal = {Phys. Rev. Res.},
  volume = {2},
  issue = {4},
  pages = {043161},
  numpages = {24},
  year = {2020},
  month = {Oct}
}

@article{Simon_Jeremie_AM20,
year = {2020},
volume = {3},
number = {122},
author = {Thalabard, S. and Bec, J. and Mailybaev, A.A.},
title = {From the butterfly effect to spontaneous stochasticity in singular shear flows},
journal = {Commun. Phys.}
}

@article{GHR_SLDP2001,
  author    = {Mihai Gradinaru and Samuel Herrmann and Bernard Roynette},
  title     = {{A Singular Large Deviations Phenomenon}},
  journal   = {Annales de l'Institut Henri Poincar{\'e} Probabilit{\'e}s et Statistiques},
  volume    = {37},
  number    = {5},
  pages     = {555--580},
  year      = {2001},
  publisher = {Éditions scientifiques et médicales Elsevier SAS},
  issn      = {0246-0203},
  doi       = {10.1016/S0246-0203(01)01075-5}
}

@article{MMRG26,
  author        = {Alexei A. Mailybaev and Luca Moriconi},
  title         = {Renormalization-group perspective on spontaneous stochasticity},
  journal       = {arXiv preprint},
  year          = {2026},
  eprint        = {2602.24221},
  archivePrefix = {arXiv},
  primaryClass  = {nlin.CD},
  doi           = {10.48550/arXiv.2602.24221}
}

@article{W_MMT,
  author        = {W. Ruffenach},
  title         = {{Spontaneous stochasticity and anomalous dissipation in collapsing wave turbulence}},
  journal       = {arXiv preprint},
  year          = {2026},
  eprint        = {},
  archivePrefix = {arXiv},
  primaryClass  = {nlin.CD},
  doi           = {}
}

@article{Drivas21,
    author = {Drivas, T.D. and Mailybaev, A.A.},
    title = "{'Life after death' in ordinary differential equations with a non-Lipschitz singularity}",
    journal = {Nonlinearity},
    volume = {34},
    number = {},
    pages = {2296},
    year = {2021},
    month = {}
    }

@article{Drivas24,
    author = {Drivas, T.D. and Mailybaev, A.A. and Raibekas, A.},
    title = "{Statistical determinism in non-Lipschitz dynamical systems}",
    journal = {Ergodic Theo. and Dyn. Syst.},
    volume = {44},
    number = {},
    pages = {1856--1884},
    year = {2024},
    month = {}
    }

@article{Maily2012,
  author    = {A.A. Mailybaev},
  title     = {Renormalization and Universality of Blowup in Hydrodynamic Flows},
  journal   = {Physical Review E},
  year      = {2012},
  volume    = {85},
  number    = {6},
  pages     = {066317},
  doi       = {10.1103/PhysRevE.85.066317}
}

@article{Maily2016,
  author    = {A.A. Mailybaev},
  title     = {Spontaneous Stochasticity of Velocity in Turbulence Models},
  journal   = {Multiscale Modeling \& Simulation},
  year      = {2016},
  volume    = {14},
  number    = {1},
  pages     = {96-112},
  doi       = {10.1137/15M1012451}
}

@article{Drivas17,
    author = {Drivas, T.D. and Eyink, G.L.},
    title = "{A Lagrangian fluctuation–dissipation relation for scalar turbulence}",
    journal = {J. Fluid Mech.},
    volume = {829},
    number = {},
    pages = {153--189},
    year = {2017},
    month = {}
    }

@article{Gawedzki98,
    author = {Bernard, D. and Gaw\c{e}dzki, K. and Kupiainen, A.},
    title = "{Slow modes in passive advection}",
    journal = {J. Stat. Phys.},
    volume = {90},
    number = {3},
    pages = {519--569},
    year = {1998},
    month = {}
    }

@article{Chaves2003,
    author = {Chaves, M. and Gawedzki, K. and Horvai, P. and Kupiainen, A. and Vergassola, M.},
    title = "{Lagrangian Dispersion in Gaussian Self-Similar Velocity Ensembles}",
    journal = {J. Stat. Phys.},
    volume = {113},
    number = {},
    pages = {643--692},
    year = {2003},
    month = {}
    }

\appendix
\section{Appendix}

\subsection{Equivalence of Definitions: Proposition \ref{SPDEFS}}\label{SP_LSP}
We justify Proposition~\ref{SPDEFS}. Figure~\ref{tricho} illustrates the mutually exclusive scenarios.  
In the definition of weak $\SP$, it is not necessary to specify whether the subsequential limits are Dirac or non-Dirac.  
We begin by explaining how a Dirac weak limit can still be compatible with a lack of selection principle.

\subsection{$\delta\text{-}LSP$: Dirac weak limit with lack of selection principle}\label{DiracLSP}

We describe the mechanism responsible for
$\delta\text{-}LSP := LSP \setminus \{ LSP \setminus \delta \}$;  
see Definition~\ref{SPdef}.  
This situation corresponds to the absence of a selection principle,  
but only on a subset whose relative Lebesgue measure vanishes in the inviscid limit.  One can indeed provide a very simple example where
a regularization curve $\gamma$ can be discontinuous everywhere but still converges in measure to a Dirac mass. Take 
$\gamma(s) = \gamma_0 + \mathbf{1}_{\mathbb{Q}}(1/s)$, 
then $\gamma_\# \operatorname{Leb}_\epsilon \rightharpoonup \delta_{\gamma_0}$.
\\

There are indeed situations in which $\gamma$ is
\emph{approximately continuous} at zero, i.e.
has an approximate inviscid limit, say $\gamma_0$. Equivalently,
$\gamma$ converges in density to $\gamma_0$
with an exceptional set having vanishing relative Lebesgue measure. There is therefore, strictly speaking, a lack of selection principle that we call $\delta$-$LSP$, but on a set of vanishing (relative) density. 
This is the following lemma:

\begin{lemma}\label{aaecontinuity}
Assume that, there is a $\gamma_0 \in H$ such that $\gamma_\# \operatorname{Leb}_\epsilon \rightharpoonup \delta_{\gamma_0}$, then there exists a Borel set $B$ such that for all $\delta>0$
\be 
\operatorname{Leb}_\epsilon\left( \|\gamma(s)-\gamma_0 \|> \delta 
\right) \to 0~\text{as}~\epsilon \to 0.
\de 
\end{lemma}
\begin{proof}
For all $F \in C_b(H;\mathbb{R})$, $\langle \gamma_\# \operatorname{Leb_\epsilon},F\rangle \to F(\gamma_0)$. We can therefore exhibit positive test functions $0 \leq F_n \leq 1 \in C_b(H;\mathbb{R}), n \in \mathbb{N}$ such that $F_n(\gamma_0)=0$, $F_n(\gamma(s)) = 1$ when $\| \gamma(s)-\gamma_0 \| > 1/n$,
 and
$\langle \gamma_\# \operatorname{Leb_\epsilon},F_n\rangle \to F_n(\gamma_0) = 0$. Define
$$
C_{n}:= \left\{ s ~:~\| \gamma(s)-\gamma_0 \| > \frac{1}{n} \right\}.
$$
One has,
\begin{align*}
\frac{\operatorname{Leb}_\epsilon(C_n)}
{\operatorname{Leb}(C_n \cap [0,\epsilon])}
\int_{C_n \cap [0,\epsilon]}
\underbrace{F_n(\gamma(s))}_{=1} \, ds
= & \\
\langle \gamma_\# \operatorname{Leb}_\epsilon, F_n \rangle 
\quad -
\frac{\operatorname{Leb}_\epsilon(C_n^c)}
{\operatorname{Leb}(C_n^c \cap [0,\epsilon])}
\int_{C_n^c \cap [0,\epsilon]}
\underbrace{F_n(\gamma(s))}_{\in [0,1]} \, ds 
&\leq
\langle \gamma_\# \operatorname{Leb}_\epsilon, F_n \rangle
\to 0.
\end{align*}
Therefore $\operatorname{Leb}_\epsilon(C_n) \to 0$ (and $\operatorname{Leb}_\epsilon(C_n^c) \to 1$) as $\epsilon \to 0$.
The technical difficulty is that this convergence does depend on $F_n$ above.
We need a diagonal argument. One can pick a decreasing sequence $\epsilon_n \to 0$ such that $\operatorname{Leb}_{\epsilon_n}(C_n) \leq e^{-n}$ for instance. Then we define
$$
B := \bigcup_{n \geq 1} [0,\epsilon_n] \setminus C_n.
$$
From this definition, for $\epsilon \in (0,\epsilon_n)$, one has $(0,\epsilon_n) \setminus C_n \subset [0,\epsilon] \cap B$, and
$$
1 \geq \operatorname{Leb}_\epsilon(B) \geq 1-\operatorname{Leb}_\epsilon(C_n) \geq
1-\operatorname{Leb}_{\epsilon_n}(C_n) \geq 1-e^{-n} \to 1.
$$
To conclude, for all $s \in (0,\epsilon_n)$, one has $\|\gamma(s)-\gamma_0\| \leq 1/n$ and taking $n=n(s):=\min \{ k~:~s \leq \epsilon_k \}$ and going to the limit, $\gamma(s) \to \gamma_0$ for $s \to 0 \in B$.
\end{proof}

For clarity, we now discuss in detail a simple example.  
Let $\gamma$ be a regularization curve (see~\eqref{gamma}), and assume --without making this link explicit -- that it corresponds to the regularization of an inviscid system~\eqref{P_0}.  
Let ${\cal K}_\epsilon \subset \mathbb{R}^+$ satisfy
$$
{\cal K}_\epsilon \cup {\cal K}_\epsilon^c = [0,\epsilon].
$$
As a concrete example, take
$$
K_\epsilon = \bigcup_{n \geq N(\epsilon)} 
\left[\frac1n-e^{-n},\frac1n+e^{-n}\right] \cap [0,\epsilon],  
\quad \lim_{\epsilon \to 0} N(\epsilon) = \infty.
$$
Let $\gamma_0, \gamma_1 \in {\mathscr S}_0 \subset H$ with $\gamma_0 \neq \gamma_1$,  
and define a smooth curve $\gamma: s \mapsto \gamma(s) \in H$ so that  
$\gamma(\frac1n) = \gamma_1$ and  
$\left. \gamma\right|_{{\cal K}_\epsilon^c} = \gamma_0$. 
\\

For $F \in C_b(H;\mathbb{R})$ we have
$$
\langle \gamma_\# \operatorname{Leb}_\epsilon, F \rangle 
= \frac{1}{\epsilon} \int_{{\cal K}_\epsilon^c} F(\gamma_0) \, ds 
+ \frac1\epsilon \int_{{\cal K}_\epsilon} F(\gamma(s)) \, ds
= \frac{|{\cal K}_\epsilon^c|}{\epsilon} F(\gamma_0)  
+ \frac{|{\cal K}_\epsilon|}{\epsilon} \frac{1}{|{\cal K}_\epsilon|} 
\int_{{\cal K}_\epsilon} F(\gamma(s)) \, ds.
$$
Since $F$ is bounded and $|{\cal K}_\epsilon| + |{\cal K}_\epsilon^c| = \epsilon$, it follows that
$$
\left| \langle \gamma_\# \operatorname{Leb}_\epsilon - \delta_{\gamma_0},F \rangle \right| 
\leq  C \left| \frac{|{\cal K}_\epsilon^c|}{\epsilon} -1 \right| 
+ C \frac{|{\cal K}_\epsilon|}{\epsilon} 
= 2 C \frac{|{\cal K}_\epsilon|}{\epsilon}.
$$
Moreover,
$$
|{\cal K}_\epsilon| = 2 \sum_{n \geq N(\epsilon)} e^{-n} 
= \frac{2}{e-1} e^{-(N(\epsilon)-1)}.
$$
Choosing, for instance,
$$
N(\epsilon) = \lfloor \epsilon^{-a} \rfloor, \quad a > 0,
$$
we deduce that
$$
\gamma_\# \operatorname{Leb}_\epsilon \rightharpoonup \delta_{\gamma_0}.
$$

In this construction, the regularization curve $\gamma$ oscillates rapidly as $\epsilon \to 0$.  
One can extract a subsequence $\epsilon_n = \frac1n$ with $\gamma(\epsilon_n) \to \gamma_1$,  
and another subsequence $\epsilon_n'$ with $\gamma(\epsilon_n') = \gamma_0$.  
Thus, although the regularization converges to the Dirac measure $\delta_{\gamma_0}$,  
there remains a lack of selection principle.

\subsection{Proof of Proposition \ref{SPDEFS}}
We indeed prove a slightly stronger version first, namely that
$LSP \Longleftrightarrow \LSPO$.
\paragraph{$LSP \Rightarrow \LSPO$}
\begin{proof}
It is simple: one has two subsequences $x^{\epsilon_n},x^{\epsilon_n'} \to x_1,x_2$ in $C([0,T];H)$, and since $x_1 \neq x_2$ there exists a $t \in (0,T)$ such that $x_1(t) \neq x_2(t)$. Defining ${\cal O}(x) = ||x(t)-x_1(t)||$, it is continuous so that ${\cal O}(x^{\epsilon_n}) \to 0$ and
${\cal O}(x^{\epsilon_n'}) \to {\cal O}(x_2(t)) > 0$. Therefore
$\LSPO$ holds.
\end{proof}
\bigskip

\paragraph{$\LSPO \Rightarrow LSP$}
The added difficulty is to recover convergence in $C([0,T];H)$ from the compactness of the family of regularized trajectories.

\begin{proof}
Let us assume by contradiction that $\LSPO$ holds while $LSP$ does not. Let us denote
$
\Psi_\epsilon(s):=\phi_s^\epsilon(x_0).
$
We show in the following that $\Psi_\epsilon$ converges uniformly in $C([0,T];H)$, giving a contradiction.
We have
$
\Psi_\epsilon(s)
=
x_0+\int_0^s f(\Psi_\epsilon(s'),\epsilon)\,ds',
$
so that
$
\|\Psi_\epsilon(s)\|
\leq
\|x_0\|+T\sup_x\|f(x,\epsilon)\|
\leq C
$
uniformly in $\epsilon$ and $s$. Similarly, since
$
\dot{\Psi}_\epsilon(s)=f(\Psi_\epsilon(s),\epsilon),
$
one has
$
\|\dot{\Psi}_\epsilon(s)\|\leq C.
$
The time derivatives being uniformly bounded, one obtains uniform equicontinuity. One can therefore use the Arzel\`a--Ascoli theorem, which shows that every sequence $\epsilon_n\to0$ admits a subsequence $\epsilon_{n_j}$ such that $\Psi_{\epsilon_{n_j}}$ converges uniformly in $C([0,T];H)$ to some $\Psi^\star$. We can write
$
\Psi_{\epsilon_{n_j}}(s)
=
x_0+\int_0^s f(\Psi_{\epsilon_{n_j}}(s'),\epsilon_{n_j})\,ds'.
$
By the uniform convergence of the subsequence, together with
$
\|f(\cdot,\epsilon_{n_j})-f_0(\cdot)\|_\infty\to0,
$
one obtains
$
\Psi^\star(s)
=
x_0+\int_0^s f_0(\Psi^\star(s'))\,ds'.
$
Thus every subsequential limit is a solution of $({\cal P}_0)$.

Since $LSP$ does not hold, two convergent subsequences cannot have distinct limits. Therefore, all convergent subsequences have the same limit, denoted by $\Psi^\star$. It follows that
$$
\Psi_\epsilon\to\Psi^\star
\quad\text{in }C([0,T];H).
$$
Indeed, otherwise there would exist $\eta>0$ and a sequence $\epsilon_n\to0$ such that
$
\|\Psi_{\epsilon_n}-\Psi^\star\|_{C([0,T];H)}\geq\eta.
$
By the Arzel\`a--Ascoli theorem, this sequence would admit a uniformly convergent subsequence, whose limit must again be $\Psi^\star$, giving a contradiction.

Using the continuity of the observable, one therefore has
$
{\cal O}(\phi_t^\epsilon(x_0))
=
{\cal O}(\Psi_\epsilon(t))
\to
{\cal O}(\Psi^\star(t)).
$
Hence
$$
\limsup_{\epsilon\to0}{\cal O}(\gamma(\epsilon))
=
\liminf_{\epsilon\to0}{\cal O}(\gamma(\epsilon)),
$$
which contradicts $\LSPO$.
\end{proof}

\paragraph{$\neg \SP \Longleftrightarrow \text{Selection principle} \bigcup \delta\text{-}LSP$}\label{bigissue}
\begin{proof}
We notice that the negation of $\SP$ is 
$\neg({\rm strong}~ \SP)$ and $\neg({\rm weak}~\SP)$.
Since the two notions are mutually exclusive. Then $\neg({\rm strong}~ \SP)$ is
either weak $\SP$ or $\gamma_\# {\rm Leb}_\epsilon \rightharpoonup \mu$ with $\mu$ a Dirac measure. Now $\neg({\rm weak}~\SP)$ is either strong $\SP$ or
$\gamma_\# {\rm Leb}_\epsilon \rightharpoonup \mu$ with $\mu$ a Dirac measure.
Therefore $\neg \SP$ is just weak convergence to a single Dirac mass, namely using Lemma \ref{aaecontinuity}, 
this is either a classical selection principle or
a lack of selection principle but occurring on a set whose relative measure w.r.t. Lebesgue is going to zero in the limit $\epsilon \to 0$; see Lemma \ref{aaecontinuity}.
\end{proof}
\begin{lemma}\label{OnotconstLemma}
Assume $H$ is finite-dimensional. If $\LSPO$ holds for the observable ${\cal O}\in C_b(H;\mathbb R)$ at
$(t,x_0)$, then ${\cal O}$ is not constant on ${\mathscr S}_0$.
\end{lemma}
\begin{proof}
Assume ${\cal O}\equiv c$ on ${\mathscr S}_0$. Since every accumulation point
of $\gamma(s)$ as $s\to0$ lies in the compact set ${\mathscr S}_0$ (see Step~2
of the proof of Theorem~\ref{M0M}), one has
${\rm dist}(\gamma(s),{\mathscr S}_0)\to0$. By continuity of ${\cal O}$ and
compactness of ${\mathscr S}_0$, it follows that
${\cal O}(\gamma(s))\to c$, hence
$\liminf_{\epsilon\to0}{\cal O}(\gamma(\epsilon))
=\limsup_{\epsilon\to0}{\cal O}(\gamma(\epsilon))=c$,
contradicting $\LSPO$.
\end{proof}

\subsection{Algebraic reparameterizations}\label{Alg}
We discuss the fact that the problem in its native form 
with $\epsilon \to 0$ can indeed be reduced, under algebraic reparameterization,
to the study of genuine Birkhoff averages as $\tau \to \infty$.

We introduce the class of admissible algebraic reparameterizations
\begin{equation}\label{AlgDef}
\alg:=
\left\{
\begin{aligned}
g\in C^2 \;:\;& g>0,\quad g'<0,\quad \lim_{r\to\infty} g(r)=0,\\
& \exists\, a,b>0 \text{ such that } g(r)=r^{-a}e^{h(r)},\\
& \sup_{r\ge r_0}|h(r)|<\infty,\qquad
|h'(r)|\le \frac{C}{r^{1+b}}
\end{aligned}
\right\},
\end{equation}
as well as the class
\begin{equation}\label{BlgDef}
\blg:=
\left\{
\begin{aligned}
\beta\in C^2 \;:\;& \beta>0,\quad \beta'>0,\\
& \exists\, a,b>0 \text{ such that } \beta(r)=r^{a}e^{h(r)},\\
& \sup_{r\ge r_0}|h(r)|<\infty,\qquad
|h'(r)|\le \frac{C}{r^{1+b}}
\end{aligned}
\right\}.
\end{equation}

\begin{nnproperty*}
Let $\gamma(s):=\phi_t^s(x_0)$, and let $g\in\alg$.
Then the following invariance property holds:
$$
\gamma_\#\mathrm{Leb}_\epsilon \rightharpoonup \mu
\quad\text{as }\epsilon\to0
\qquad\Longleftrightarrow\qquad
(\gamma\circ g)_\#\mathrm{Leb}_\tau \rightharpoonup \mu
\quad\text{as }\tau\to\infty.
$$
Moreover, the same equivalence holds for all subsequential weak limits.
In addition,
$$
\alg\circ\blg=\alg.
$$
Consequently, for every $\gamma\in\Gamma$ and every $g\in\blg$, one has
$\gamma\circ g\in\Gamma$ and
$$
{\mathscr M}(\gamma\circ g)={\mathscr M}(\gamma).
$$
\end{nnproperty*}
In other words, inviscid statistics are unchanged  upon suitable reparameterization, one has for instance
$$
\lim_{R\to\infty}\frac{1}{R}\int_0^R F\bigl(e^{i\theta^\alpha}\bigr)\,d\theta
=
\frac{1}{2\pi}\int_0^{2\pi}F\bigl(e^{i\theta}\bigr)\,d\theta,
$$
independently of $\alpha$, for smooth enough observables $F$ on $\mathbb S^1$.
This is a manifestation of the Riemann-Lebesgue lemma and/or stationary phase.

Moreover, this property induces the straightforward equivalence relation on $\Gamma$:
$$
\gamma_1\sim\gamma_2
\qquad\Longleftrightarrow\qquad
\exists\, g\in\blg \text{ such that } \gamma_1=\gamma_2\circ g.
$$
\begin{proof}
We consider a $C^2$ diffeomorphism $g: \mathbb{R}^+ \to 
\mathbb{R}^+, \tau \mapsto g(\tau)$ such that
$\lim_{\tau \to \infty} g(\tau) = 0$ and is strictly decreasing $g' < 0$.
One assume without loss of generality that $\lim_{s \to \infty} \gamma(s)$ exists which correspond to the limit where $\epsilon$ for $({\cal P}_\epsilon)$ is large. We can thus extend $F\circ \gamma \circ g$ by continuity at  $\tau=0$.
Let us define $F_g(\tau) = F(\gamma \circ g(\tau))$ for all $\tau \geq 0$.
Before hand, we use the mean value theorem:
one has for some $\theta > 1$, $g(\theta \tau) - g(\tau) = (\theta-1) \tau g'(\zeta),
\zeta \in [\tau,\theta \tau]$ so that $\tau g'(\zeta) \to 0$ and in particular 
\be \label{tdg}
\lim_{\tau \to \infty} \tau g'(\tau) = 0.
\de 
First notice that after the change of variable $s = g(\tau)$ with $\epsilon = g(R), R = g^{-1}(\epsilon)$, it gives 
$$
\frac{1}{\epsilon} \int_0^\epsilon F(\gamma(s))ds = J(R),
$$
where
\be 
J(R):=  {\displaystyle -\frac{1}{g(R)} \int_R^\infty F_g(\tau) g'(\tau) d\tau},~\text{and we write}~
B(R) :=  {\displaystyle \frac{1}{R} \int_0^R F_g(\tau) d\tau}.
\de 
\subsection{$B \to \rho \Rightarrow J \to \rho$}
We can express $J$ as a function of $B$ using (\ref{tdg}) like
$$
J(R) = -\frac{1}{g(R)}\int_R^\infty (\tau B(\tau))' g'(\tau) d\tau =
\Omega(R) B(R) + \frac{1}{g(R)} \int_R^\infty \tau g''(\tau) B(\tau) d\tau,~~\Omega(R):= \frac{R g'(R)}{g(R)}.
$$
By remarking that $\Omega = -\frac{1}{g(R)} \int_R^\infty (\tau g'(\tau))' d\tau$, $1 = -\frac{1}{g(R)} \int_R^\infty g'(\tau)d\tau$, and $\tau g''(\tau) = (\tau g'(\tau))' - g'(\tau)$, one has
$$
J(R)-\rho = \Omega(R)(B(R)-\rho) + \int_R^\infty \omega_R'(\tau) (B(\tau)-\rho) d\tau~\text{with}~\omega_R(\tau) = \frac{\tau g'(\tau)-g(\tau)}{g(R)}.
$$
It yields
$$
|J(R)-\rho| \leq |\Omega(R)| ~|B(R)-\rho| +
\left( \int_R^\infty |\omega_R'| ~d\tau \right) ~\sup_{s \in [R,\infty)} |B(s)-\rho|.
$$
It imposes that there exist constants $C$ independent of $R$ such that
\be 
|\Omega(R)| \leq C~\text{and}~\int_R^\infty \left|\frac{\tau g''(\tau)}{g(R)}
\right| d\tau \leq C.
\de 
\end{proof}
\subsection{Regularization encoding all statistical behaviors: Proposition \ref{CHOCBAR2}}\label{Chocbar2proof}
\begin{proof}[Proof of Proposition \ref{CHOCBAR2}]
Let $b_n := (x_n,y_n,\theta_n) \in B:={\mathscr S}_0 \times {\mathscr S}_0 \times [0,1]$ a sequence which is sequentially dense, meaning that for all $b = (x,y,\theta) \in B$ one has a convergent subsequence $b_{n_k} \to b$. Let $T_n \to \infty$ a given sequence to choose. Call $\chi_n:\mathbb{R}^+ \to [0,1]$ a partition of unity such that $\sum_n \chi_n(s) = 1$ and
$\chi_n(s)=1, s \in [T_n+\delta,T_{n+1}-\delta]$ and zero if not. We define $\gamma_{\rm un}(\tau) = \sum_n \chi_n(\tau) \gamma_{x_n,y_n,\theta_n}(\tau)$ where $\gamma_{x,y,\theta}$ is the regularization curve constructed in the proof of Theorem \ref{M0M}, Section \ref{allprob}.
We must choose $T_n$ so that $T_{n+1}-T_n$ is large enough to distinguish at least one period of $a_\theta$ in (\ref{controlswitch}).
In the interval $(T_n,T_{n+1})$, the curve is approximating $\theta_n x_n + (1-\theta_n) y_n$. Due to the sequential density, for all $(x,y,\theta) \in B$, one can exhibit subsequences $n_k \to \infty$ such that $x_{n_k},y_{n_k},\theta_{n_k} \to x,y,\theta$, namely $\mu_{T_{n_k}} := \frac{1}{T_{n_k}} \int_0^{T_{n_k}} \delta_{\gamma_{\rm un}(\tau)} d\tau \rightharpoonup  \theta \delta_{x} + (1-\theta) \delta_{y}$. The set of accumulation points is therefore
$\overline{\operatorname{co}}({\cal E}) = {\mathscr M}_0$.
\end{proof}
\subsection{Proof of Theorem \ref{CN}} \label{DiniApp}
Although we do not use Dini derivatives explicitly, it is closely related to
Definition \ref{DiniME}. We just recall very few well-known properties.
\begin{definition}
Let $f(x)$ be a  real valued function defined in $(a,b)$, the four Dini derivatives at $x_0$ are defined as
\begin{align}
D^\pm f(x_0) = \limsup_{x \to x_0^\pm} \frac{f(x)-f(x_0)}{x-x_0}, \nonumber \\ 
D_\pm f(x_0) = \liminf_{x \to x_0^\pm} \frac{f(x)-f(x_0)}{x-x_0}. \nonumber
\end{align}
$\pm \infty$ is allowed.
\end{definition}
We state some Dini properties for general functions (not necessarily continuous)
\begin{itemize}
\item if $f$ is continuous in $(a,b)$, and $D(x) \in \{D_\pm^\pm f(x)\} > 0 (<0),~\forall x \in (a,b)$ then $f$ is
strictly increasing (decreasing). In other words, one needs only
one of the Dini derivatives. 
\item Let $f$ be continuous on $(a,b)$ with at least one of the 
Dini derivatives bounded  (e.g. $|D^+ f(x)| \leq C, \forall x \in (a,b)$)
then $f$ is  Lipschitz on $(a,b)$.
\end{itemize}
Deeper results can be found in \cite{giorgi1992dini,bruckner1978}.
\\
The proof of Theorem \ref{CN} is simple, since it is mostly a rephrasing of the non-Lipschitz
property using Dini-like quantities but translated in the "Osgood framework".
\begin{proof}
We show the contrapositive. We restrict $x$ to a compact set
$M \subset \mathbb{R}^n$. Therefore, $\exists C, \forall (x,v) \in M \times \mathbb{S}^{n-1}, 
\Lambda^+_\Omega(x,v) \leq C < + \infty$. We need to show that in this case
$\dot x =f(x), x(0)=x_0 \in M$ has a unique solution. Note that $f$ might not be necessarily Lipschitz. Consider $v = \frac{y-x}{||y-x||}$. Using the uniform bound above, for small enough $t>0$, one has
$\langle f(x+tv)-f(x),v \rangle \leq C \Omega(t)$, and taking $t=||y-x||$  gives the one-sided inequality
$$
\langle f(y)-f(x),y-x \rangle \leq C ||y-x|| \Omega(||y-x||).
$$
Denote $g(||y-x||^2) = ||y-x|| \Omega(||y-x||)$. Then by the one-sided 
Osgood Lemma, uniqueness
occurs if $\Int_{0^+} \frac{dz}{g(z)} = \Int_{0^+} \frac{dz}{\Omega(z)} = +\infty$ which is just the hypothesis of Theorem \ref{CN}. 
\end{proof}
Note that the nonautonomous case holds
as well: this is Giuliano's uniqueness Theorem with a time-dependent measurable constant (see Theorem 3.5.1 in 
\cite{Agarwal_Lak}).

\subsection{Example: $\SP_{\rm turb}$ without $\SP$}\label{tanh_eps}
We illustrate here Property \ref{modulus_spst} on  a simple example.
We consider the particular family,  $\rho_s = (g(s) {\rm I})_\# \rho$ for a given fixed probability measure $\rho$ on $H$, namely the random variable $\chi_s = g(s) \chi$.
Consider
\begin{equation}\label{tanhmodel}
\dot x = \tanh \frac{x}{\epsilon},
\qquad
x(0)=0.
\end{equation}
This system develops a discontinuity at $x=0$ as $\epsilon\to0$, and indeed
$f\notin\V$. However, this is irrelevant for the present illustration.
For every fixed $\epsilon>0$, the solution of \eqref{tanhmodel} with initial condition $x(0)=0$ is uniquely given by $x_\epsilon(t)\equiv0$.
Hence the deterministic vanishing-$\epsilon$ limit selects the Dirac mass
$$
\mu=\delta_0,
$$
so that $\SP$ does not occur. Now randomize the initial condition according to
$$
x(0)=g(\epsilon)u,
\qquad
u\sim\rho.
$$
The solution can be written explicitly as
$$
X_\epsilon(t;g(\epsilon)u)
=
\epsilon \operatorname{arsinh}
\left(
\exp\!\left(\frac{t}{\epsilon}\right)
\sinh\!\frac{g(\epsilon)u}{\epsilon}
\right).
$$
It is easy to infer that there is a transition according to whether
$$
e^{t/\epsilon}\frac{g(\epsilon)}{\epsilon}\to0
\qquad\text{or}\qquad
e^{t/\epsilon}\frac{g(\epsilon)}{\epsilon}\to\infty.
$$
In the first case, one has a direct illustration of Property \ref{modulus_spst} with 
$$
\mu=\mu_{\rm turb}=\delta_0.
$$
In the second case, \eqref{omegas} does not hold:
$$
X_\epsilon(t;g(\epsilon)u)\to t\,\operatorname{sign}(u),
$$
so that
$$
\mu=\delta_0,
\qquad
\mu_{\rm turb}
=
\rho((0,\infty))\,\delta_t
+
\rho((-\infty,0))\,\delta_{-t}
+
\rho(\{0\})\,\delta_0.
$$
Therefore, TBM and strong $\SP$-turb may hold even though $\SP$ fails, provided the cloud of initial conditions does not shrink sufficiently fast.

\subsection{Flow map expression for system \eqref{regamb}} \label{flowamb}
The interest of such a system is that we can integrate it explicitly on each
interval.
\subsubsection{Regularized flow map}

Let
$$
a_\epsilon=-\epsilon^2,
\qquad
b_\epsilon=\epsilon T_\epsilon-\epsilon^2,
\qquad
c_\epsilon=\epsilon T_\epsilon-2\epsilon^2.
$$
For $x\leq a_\epsilon$, set $\tau_-(x)=2(\sqrt{-x}-\epsilon)$, and for $x\in[a_\epsilon,b_\epsilon]$, set $\tau_0(x)=(b_\epsilon-x)/\epsilon$. Then the flow map $x\mapsto\phi_t^\epsilon(x)$ is given by
$$
\begin{array}{llcll}
x\leq a_\epsilon,
&
\phi_t^\epsilon(x)
&=&
\left\{
\begin{array}{l}
-\left(\sqrt{-x}-\dfrac{t}{2}\right)^2
\\[1em]
-\epsilon^2+\epsilon\left(t-\tau_-(x)\right)
\\[1em]
c_\epsilon+
\left(
\epsilon+\dfrac{t-\tau_-(x)-T_\epsilon}{2}
\right)^2
\end{array}
\right.
&
\begin{array}{l}
0\leq t\leq \tau_-(x),
\\[1em]
\tau_-(x)\leq t\leq \tau_-(x)+T_\epsilon,
\\[1em]
t\geq \tau_-(x)+T_\epsilon,
\end{array}
\\[4em]
x\in[a_\epsilon,b_\epsilon],
&
\phi_t^\epsilon(x)
&=&
\left\{
\begin{array}{l}
x+\epsilon t
\\[1em]
c_\epsilon+
\left(
\epsilon+\dfrac{t-\tau_0(x)}{2}
\right)^2
\end{array}
\right.
&
\begin{array}{l}
0\leq t\leq \tau_0(x),
\\[1em]
t\geq \tau_0(x),
\end{array}
\\[3em]
x>b_\epsilon,
&
\phi_t^\epsilon(x)
&=&
\left\{
\begin{array}{l}
c_\epsilon+
\left(
\sqrt{x-c_\epsilon}+\dfrac{t}{2}
\right)^2
\end{array}
\right.
&
\begin{array}{l}
t\geq0.
\end{array}
\end{array}
$$

\subsubsection{Inviscid flow map}
For the inviscid equation
$$
\dot x=\sqrt{|x|},
$$
the only obstruction to uniqueness is the point $x=0$. If $x<0$, the trajectory reaches this point at the time
$$
\tau_0(x)=2\sqrt{-x}.
$$
Before this time, the solution is uniquely determined; after it, the trajectory may either stay at $0$ for an arbitrary waiting time or leave along the positive branch. Thus the maximal single-valued part of the inviscid flow is
$$
\phi_t^0(x)
=
\begin{cases}
-\left(\sqrt{-x}-\dfrac{t}{2}\right)^2,
& x<0,\quad 0\leq t\leq \tau_0(x),
\\[1em]
\left(\sqrt{x}+\dfrac{t}{2}\right)^2,
& x>0,\quad t\geq0.
\end{cases}
$$
In particular, the endpoint is included:
$
\phi_{\tau_0(x)}^0(x)=0~\text{for } x<0.
$
For $x=0$, and for negative initial data after the hitting time, the inviscid problem is no longer single-valued. The graph is shown in Fig. \ref{Vec_Flow} (right panel).
\section{Numerics}\label{appendix_examples}
We give here various examples showing the great flexibility and diversity of behavior.

\subsection[dx/dt = x to the 1/3]{\texorpdfstring{$\dot{x}=x^{1/3}$}{dx/dt = x to the 1/3}}
\label{Ex1}

We consider the classical example
$$
\dot x=f(x),\qquad x(0)=x_0,
\qquad f(x)=x^{1/3}.
$$
The function $f$ has a non-Lipschitz singularity at $0$.
Following \cite{Drivas21}, we consider the regularized system $\dot x=f(x,\epsilon)$, where $f(\cdot,\epsilon)$ is defined by
\begin{equation}\label{toy1_3}
({\cal P}_\epsilon):~\left\{
\begin{array}{lllcccc}
\dot x & = & f(x) & |x| & > & \epsilon, \\
\dot x & = & \epsilon^{1/3}\,G(\frac{x}{\epsilon},\epsilon) & |x| & \leq & \epsilon,
\end{array}\right.
\end{equation}
with
$$
G(x,\epsilon)=\xi(x)\,f(x)+(1-\xi(x))\,\frac{\omega(\epsilon)+x}{2},
\qquad
\xi(x)=3x^2-2|x|^3.
$$
Since $({\cal P}_\epsilon)$ is Lipschitz, it is well-posed for every $\epsilon>0$. In particular, it guarantees that
$$
\gamma:s\mapsto \phi_t^s(x_0)
$$
is well defined.
We will use two different regularizations, denoted $\gamma_1$ and $\gamma_2$, with respective choices
$$
\omega_1(\epsilon)=\sin\frac{1}{\epsilon},
\qquad
\omega_2(\epsilon)=0.5+\sin\!\left(\frac{1}{\epsilon}+{\rm cst}\right).
$$
Note that these deterministic functions rapidly change sign when $\epsilon$ becomes small. They introduce a source of uncertainty that vanishes in the limit $\epsilon\to0$, thereby playing a role similar to that of a random variable.
Considering Definition~\ref{LSPOdef}, we use the observable ${\cal O}(x)=x$ and define
$$
\gamma(\epsilon):=\phi_{t=1}^\epsilon(0)
$$
for the two corresponding choices, each $\gamma$ is continuous in $\epsilon$, by the continuity of the regularized vector field w.r.t.\ $\epsilon$ and standard continuous dependence of solutions on parameters.
The result is shown in Fig.~\ref{Toyess}.

\begin{figure}[htbp]
\centerline{\includegraphics[scale=0.8]{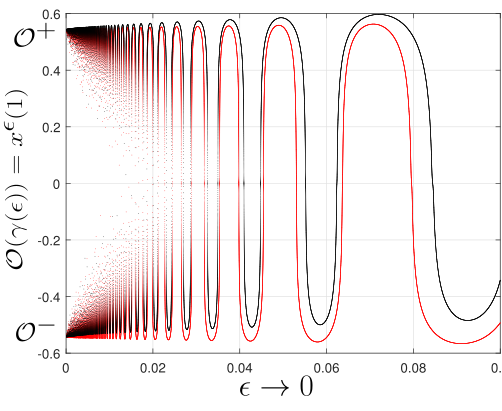}}
\caption{$\gamma(\epsilon)=\phi_1^\epsilon(0)$, solution of \eqref{toy1_3} for $\epsilon$ uniformly distributed in $[10^{-4},10^{-1}]$: in red $\gamma_1$ and in black $\gamma_2$. Both limiting measures concentrate on $\pm a$ with $a\approx0.544$.}
\label{Toyess}
\end{figure}

For the two regularizations $\gamma_1$ and $\gamma_2$, the limiting pushforward measures are different, although they have the same support: the first perturbation yields
$$
\frac{\delta_{+a}+\delta_{-a}}{2},
$$
whereas the second perturbation yields
$$
\theta\,\delta_{+a}+(1-\theta)\,\delta_{-a},
\qquad
\theta\approx0.7.
$$
To summarize: $({\cal P}_\epsilon)$ is strongly $\SP$ for $\gamma_1$, $\gamma_2$ and one must necessarily have
$x_0=0$ for any choice of $t>0$. When one takes
$\omega(\epsilon)={\rm cst}$, a single solution is selected in the limit (not shown).
We finally illustrate in Fig.~\ref{WSSP} the three possible manifestations of spontaneous stochasticity: strong, weak $\SP$ and
$\delta$-$LSP$.

\begin{figure}[htbp]
\centerline{\includegraphics[scale=0.12]{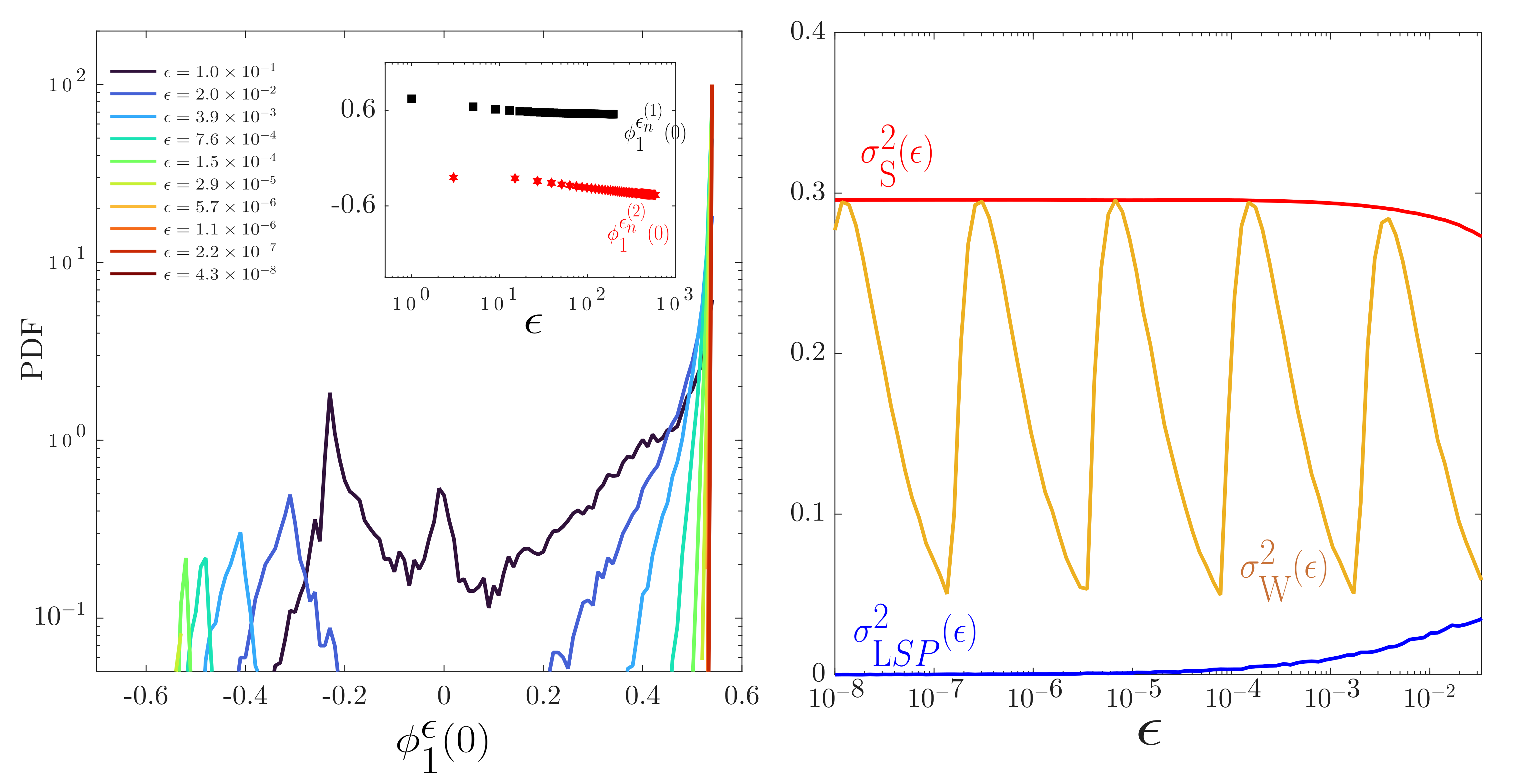}}
\caption{
Simulations of system~\eqref{toy1_3} for different choices of the function $\omega$. Specifically, we consider $\omega_{\rm S}(\epsilon)=\sin \tfrac{1}{\epsilon}$, $\omega_{\rm W}(\epsilon)=\sin \ln \tfrac{1}{\epsilon}$, and $\omega_\delta(\epsilon)=1-\epsilon+\sin \tfrac{1}{\epsilon}$ to illustrate, respectively, Strong-$\SP$, Weak-$\SP$, and $\delta$-$LSP$.
Left: Probability density of $x^\epsilon(1)$ in the $\omega_\delta$ case, for various values of $\epsilon$, where the regularization parameter is uniformly distributed over $[0,\epsilon]$. We observe clear convergence toward a Dirac mass. The inset displays two subsequences $x^{{\epsilon_n}^{(1,2)}}$ converging toward distinct values, thereby illustrating $LSP$. We used $\epsilon_n^{-1}=\pi n/2$, with subsequences corresponding to $\mathrm{mod}(n,4)=1$ and $\mathrm{mod}(n,4)=3$.
Right: Variance of $\phi_1^\epsilon(0)$ under $\mu_\epsilon$, namely
$
\operatorname{Var}({\mu_\epsilon})
:=
\int x^2\,\mu_\epsilon(dx)
-
\left(\int x\,\mu_\epsilon(dx)\right)^2.
$
For $\omega_{\rm S}$, the orange curve converges toward a positive value, consistent with Strong-$\SP$. In contrast, the yellow curve for $\omega_{\rm W}$ displays persistent oscillations, compatible only with Weak-$\SP$. Finally, the blue curve corresponding to $\omega_\delta$ converges to zero, as expected when $\mu_\epsilon\rightharpoonup\delta_a$; in this case, $\omega_\delta$ corresponds to $\delta$-$LSP$ and does not yield a genuine selection principle (see inset, right panel).}
\label{WSSP}
\end{figure}

\subsection{Two degrees of freedom with isolated singularity }\label{Ex3}

We consider a system with an isolated H\"older singularity at the origin.
Such singularities are believed to play an important role in shell models of turbulence \cite{Dombre1998,Maily2012}. A detailed measure-theoretic analysis can be found in \cite{Drivas21,Drivas24}. These systems take the form
$$
\dot x = |x|^\alpha F(y),
\qquad
y=\frac{x}{|x|}\in\mathbb S^{n-1},
\qquad
F:\mathbb S^{n-1}\to\mathbb R^n.
$$
Using the convention of \cite{Drivas21}, one can decompose $F$ into radial and tangential components:
$$
F(y)=F_r(y)\,y+F_s(y),
$$
with
$$
F_r:\mathbb S^{n-1}\to\mathbb R,
\qquad
F_s:\mathbb S^{n-1}\to T\mathbb S^{n-1}.
$$
In this case, the Dini directional derivatives at the singularity $x=0$ reduce to
$$
\Lambda^+_{z}(0,v)
=
\limsup_{t\to0^+} t^{\alpha-1}F_r(v)
\in\{-\infty,0,+\infty\},
\qquad
v\in\mathbb S^{n-1}.
$$
In order for Theorem~\ref{CN} to apply, one must necessarily have directions $v$ for which
$$
F_r(v)>0.
$$
This is exactly the \emph{defocusing} property considered in \cite{Drivas24}; see condition (b), p.~1859. Such singularities may be viewed as generalized saddles, including purely unstable critical points, and in favorable situations they possess both nontrivial stable and unstable sets.
We also allow the initial condition to lie on the singularity, as in the one-dimensional setting, but the stable set ${\cal W}^-(0)$ provides a more detailed characterization of all pre-blowup initial conditions.

The example presented here is a slight modification of the two-dimensional example studied in \cite{Drivas21}, but, unlike the original one, it does exhibit strong $\SP$. It takes the form
\begin{equation}\label{Ex2eq}
F(y)=y(y_1+y_2)-y^\perp y_1y_2^2,
\qquad
y^\perp=(-y_2,y_1).
\end{equation}
This example is nontrivial and has, in particular, homoclinic orbits with overlapping stable and unstable sets. Indeed, using the conventions of \cite{Drivas21,Drivas24}, the renormalized system
$$
\dot y=(F(y)\cdot y^\perp)\,y^\perp
\qquad\text{on }\mathbb S^1
$$
has a well-defined focusing attractor at $\varphi=\frac{3\pi}{2}$, which guarantees finite-time blowup. By contrast, the defocusing attractor at $\varphi=0$ on $\mathbb S^1$ is nonhyperbolic and one-sided, and carries the physical measure $\delta_{\varphi=0}$.
The distribution at a fixed time of the solutions, for a single pre-blowup initial condition
$$
x_0\in{\cal W}^-(0),
$$
is shown in Fig.~\ref{figex2}.

\begin{figure}[htbp]
\centerline{\includegraphics[scale=0.8]{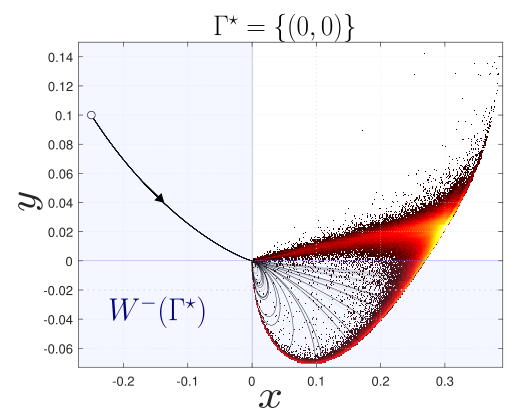}}
\caption{Distribution of the solutions of \eqref{Ex2eq} at $t=1.5$, starting from the initial condition
$x_0=(-0.25,0.1)$ (log-scaled colormap). The system has been regularized by additive noise of amplitude
$\sqrt{\epsilon}$ with $\epsilon\to0$. The solution for this initial condition reaches the singularity at the finite hitting time
$
t^\star\approx0.875.
$
The light-blue region corresponds to the stable set associated with initial conditions yielding $\SP$.}
\label{figex2}
\end{figure}

\subsection{$\SP$ without critical point}\label{Ex4}

We now discuss a simple example exhibiting $\SP$ although there are no critical points. It is in fact a one-dimensional nonautonomous system rewritten as an autonomous one in dimension two.
Let
$
s_{\alpha,\sigma}(x_2)
=
\operatorname{sign}(\sin 2\pi\sigma x_2)\,
|\sin 2\pi\sigma x_2|^\alpha,
$
and for $\alpha\in(0,1)$, $\mu\in\mathbb R$, consider
\begin{equation}\label{ex3}
\begin{array}{lll}
\dot x_1 & = & \mu,\\
\dot x_2 & = & s_{\alpha,\sigma}(x_2)\,\sin(2\pi x_1).
\end{array}
\end{equation}
Then, by a slight abuse of notation, one may write
\begin{align}
\Lambda^+_z(x,v)
&=
2\pi\,s_{\alpha,\sigma}(x_2)\cos(2\pi x_1)\,v_1v_2
\nonumber\\
&\quad
+
2\pi\sigma\alpha\,s_{\alpha-1,\sigma}(x_2)\cos(2\pi\sigma x_2)\sin(2\pi x_1)\,v_2^2,
\qquad
|v|=1.
\end{align}
Therefore, one can identify a nonempty singular set $\Gamma^\star$ corresponding to all
$x=(x_1,x_2)$ such that 
$\sin(2\pi\sigma x_2)=0,
$
that is,
$x_2=\frac{j}{2\sigma},~
j\in\mathbb Z,
$
with the additional sign condition
$
(-1)^j\sin(2\pi x_1)>0.
$
Equivalently, introducing the open intervals
$$
I_{k,1}=\Bigl(k,k+\frac12\Bigr),
\qquad
I_{k,2}=\Bigl(k+\frac12,k+1\Bigr),
$$
and the levels
$$
Y_{p,1}=\frac{p}{\sigma},
\qquad
Y_{p,2}=Y_{p,1}+\frac{1}{2\sigma},
$$
one obtains
$$
\Gamma^\star
=
\bigcup_{j,p\in\mathbb Z^2}
I_{j,1}\times\{Y_{p,1}\}
\;\cup\;
I_{j,2}\times\{Y_{p,2}\}.
$$
For all $x\in\Gamma^\star$, one has
$$
\Lambda^+_z(x,v)=+\infty
\qquad\text{provided } v_2\neq0.
$$

The stable set ${\cal W}^-(\Gamma^\star)$ is more difficult to characterize and depends crucially on $\mu$, but it includes $\Gamma^\star$ by definition. For the parameters used in Fig.~\ref{figex3}, numerical evidence suggests that
$$
{\cal W}^-(\Gamma^\star)=\mathbb R^2,
$$
so that all initial conditions reach the singular set in finite time. This in turn strongly suggests the occurrence of $\SP$, leading to the formation of space-filling beehive-like patterns.
The hexagonal cells are constrained by the set $\Gamma^\star$ in the $x_2$-direction and by $\mu$ in the $x_1$-direction.

These patterns result from a periodic alternation: pairs of trajectories merge into one or split into two, depending on the $x_1$-position and driven by alternating phases of contraction and expansion.
The number of trajectories reaching the axis $x_1=t$ from an initial condition at time $t_0$ appears to grow exponentially, of order
$
2^{\lfloor \lambda(\sigma,\mu)(t-t_0)\rfloor},
$
for some effective branching rate $\lambda(\sigma,\mu)>0$.
We do not display the distribution at fixed time; it can be guessed from Fig.~\ref{figex3} as a sum of Dirac masses with complicated weights.

\begin{figure}[htbp]
\centerline{\includegraphics[scale=0.8]{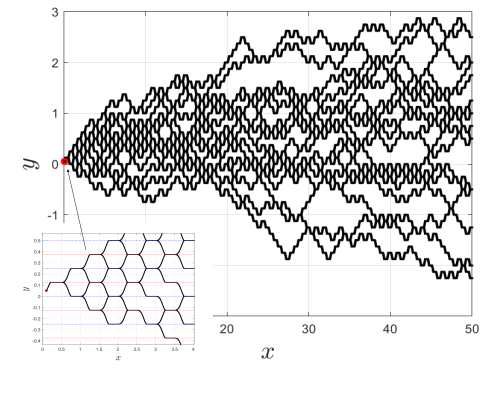}}
\caption{System \eqref{ex3} for $t\in[0,50]$, $\mu=1$, and $\sigma=4$, starting from the initial condition
$x_0=(0.1,0.05)$ (the system has been regularized by additive noise of amplitude
$\sqrt{\epsilon}$, with $\epsilon\to0$, and integrated for 20 different realizations).
The first hitting time of the singular set is
$
t^\star\approx0.5
$
(see zoom: the levels $x_2=Y_{p,1},Y_{p,2}$ of $\Gamma^\star$ are also shown as red and blue dash-dotted lines).}
\label{figex3}
\end{figure}

\section{Basic Notions from Measure Theory and Dynamical Systems}\label{PFTP}
\begin{enumerate}
    \item{\bf Pushforward}.
    \begin{definition}[Pushforward]
       Let $(X,{\cal A})$ and $(Y,{\cal B})$ be two measurable spaces. Let $\mu$ be a measure on $X$ and a measurable map
       $$
       f : X \to Y,~f^{-1}(B) \in {\cal A},~\forall B \in {\cal B}.
       $$
    \end{definition}
    The pushforward measure $f_\# \mu $ is defined as
    $$
    (f_\# \mu)(B) = \mu(f^{-1}(B)),~\forall B \in {\cal B}.
    $$
    One can also write, for all (bounded continuous) test functions $F \in C_b(Y;\mathbb{R})$:
    $$
    \langle f_\# \mu,F \rangle = \int_Y F(y) (f_\# \mu)(dy) = \int_X (F\circ f)(x) \mu(dx).
    $$
    In words, it is merely a change of variables where the statistics encoded by $\mu$ on $X$ are mapped to statistics on $Y$ through the map $f:X \to Y$. It is of interest to notice that for a semigroup $S_t:X \to Y$, the pushforward $(S_t)_\#$ has another well-known name: it is the \emph{Perron-Frobenius/transfer operator}.
    Namely, one has 
    $$\langle ({S}_t)_\# \mu,F \rangle = \int_X F(S_t x) \mu(dx) = \int_X (U_t F)(x) \mu(dx) = \langle \mu, U_t F \rangle = \langle (U_t)^\star \mu, F \rangle. $$ Here $U_t$ is the \emph{Koopman operator} acting on observables:
    $$
    U_t: C_b(X;\mathbb{R}) \to C_b(Y;\mathbb{R}): (U_t F)(x) = F(S_t x),
    $$
    and $(U_t)^\star$ is the dual Perron-Frobenius operator acting on measures. Although, we do not use these concepts, it is useful to recall that the infinitesimal generator of the Perron-Frobenius operator for $S_t$ the solution of some deterministic autonomous ODE/PDE $\dot x = f(x)$  is precisely the \emph{Liouville} operator ${\cal L} \mu = -\nabla \cdot (f \mu)$ acting on densities $\mu$. It is often written as $(S_t)_\#  = e^{t {\cal L}} $.

\item{\bf Tightness and Prokhorov's Theorem}. 
     Let us now defined tightness of a family of probability measures and state Prokhorov's theorem linking tightness and relative compactness
\begin{definition}[Tightness]
        Let $(X,T)$ be a Hausdorff space and let $\Sigma$ be a $\sigma$-algebra on $X$ that contains the topology $T$. Let $\lbrace\mu_\kappa \rbrace$ be a family of probability measures defined on $\Sigma$. The collection $\lbrace\mu_\kappa \rbrace$ is called tight if for any $\varepsilon>0$ there exists a compact set $K_\varepsilon \subset X$ such that for all $\kappa$, 
        $$ \mu_\kappa(K_\varepsilon) \geq 1-\varepsilon $$
\end{definition}
Tightness of collections of probability measure ensures relative compactness and therefore subsequential weak convergence as stated by Prokhorov's theorem. For simplicity, we restrict to \emph{Polish spaces} which are separable completely metrizable topological spaces. In particular, Banach spaces are Polish.
        
\begin{theorem}[Prokhorov's Theorem]
         Let $(S,\rho)$ be a separable and complete metric space and $\mathcal{P}(S)$ be the set of probability measures on $S$. There is a metric $d_0$ on $\mathcal{P}(S)$ equivalent to the topology of weak convergence. Moreover, a family of probability measures $\lbrace\mu_\kappa \rbrace\subset \mathcal{P}(S)$ is tight if and only if the closure of $\lbrace\mu_\kappa \rbrace$ in $(\mathcal{P}(S),d_0)$ is compact.
\end{theorem}
\item{\bf Portmanteau Theorem}
There are many different formulations, we give only two of them:
Let $(\mu_n)$ and $\mu$ be probability measures on a metric space $X$.
The Portmanteau theorem: $\mu_n \rightharpoonup \mu$ \emph{if and only if} any one of the following equivalent conditions holds.
(Closed-set form) For every closed $F\subset X$,
$$
\limsup_{n\to\infty} \mu_n(F) \le \mu(F).
$$
(Equivalently, open-set form) For every open $G\subset X$,
$$
\liminf_{n\to\infty} \mu_n(G) \ge \mu(G).
$$
The rule of thumb is that one uses the closed-set form to show $x \in \operatorname{supp}(\mu)$ and the open-set form to show that for $n$ large enough $x \in \operatorname{supp}(\mu_n)$.
\item{\bf Continuous mapping theorem}. 
Let $X,Y$ be two Polish spaces and $\{\mu_\epsilon\}_{\epsilon > 0}$ a family of probability measures on $X$ such that $\mu_\epsilon \rightharpoonup \mu \in {\cal P}(X)$ as $\epsilon \to 0$. Let $h:X\to Y$ be a continuous function, then
$$
h_\# \mu_\epsilon \rightharpoonup h_\# \mu \in {\cal P}(Y).
$$
\item {\bf Ergodic invariant measures}.
Let $(X, \mathcal{B}, \mu, (\mathcal{R}_t)_{t \in \mathbb{R}})$ be a measure-preserving flow, i.e., $(X, \mathcal{B}, \mu)$ is a probability space and each $\mathcal{R}_t : X \to X$ is measurable with $\mu(\mathcal{R}_t^{-1}B) = \mu(B)$ for all $B \in \mathcal{B}$ and $t \in \mathbb{R}$, or equivalently ${\cal R}_\# \mu = \mu$. A $\mathcal{R}_t$-invariant measure $\mu$ is \emph{ergodic} if every invariant set $A \in \mathcal{B}$ satisfies $\mu(A) \in \{0,1\}$. Equivalently, $\mu$ is ergodic if for all $F \in L^1(\mu)$, the time averages
$$
\frac{1}{T} \int_0^T F(\mathcal{R}_t x)\, dt \ \xrightarrow[T \to \infty]{} \ \int_X F(x)\, \mu(dx)~~\mu\text{-a.e.}
$$
Ergodicity means that the statistical behavior under the flow is indecomposable: $\mu$ cannot be expressed as a nontrivial convex combination of other invariant probability measures.
An interesting consequence is that if $F$ is ${\cal R}_t$-invariant:
$F({\cal R}_t x) = F(x)$ $\mu$-a.e. then from the above $F(x) = cst = \int F d\mu$ $\mu$-a.e.
\item {\bf Uniquely ergodic}.
A continuous flow $({\cal R}_t)_{t\in\mathbb{R}}$ on a compact metric space $X$ is \emph{uniquely ergodic} if there exists a single ${\cal R}_t$-invariant Borel probability measure on $X$.  
Equivalently, the space $\operatorname{Inv}({\cal R})$ of invariant probability measures is a singleton.  
In this case, time averages of every continuous function $F$ converge \emph{uniformly} to the same constant independently of $x \in X$.
\item {\bf Pointwise uniquely ergodic/generic}. A continuous flow $({\cal R}_t)_{t\in\mathbb{R}}$ on a compact metric space $X$ is \emph{pointwise uniquely ergodic} if there exists $\mu\in\operatorname{Inv}({\cal R})$ such that, for $\mu$-almost every $x\in X$ and every continuous $F$,
$
\frac{1}{T} \int_0^T F({\cal R}_t x)\,dt \ \xrightarrow[T\to\infty]{} \ \int_X F\,d\mu.
$
Such $x$ are called \emph{$\mu$-generic} or \emph{generic points} for $\mu$. Unlike unique ergodicity, the invariant measure need not be the only one; it is the unique limit for almost every starting point. 
\item {\bf Axiom A}.
Let $({\cal R}_t)_{t\in\mathbb{R}}$ be a $C^1$ flow on a compact manifold $M$.  We need first to define few important concepts.

\emph{Nonwandering set:} $\Omega({\cal R})=\{x\in M:\ \forall\,$neighborhoods $U\ni x$ and $\forall T>0,\ \exists t\ge T$ with ${\cal R}_t(U)\cap U\neq\varnothing\}$.  
Intuitively, these are points whose trajectories return arbitrarily close after arbitrarily long times.  
\emph{Topologically transitive:} an invariant set $\Lambda\subset M$ is transitive if for all nonempty open $U,V\subset \Lambda$ (relative topology) there exists $t\in\mathbb{R}$ with ${\cal R}_t(U)\cap V\neq\varnothing$ (equivalently, $\Lambda$ has a dense orbit).  
It means the dynamics can move from any region of $\Lambda$ to any other, allowing orbits to explore the whole set.  
\emph{Locally maximal:} an invariant set $\Lambda$ is locally maximal if there is a neighborhood $U$ of $\Lambda$ with $\Lambda=\bigcap_{t\in\mathbb{R}}{\cal R}_t(U)$.  
Such a set is isolated from the rest of the dynamics: all nearby orbits remain in $\Lambda$ for all time.  
\emph{Basic set:} a compact invariant set $\Lambda$ is basic if it is hyperbolic (in the flow sense), locally maximal, and topologically transitive.  
A basic set is an indecomposable, self-contained chaotic component of the dynamics. It is not necessarily a (local) attractor.

\emph{Axiom A (flow):} ${\cal R}_t$ satisfies Axiom~A if $\Omega({\cal R})$ is hyperbolic and periodic orbits are dense in $\Omega({\cal R})$.  
Under Axiom~A, $\Omega({\cal R})$ is the disjoint union of finitely many basic sets.

\item{\bf Morse-Smale systems}. 
A $C^1$ flow $({\cal R}_t)_{t\in\mathbb{R}}$ on a compact manifold $M$ is \emph{Morse–Smale} if its nonwandering set consists of finitely many hyperbolic equilibria and periodic orbits, and the stable and unstable manifolds of these intersect transversely.  
Morse–Smale systems form a special case of Axiom~A systems, with $\Omega({\cal R})$ finite and each basic set reduced to a single periodic orbit or equilibrium.

\end{enumerate}

\end{document}